\documentclass[
    11pt,
    letterpaper,
    reprint,
    nofootinbib,
    notitlepage,
    groupedaddress,
    eqsecnum,
    aps,
    physrev,
    twocolumn
]{revtex4-2}
\usepackage{chngcntr}
\usepackage{amsmath,amssymb,amsfonts} % standard AMS packages
\usepackage{empheq}
\usepackage{physics}

\usepackage{amsthm} % AMS theorem package

\newtheorem{theorem}{Theorem}
\newtheorem{lemma}{Lemma}
\newtheorem{proposition}{Proposition}
\newtheorem{corollary}{Corollary}
\usepackage{bbm}

\usepackage{subdepth}   % to equalise heights in sub/super-scripts

\usepackage{bm} % bold symbols in math mode \bm{...}
\renewcommand{\mathbf}{\bm}
\usepackage{dsfont} % proper mathbb format
\renewcommand{\mathbb}{\mathds} % redefine \mathbb
\usepackage{comment}
\usepackage[svgnames,dvipsnames]{xcolor}
\usepackage{graphicx}
\definecolor{NewBlue}{rgb}{0.1, 0.1, 0.7}
\definecolor{NewRed}{rgb}{0.7, 0.1, 0.1}
\usepackage[colorlinks,
    linkcolor=Blue,
    citecolor=NewBlue,
    urlcolor=cyan]{hyperref}

\usepackage{enumitem}

\usepackage{cleveref}

\theoremstyle{definition}
\newtheorem{definition}{Definition}[section]

\usepackage{array}
\newcolumntype{L}[1]{>{\raggedright\let\newline\\\arraybackslash\hspace{0pt}}m{#1}}
\newcolumntype{C}[1]{>{\centering\let\newline\\\arraybackslash\hspace{0pt}}m{#1}}
\newcolumntype{R}[1]{>{\raggedleft\let\newline\\\arraybackslash\hspace{0pt}}m{#1}}

\begin{document}

\title{Entanglement buffering with multiple quantum memories}

\date{\today}

\author{Álvaro G. Iñesta}
\thanks{These two authors contributed equally.\\
Email:
\href{mailto:a.gomezinesta@tudelft.nl}{a.gomezinesta@tudelft.nl},
\href{mailto:b.j.davies@tudelft.nl}{b.j.davies@tudelft.nl}.}

\author{Bethany Davies}
\thanks{These two authors contributed equally.\\
Email:
\href{mailto:a.gomezinesta@tudelft.nl}{a.gomezinesta@tudelft.nl},
\href{mailto:b.j.davies@tudelft.nl}{b.j.davies@tudelft.nl}.}

\author{Sounak Kar}

\author{Stephanie Wehner}

\affiliation{QuTech, Delft University of Technology, Lorentzweg 1, 2628 CJ Delft, The Netherlands}
\affiliation{EEMCS, Quantum Computer Science, Delft University of Technology, Mekelweg 4, 2628 CD Delft, The Netherlands}
\affiliation{Kavli Institute of Nanoscience, Delft University of Technology, Lorentzweg 1, 2628 CJ Delft, The Netherlands}

%%%%%%%%%%%%%%%%%%%%%%%%%%%%%%%%%%%%%%%%%%%%%%%%%%%
%%%%%%%%%%%%%%%%%%%% ABSTRACT %%%%%%%%%%%%%%%%%%%
%%%%%%%%%%%%%%%%%%%%%%%%%%%%%%%%%%%%%%%%%%%%%%%%%%%
\begin{abstract}
    Entanglement buffers are systems that maintain high-quality entanglement, ensuring it is readily available for consumption when needed.
    In this work, we study the performance of a two-node buffer, where each node has one long-lived quantum memory for storing entanglement and multiple short-lived memories for generating fresh entanglement.
    Newly generated entanglement may be used to purify the stored entanglement, which degrades over time. Stored entanglement may be removed due to failed purification or consumption.
    We derive analytical expressions for the system performance, which is measured using the entanglement availability and the average fidelity upon consumption. 
    Our solutions are computationally efficient to evaluate, and they provide fundamental bounds to the performance of purification-based entanglement buffers.
    We show that purification must be performed as frequently as possible to maximise the average fidelity of entanglement upon consumption, even if this often leads to the loss of high-quality entanglement due to purification failures.
    Moreover, we obtain heuristics for the design of good purification policies in practical systems. A key finding is that simple purification protocols, such as DEJMPS, often provide superior buffering performance compared to protocols that maximize output fidelity.
\end{abstract}

\maketitle

%%%%%%%%%%%%%%%%%%%%%%%%%%%%%%%%%%%%%%%%%%%%%%%%%%%
%%%%%%%%%%%%%%%%%%%% INTRO %%%%%%%%%%%%%%%%%%%
%%%%%%%%%%%%%%%%%%%%%%%%%%%%%%%%%%%%%%%%%%%%%%%%%%%
\section{Introduction}\label{sec.intro}
% PROBLEM: ENTANGLEMENT STORAGE
Entanglement is a fundamental resource for many quantum network applications, including some quantum key distribution protocols \cite{Ekert1991,Bennett1992}, distributed quantum sensing \cite{Lloyd2008,Qian2019,England2019,Wu2023}, and coordination tasks where communication is either prohibited or insufficiently fast~\cite{Brassard2005,Broadbent2008}.
%, blind quantum computing \cite{Leichtle2021}
%and quantum secret sharing \cite{BenOr2006} protocols.
Pre-distributing entanglement between remote parties would eliminate the need to generate and distribute entangled states on demand, saving time and resources \cite{Chakraborty2019,Ghaderibaneh2022,Pouryousef2022,inesta2023performance}. However, entanglement degrades over time due to decoherence, preventing long-term storage.

% SOLUTION: ENTANGLEMENT BUFFERING
Entanglement buffers are systems that store entanglement until it is needed for an application.
Passive buffers, which store entanglement in quantum memories, are constrained by the coherence time of these memories~\cite{Askarani2021}.
To overcome this limitation, purification-based entanglement buffers have been proposed~\cite{Davies2023a,Elsayed2023}.
These systems store entangled states and employ purification protocols to ensure the states remain high quality, mitigating the effects of decoherence.
Purification protocols take $m$ low-quality entangled states as input and produce $n$ higher-quality states as output, typically with $m>n$~\cite{Bennett1996, Deutsch1996, Dur1999, Yan2023}.
These protocols often involve some probability of failure, in which case all the input states are lost and no entanglement is produced.
Here, we focus on purification-based buffers.
%Purification-based entanglement buffers have been recently proposed to address this challenge \cite{Davies2023a,Elsayed2023}. These systems store entangled states and use purification subroutines to ensure the states remain high quality, alleviating the effects of decoherence.

As proposed in ref.~\cite{Davies2023a}, the performance of an entanglement buffer can be measured with two quantities: the availability (probability that entanglement is available for consumption when requested, see Definition~\ref{def.A}) and the average consumed fidelity (average quality of entanglement at the time of consumption, see Definition~\ref{def.F}). As well as having practical utility, entanglement buffers are a useful theoretical tool in order to understand the impact of several important interacting processes that occur in a quantum network: ongoing generation, purification, and consumption of entanglement. 
Of major interest is the impact of the entanglement purification protocol on the performance of the system. Since the success probability of entanglement purification typically depends on the fidelity of the input states, any rate and fidelity metrics are inherently coupled in systems making use of purification.
This coupling adds complexity to analytical calculations.
Consequently, most analytical studies on the performance of quantum networking systems exclude purification, and its impact on performance is typically explored with numerical methods~\cite{Haldar2024, Victora2023}. Nevertheless, as is a main result in this work, for entanglement buffering systems closed-form solutions are obtainable for a fully general purification protocol. One may then efficiently compute the performance of a particular purification policy, as well as make formal statements about how often purification should be applied to the buffered entanglement.

Here, we study the \emph{1G$n$B system}: a purification-based entanglement buffer with one good (long-lived) memory and $n$ bad (short-lived) memories.
The good memory can store entanglement, which can be consumed at any time by an application.
In contrast, bad memories can generate entanglement concurrently but cannot store it; they act as communication qubits.
For instance, carbon-13 nuclear spins in diamond can serve as good memories with coherence times up to $1$ min~\cite{Bradley2019}, while electron spins in nitrogen-vacancy centers may function as communication qubits, \mbox{with coherence times generally below $1$~s~\cite{Abobeih2018}}.

Each time entanglement is generated in some of the bad memories, the system may choose to immediately use it to purify the entanglement stored in the good memory. If purification is not attempted, the newly generated entanglement is discarded.
We illustrate the 1G$n$B system in Figure~\ref{fig.1GnB-sketch}.
Note that the physical platform must enable easy access to stored entanglement for consumption and purification.
However, network activities, such as repeated entanglement generation attempts and purification, may introduce additional noise, reducing memory lifetimes.
For example, in ref.~\cite{Pompili2021}, even when the carbon-13 nuclear spin used as a storage qubit is protected from network noise by applying stronger magnetic fields, it exhibits a shortened lifetime of approximately 11.6 ms.
%The 1G$n$B system allows for the implementation of complex entanglement purification protocols that can store entanglement for extended periods of time with a higher quality than before.

% Additionally, we lift some assumptions that were used in previous work, namely, we consider arbitrary and realistic purification protocols (e.g., in ref. \cite{Elsayed2023} a single purification protocol is considered; and in ref. \cite{Davies2023a}, the purification protocol is assumed to have a success probability that is independent of the quality of the buffered entanglement).

The 1G$n$B buffering system is a generalisation of the 1G1B system that was originally proposed in \cite{Davies2023a}. 1G1B is a system with only one good quantum memory and one bad memory. Here, we generalise the work in \cite{Davies2023a} in three main ways. 
Firstly, we now consider several ($n$) bad memories. Including several bad memories in our model now means that there is the possibility of generating multiple entangled links in the same entanglement generation attempt, for example via frequency~\cite{Wengerowsky2018,Askarani2021,Chen2023} or time multiplexing~\cite{Mower2011,Krutyanskiy2024}, which are commonly proposed ways of improving the rate of entanglement generation~\cite{Collins2007,Munro2010,Dam2017}.
Moreover, the simultaneous generation of multiple links opens up the use of stronger purification protocols, thereby providing an improvement to system fidelity metrics as well as the rate.
Note again that the physical implementation of the buffer must allow for such multiplexing and for purification of the generated entanglement.
The second generalisation from previous work is that we now model the system in discrete time rather than continuous time, which is more accurate to real-world systems, as entanglement generation typically happens in discrete attempts~(see e.g. refs.~\cite{Barrett2005, Togan2010, Bernien2013, Zhou2024}).
%Finally, our solutions apply to a fully arbitrary purification protocol, in contrast to previous work, which only applies to protocols that have a constant probability of success (in the fidelity of the buffered quantum state).
Finally, we now derive our solutions for a fully arbitrary purification protocol. In particular, the solutions for performance metrics presented in ref.~\cite{Davies2023a} only apply for purification protocols with a constant probability of success (i.e. the success probability must be independent of the fidelity of the buffered quantum state).
%In~\cite{Davies2023a}, these were used to provide bounds on the performance metrics for other protocols.
However, in this work, we remove this assumption and derive closed-form solutions for the availability and the average consumed fidelity of buffers that use arbitrary purification protocols. This is in contrast to \cite{Elsayed2023}, where although performance metrics are derived analytically and the probability of success is not necessarily constant, their computation requires solving a linear system of equations, which has dimension that scales with system parameters such as the memory lifetime.

% OUR SYSTEM
% Here, we study the \emph{1G$n$B system}: a entanglement buffer with one good memory and $n$ bad memories.
% The good memory can store entanglement, which can be consumed at any time by an application.
% The bad memories can be used to generate entanglement concurrently (i.e., they are used as communication qubits with frequency multiplexing).
% Every time entanglement is generated in some of the bad memories, we may immediately use it to purify the entanglement stored in the good memory -- otherwise, the newly generated entanglement is lost.
% Due to the presence of one good memory and $n$ bad memories, we call this buffering setup the \emph{1G$n$B system}.
% We illustrate the setup in Figure~\ref{fig.1GnB-sketch}.
% The 1G$n$B system allows for the implementation of complex entanglement purification protocols that can store entanglement for extended periods of time with a higher quality than before.
% Additionally, we lift some assumptions that were used in previous work, namely, we consider arbitrary and realistic purification protocols (e.g., in ref. \cite{Elsayed2023} a single purification protocol is considered; and in ref. \cite{Davies2023a}, the purification protocol is assumed to have a success probability that is independent of the quality of the buffered entanglement).

\begin{figure}[th]
    \centering
    \includegraphics[width=0.9\linewidth]{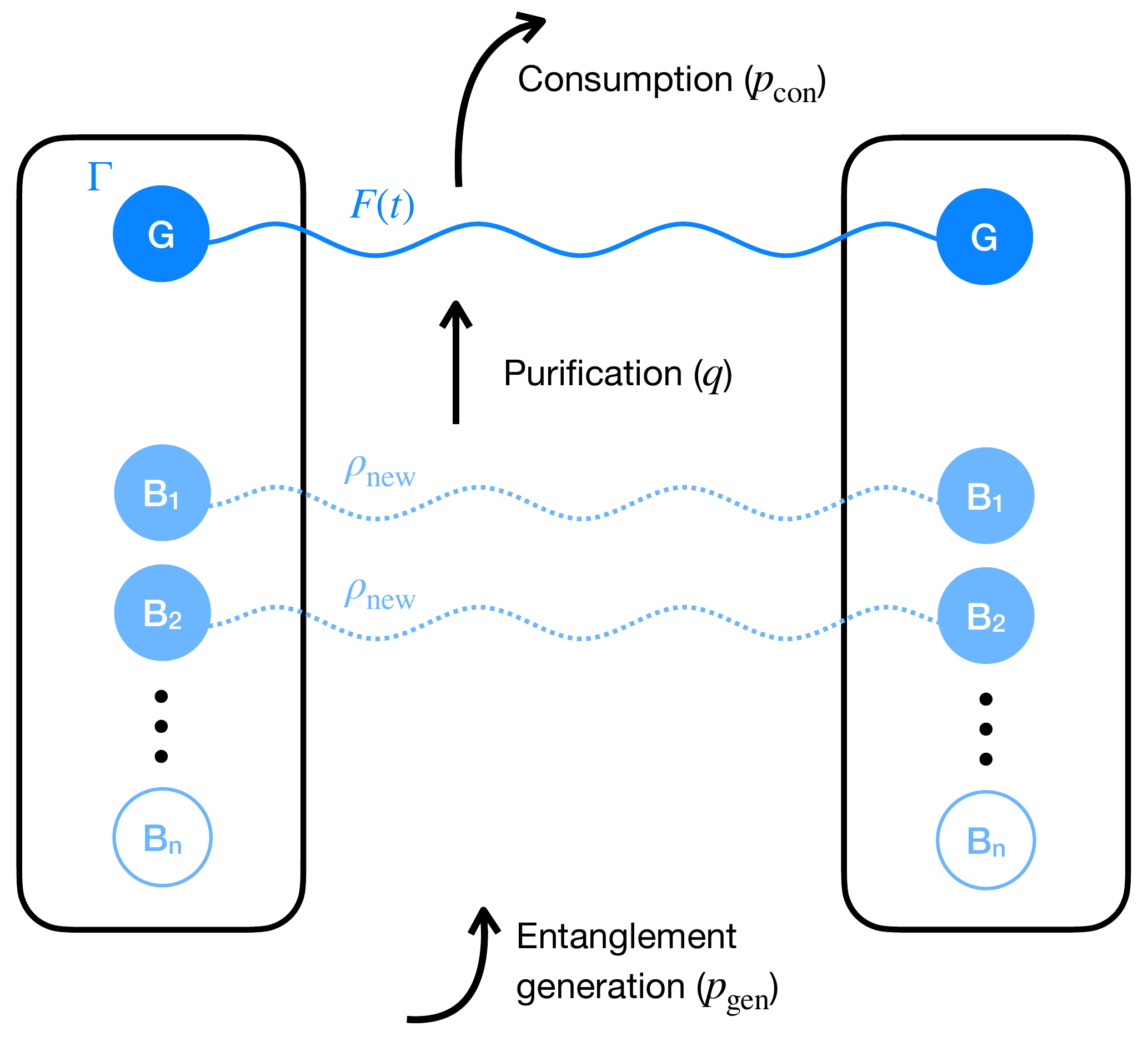}
    \caption{\textbf{Illustration of the 1G$n$B buffering system.}
    Entanglement generation is attempted in every bad memory (B$_1$, $\dots$, B$_n$) simultaneously in each time slot. Each memory succeeds with probability $p_\mathrm{gen}$.
    The good memory, G, stores entanglement, which decoheres at rate $\Gamma$.
    When G is full and new entanglement is generated in any of the B memories, a purification subroutine is applied with probability $q$.
    Entanglement is consumed from G with probability $p_\mathrm{con}$ in each time slot.
    }
    \label{fig.1GnB-sketch}
\end{figure}

% OUR WORK
In this letter, we firstly provide analytical expressions for the availability, $A$, and the average consumed fidelity, $\overline F$, of the 1G$n$B system (see model description in Section~\ref{sec.system}).
Then, we use these expressions to find fundamental limits to the performance of entanglement buffers.
Lastly, we investigate how the 1G$n$B system should be operated: because there is a large amount of freedom in the choice of purification protocols, it is not clear what purification strategies should be employed to maximise $A$ and $\overline F$.
For example, would it be beneficial to use a purification subroutine that provides a larger fidelity boost (which could increase $\overline F$) if this comes at the cost of a higher probability of failure (which means losing high-quality entanglement more frequently, decreasing $A$ and maybe also $\overline F$)?
Our main findings are the following:
\begin{itemize}
    \item {\sc Monotonic performance} --
    We show that, to maximise the average consumed fidelity, purification must be performed as much as possible, i.e. every time entanglement is generated in any of the bad memories. This holds even if the purification protocol has a large probability of failure.
    Nevertheless, there is a tradeoff between both performance metrics, since the availability decreases when purification is performed more frequently.
    \item {\sc Fundamental bounds} --- We provide upper and lower bounds for the availability and the average consumed fidelity of a 1G$n$B system, which constitute fundamental limits to the impact that a purification policy can have on the performance.
    %We show that these bounds apply to all buffers, regardless of the purification policy employed, and therefore constitute fundamental limits to the performance of entanglement buffers.
    \item {\sc Simple can be better than optimal} --- Simple purification protocols can greatly outperform advanced purification protocols that maximise the fidelity of the output entangled state. For example, we find that a buffering system using the 2-to-1 purification protocol from ref.~\cite{Deutsch1996} (known as DEJMPS) can outperform a system using the \mbox{$n$-to-1} optimal bilocal Clifford protocol from ref.~\cite{Jansen2022}, in terms of both availability and average consumed fidelity.
    %\item Purification protocols with a large probability of failure have a strong detrimental impact on the availability but they can still provide a large average consumed fidelity.
\end{itemize}

%%%%%%%%%%%%%%%%%%%%%%%%%%%%%%%%%%%%%
%%%%%%%%%%% 1GnB SYSTEM %%%%%%%%%%%%
%%%%%%%%%%%%%%%%%%%%%%%%%%%%%%%%%%%%%
\section{The 1G$n$B system}\label{sec.system}
In this section, we provide a short description of the entanglement buffering setup (see Figure~\ref{fig.1GnB-sketch}). The goal of the system is to buffer bipartite entanglement shared between two nodes. These nodes could be, for example, two end users in a quantum network or two processors in a quantum computing cluster.
We refer to bipartite entanglement as an \emph{entangled link} between the two nodes.
In the 1G$n$B system:
\begin{itemize}
    \item Each node has \emph{one long-lived memory} (good, G) and \emph{$n$ short-lived memories} (bad, B).
    
    \item The G memories are used to store the entangled link. We assume \emph{the link stored in memory is a Werner state} (any bipartite state can be transformed into a Werner state with the same fidelity by applying extra noise, a process known as \emph{twirling}~\cite{Bennett1996a,Horodecki1999}).
    Such a state can be parametrised with its fidelity to the target maximally entangled state, $F$.
    
    \item The entangled link stored in G is subject to \emph{depolarising noise} with memory lifetime~$1/\Gamma$, which causes an exponential decay in fidelity with rate $\Gamma$.
    That is, if the link in memory has an initial fidelity $F$, after time $t$ this reduces to
    \begin{equation}
        F \mapsto \left(F-\frac{1}{4}\right) e^{-\Gamma t} + \frac{1}{4}.
        \label{eq.fidelity_decay}
    \end{equation}

    \item Before each entanglement generation attempt, the system checks if a new \emph{consumption request} has arrived. The arrival of a new consumption request in each time step occurs with probability $p_\mathrm{con}$. If there is a link stored in memory G when a consumption request arrives, the link is immediately consumed and therefore removed from the memory. This takes up the entire time step. If there is no link available, the request is discarded and the system proceeds with the entanglement generation attempt.
    
    \item The B memories are used to generate new entangled links. In the literature, these are usually called communication or broker qubits~\cite{Benjamin2006}. This communication qubit can be, for example, the electron spin in a nitrogen-vacancy center~\cite{Bernien2013,Rozpedek2019,Lee2022}.
    Every time step that is not taken up by consumption, \emph{entanglement generation} is attempted in all $n$ bad memories simultaneously, e.g. using frequency or spatial multiplexing, and each of them independently generates an entangled link with probability $p_\mathrm{gen}$.
    This means that, after each multiplexed attempt, the number of successfully generated links follows a binomial distribution with parameters $(n,p_{\mathrm{gen}})$.
    Each of these new links is of the form $\rho_\mathrm{new}$, which is an arbitrary state that depends on the entanglement generation protocol employed (see e.g. refs. \cite{Barrett2005,Campbell2008,Togan2010,Jones2016}).

    \item When $k\geq1$ entangled links are generated in the B memories and the G memory is empty, one of the links is \textit{transferred} to the G memory. If the G memory is occupied, the new links may be used to \mbox{\emph{purify}} the link in memory. The system decides to attempt purification with probability $q$. If the system does not decide to purify, the new links are discarded. If the system decides to attempt purification and this succeeds, then the resultant link in the G memory is twirled, converting it into the form of a Werner state with the same fidelity.
    %\memento{Justify twirling, maybe in Appendix A.}
\end{itemize}

%In Appendix~\ref{app.system}, we provide a more detailed description of the 1G$n$B system.
Table \ref{tab.variables} summarises all variables of the system.
Next, we discuss how to model the purification strategy.

\subsection{Purification policy}
The main degree of freedom in the 1G$n$B system is the choice of purification protocol. This is given by the purification policy.
\begin{definition}
    The \emph{purification policy} $\pi$ is a function that indicates the purification protocol that must be used when $k$ links are generated in the B memories,
    \begin{equation}
        \pi: k \in \{1, \dots,n\} \mapsto \pi(k) \in \mathcal{P}_{k+1},
    \end{equation}
    where $\mathcal{P}_{m}$ is the set of all $m$-to-1 purification protocols.
\end{definition}

%Note that, when $k$ new links are generated, a \mbox{($k+1$)-to-1} protocol can be implemented, since we must also consider the buffered link.

    Protocol $\pi(k)$ of purification policy $\pi$ is the \mbox{$(k+1)$-to-1} purification protocol that is used when $k$ new links are generated in the B memories (examples of basic protocols can be found in refs.~\cite{Deutsch1996,Bennett1996,Dehaene2003}; see ref.~\cite{Dur2007} for a survey). The purification protocol updates the fidelity of the buffered link from $F$ to $J_k(F)$, where
    \begin{equation}\label{eq.Jk_def}
            J_k(F) = \frac{1}{4} + \frac{a_k(\rho_{\mathrm{new}})\left(F-\frac{1}{4}\right)+b_k(\rho_{\mathrm{new}})}{c_k(\rho_{\mathrm{new}})\left(F-\frac{1}{4}\right)+d_k(\rho_{\mathrm{new}})}.
    \end{equation}
    We call $J_k$ the \emph{jump function of protocol $\pi(k)$}.
    The protocol succeeds with probability
    \begin{equation}\label{eq.pk_def}
        p_k(F) = c_k(\rho_{\mathrm{new}}) \left(F-\frac{1}{4}\right) + d_k(\rho_{\mathrm{new}}),
    \end{equation}
    otherwise all of the links (including the buffered one) are discarded and the G memory becomes empty.
In Appendix B of ref. \cite{Davies2023a}, the forms (\ref{eq.Jk_def}) and (\ref{eq.pk_def}) for the output fidelity and success probability are justified, given that the buffered link is a Werner state with fidelity $F$ and any other input state is given by the same arbitrary density matrix $\rho_\mathrm{new}$.
We therefore see that the action of any purification protocol on the fidelity of the buffered link is determined by the four parameters $a_k(\rho_\mathrm{new})$, $b_k(\rho_\mathrm{new})$, $c_k(\rho_\mathrm{new})$, $d_k(\rho_\mathrm{new})$.
In Appendix~\ref{app.abcd}, we discuss the values that these coefficients can take. As an example, we also provide the explicit form of these coefficients for the well-known 2-to-1 DEJMPS protocol\cite{Deutsch1996}.
    
Lastly, note that purification policy $\pi$ employs protocol $\pi(k)$ when $k$ new links are generated. However, this does not mean that all the new links are used in the protocol. For example, a policy may simply replace the link in memory with a newly generated link and ignore the rest of the new links.

\subsection{Fidelity of the buffered entanglement}
Given the system description, we now view 1G$n$B as a discrete-time stochastic process. In particular, at time $t$ the state of the system is the fidelity $F(t)$ of the buffered link, as this is the only quantity that can change over time. If there is no link in the buffered memory at time $t$, we let $F(t)=0$. This is for notational convenience, as recalling the decoherence (\ref{eq.fidelity_decay}), one can never reach zero fidelity if there is a link present.

We now outline the characteristic behaviours of $F(t)$ when moving from time $t$ to time $t+1$. 

If $F(t)=0$, then if entanglement generation is unsuccessful, in the next time step the fidelity will remain at that value if entanglement generation was unsuccessful. If entanglement generation is successful, in the next time step the fidelity will be $F_{\mathrm{new}}$, where $F_{\mathrm{new}} = \bra{\Phi_{00}}\rho_{\mathrm{new}}\ket{\Phi_{00}}$ is the fidelity of freshly generated links. We will assume that $F_{\mathrm{new}}>1/4$.

If $F(t)>0$, then in the next time step this could evolve in one of the following ways: ($i$) if no purification is attempted then the fidelity simply decoheres by one unit of time according to (\ref{eq.fidelity_decay}); ($ii$) if $k$ new links are generated and purification is successfully performed, the fidelity decoheres by one time step and is then mapped according to the corresponding jump function (\ref{eq.Jk_def}); ($iii$) if a consumption request has arrived or if purification fails, the link is removed and the system becomes empty.

In Figure~\ref{fig.fidelity-vs-time-sketch}, we illustrate an example of how the fidelity may evolve. 
%We see that the fidelity at time $t$ depends on the time between each purification round. 
%In \cite{Davies2023a}, these times had a simplified distribution since the probability of successful purification was assumed to be constant, thereby decoupling the evolution of $F(t)$ with the values it took previously. However, now we are working with the general model for purification, where the probability of successful purification (\ref{eq.pk_def}) is a (linear) function of the fidelity. 

In the following subsection, we define the two performance metrics: the availability and the average consumed fidelity. We then present simple closed-form solutions for these two performance metrics in the 1G$n$B system.

\begin{figure}[th]
    \centering
    \includegraphics[width=\linewidth]{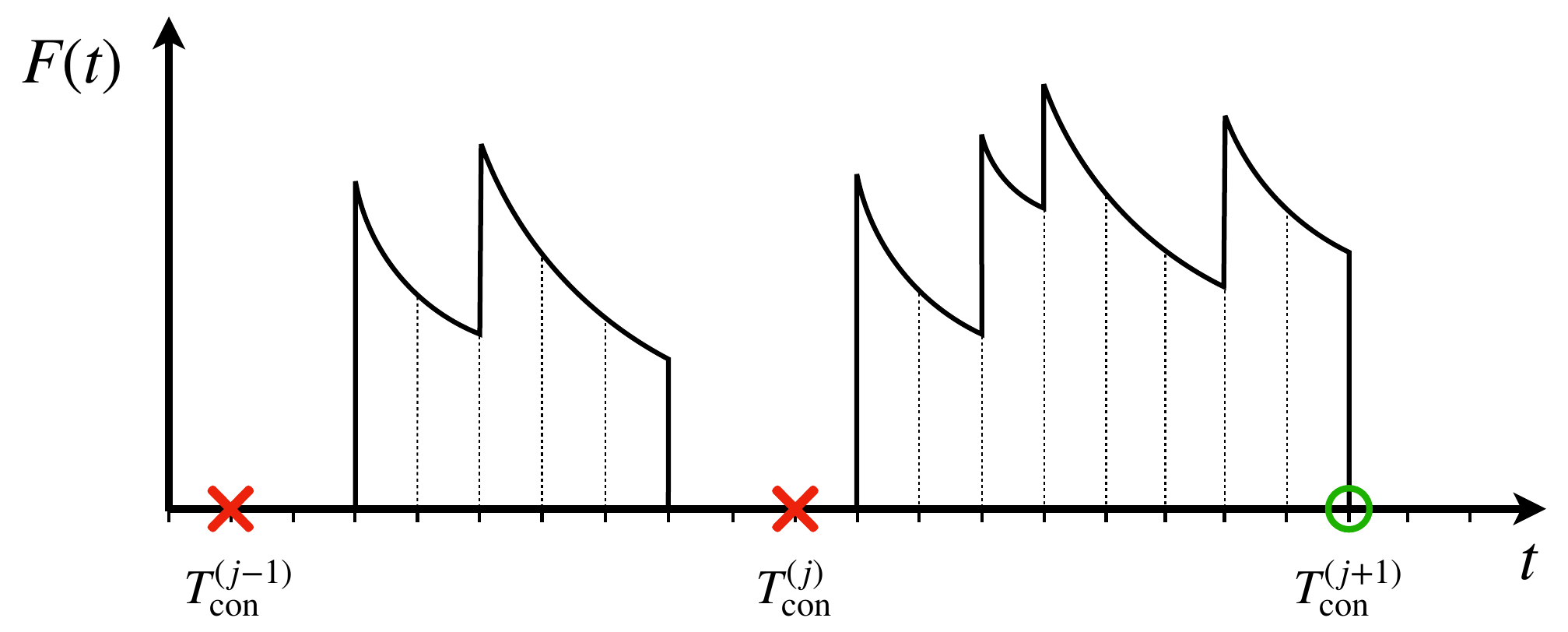}
    \caption{\textbf{Example dynamics of the 1G$n$B system.} Here, the fidelity $F(t)$ of the link in the G memory is plotted against time. The vertical lines represent discretisation of time. The jumps in fidelity occur when the link is purified successfully. In between purifications, the link is subject to decoherence and the fidelity decreases. The link in the G memory is removed due to either failed purification or consumption. When there is no link in memory, $F(t) = 0$. The $j$-th consumption request arrives at time $T_{\mathrm{con}}^{(j)}$. The green tick (red crosses) represent when a consumption request is (is not) served. 
    }
    \label{fig.fidelity-vs-time-sketch}
\end{figure}

%%%%%%%%%%%%%
%%%%%%%%%%%%%
\renewcommand{\arraystretch}{1.2}
\begin{table}[t]
    \centering
    \caption{Parameters of the 1G1B system. See main text for further details.
    }\label{tab.variables}
    \vspace{-2mm} % Adjust the height of the space between caption and tabular
\begin{tabular}{p{0.1\linewidth}p{0.75\linewidth}}
\multicolumn{2}{c}{\textbf{Hardware}}\\[1pt]
\hline
    $n$ & Number of short-lived memories\\
	$p_\mathrm{gen}$ & Probability of successful entanglement generation attempt\\
    $\rho_\mathrm{new}$ & Bipartite entangled state produced after a successful entanglement generation\\
	$\Gamma$ & Rate of decoherence\\[5pt]
\multicolumn{2}{c}{\textbf{Application}}\\[1pt]
\hline
    $p_\mathrm{con}$ & Probability of consumption request\\[5pt]
\multicolumn{2}{c}{\textbf{Purification policy}}\\[1pt]
\hline
    $q$ & Probability of attempting purification immediately after a successful entanglement generation attempt (otherwise the new links are discarded)\\
    $J_k(F)$ & Jump function. Given a buffered link with fidelity $F$, $J_k(F)$ is the fidelity immediately following a successful purification using $k$ newly generated links. Rational function with coefficients $a_k,b_k,c_k,d_k$ -- see (\ref{eq.Jk_def}).\\
    $p_k(F)$ & Probability of successful purification using $k$ newly generated links. Linear function with coefficients $c_k,d_k$ -- see (\ref{eq.pk_def}).
\end{tabular}
\end{table}
%%%%%%%%%%%%%
%%%%%%%%%%%%%

%%%%%%%%%%%%%%%%%%%%%%%%%%%%%%%%%%%%%
%%%%%%%%%%% PERFORMANCE %%%%%%%%%%%%
%%%%%%%%%%%%%%%%%%%%%%%%%%%%%%%%%%%%%
\subsection{Buffering performance}\label{subsec.performance}
The first step towards the design of useful entanglement buffers is to determine a suitable way to measure performance.
Here, we define two performance metrics for entanglement buffers -- these quantities were proposed in ref. \cite{Davies2023a}, where they were used to study the 1G1B system.
Then, we provide exact, closed-form expressions for these two performance metrics in the 1G$n$B system.

%In ref. \cite{Davies2023a}, the authors proposed two quantities to measure the quality of service of an entanglement buffering system can be measured with two quantities: the availability and the average consumed fidelity. Next, we define both quantities and provide analytic expressions to compute them in a 1G$n$B system.

Our first metric is the \textit{availability}. A user is able to consume entanglement only when there is a link available in memory~G at the time of requesting the entanglement.
Therefore, an important performance measure is the probability that entanglement is available when a consumption request arrives. 

\begin{definition}[Availability]\label{def.A}
The availability $A$ is the probability that there is an entangled link present in memory~G when a consumption request arrives.
This is defined as
\begin{equation}
    A = \lim_{m\rightarrow \infty}\frac{1}{m} \sum_{j = 1}^{m} \mathbbm{1}_{\text{link exists}}(T_{\mathrm{con}}^{(j)}),
\end{equation}
where $T_\mathrm{con}^{(j)}$ is the arrival time of the $j$-th consumption request, and $\mathbbm{1}_{\text{link exists}}(t)$ is and indicator function that takes the values one if there is a link stored in memory G at time $t$, and zero otherwise.
\end{definition}

The availability may be seen as a rate metric: it determines the rate at which entanglement can be consumed.
The second performance metric is the \textit{average consumed fidelity}, which captures the average quality of consumed entanglement.

\begin{definition}[Average consumed fidelity]\label{def.F}
The average consumed fidelity is the average fidelity of the entangled link upon consumption, conditional on a link being present. More specifically,
\begin{equation}
    \overline{F} = \lim_{m\rightarrow \infty} \frac{\sum_{j = 1}^{m} \mathbbm{1}_{\text{link exists}}\left(T_{\mathrm{con}}^{(j)}\right)\cdot F^-(T_{\mathrm{con}}^{(j)})}{\sum_{j = 1}^{m} \mathbbm{1}_{\text{link exists}}(T_{\mathrm{con}}^{(j)})},
    \label{eq.Fbar-def}
\end{equation}
where 
\begin{equation}
    F^-(t) =  \begin{cases}
        e^{-\Gamma}\left(F(t-1)-\frac{1}{4}\right) + \frac{1}{4}, \text{ if } F(t-1)>0, \\ 
        0, \text{ if } F(t-1) = 0.
    \end{cases}
    \label{eq.def-F(t^-)}
\end{equation}
is the fidelity of the link stored in memory G at the end of the previous timestep at time $t-1$ (and therefore consumed at time $t$), and $T_\mathrm{con}^{(j)}$ is the arrival time of the $j$-th consumption request.
\end{definition}
The indicator function in the numerator of (\ref{eq.Fbar-def}) is included for clarity, but is not necessary: if there is no link in memory at time $t$, then $F(t)=0$ by definition. 

We note that the Definitions \ref{def.A} and \ref{def.F} are presented differently to how they were in ref.~\cite{Davies2023a}. This is because the new definitions have a clearer operational meaning, as they are from the viewpoint of the consumer. However, in Appendix~\ref{app.viewpoint} we show that these metrics are equivalent for the 1G$n$B system. 

As our first main result, we derive analytical solutions for the availability and the average consumed fidelity in the 1G$n$B system

\begin{theorem}[Formula for the availability]\label{theorem.A}
The availability of the 1G$n$B system is given by
\begin{equation}\label{eq.A_1GnB}
    A = \frac{\mathbb{E}[T_\mathrm{occ}]}{\mathbb{E}[T_{\mathrm{gen}}]+\mathbb{E}[T_\mathrm{occ}]}\text{ a.s.}
\end{equation}
where $T_{\mathrm{gen}}$ is the time to generate new entangled links and $T_\mathrm{occ}$ is the time from when the G memory becomes occupied until it is emptied due to consumption or to failed purification. The expected values are given by
\begin{equation}\label{eq.Tgen}
    \mathbb{E}[T_{\mathrm{gen}}] = \frac{1}{1- \left(1-p_{\mathrm{gen}} \right)^n}
\end{equation}
and 
\begin{equation}\label{eq.Tocc}
    \mathbb{E}[T_\mathrm{occ}] = 
        \frac{1-\tilde{A}+\tilde{C} (F_\mathrm{new}-\frac{1}{4})}
        {\big[(1-\tilde{A})(1-\tilde{D}) - \tilde{B}\tilde{C} \big] \tilde{P} },
\end{equation}
with
\begin{equation*}
\begin{split}
    \tilde{P} &\coloneqq p_{\mathrm{con}}+ q  \left(1- \left(1-p_{\mathrm{gen}} \right)^n\right) (1-p_{\mathrm{con}}),\\
    \tilde{A} &\coloneqq \frac{q  (1-p_{\mathrm{con}}) \tilde{a}}{e^{\Gamma} -   \left(1- q+ q\left(1-p_{\mathrm{gen}} \right)^n\right) (1-p_{\mathrm{con}})},\\
    \tilde{B} &\coloneqq \frac{ q (1-p_{\mathrm{con}})\tilde{b}}{p_{\mathrm{con}} + q\left(1-\left(1-p_{\mathrm{gen}} \right)^n\right) (1-p_{\mathrm{con}})}, \\
    \tilde{C} &\coloneqq \frac{q  (1-p_{\mathrm{con}})\tilde{c}}{e^{\Gamma} -   \left(1- q+ q\left(1-p_{\mathrm{gen}} \right)^n\right) (1-p_{\mathrm{con}})},\\
    \tilde{D} &\coloneqq \frac{ q (1-p_{\mathrm{con}})\tilde{d}}{p_{\mathrm{con}} + q\left(1-\left(1-p_{\mathrm{gen}} \right)^n\right) (1-p_{\mathrm{con}})},
\end{split}
\end{equation*}
and 
\begin{align*}
    \tilde{a} &\coloneqq \sum_{k=1}^n a_k \cdot \binom{n}{k}(1-p_{\mathrm{gen}})^{n-k}p_{\mathrm{gen}}^k, \\
    \tilde{b} &\coloneqq \sum_{k=1}^n b_k \cdot \binom{n}{k}(1-p_{\mathrm{gen}})^{n-k}p_{\mathrm{gen}}^k, \\
    \tilde{c} &\coloneqq \sum_{k=1}^n c_k \cdot \binom{n}{k}(1-p_{\mathrm{gen}})^{n-k}p_{\mathrm{gen}}^k, \\
    \tilde{d} &\coloneqq \sum_{k=1}^n d_k \cdot \binom{n}{k}(1-p_{\mathrm{gen}})^{n-k}p_{\mathrm{gen}}^k.
\end{align*}
\end{theorem}
\begin{proof}
    See Appendix~\ref{app.solutions}.
\end{proof}

From Theorem \ref{theorem.A}, we see that the availability depends on all the parameters of the system (listed in Table~\ref{tab.variables}), including the noise level $\Gamma$. 
The latter may come as a surprise, since one would expect noise to have an impact on the average consumed fidelity but maybe not on the availability, which is only affected by processes that fill or deplete the G memory. These processes are entanglement generation, failed purification, and consumption.
In our model, the probability of failed purification depends via (\ref{eq.pk_def}) on the fidelity of the buffered link, which is in turn affected by the level of noise. As a consequence, noise has an indirect effect on the availability.

\begin{theorem}[Formula for the average consumed fidelity]\label{theorem.F}
The average consumed fidelity of the 1G$n$B system is given by
\begin{equation}\label{eq.F}
    \overline{F} = \frac{\tilde w F_\mathrm{new} + \tilde x}
    {\tilde y F_\mathrm{new}  + \tilde z} \text{ a.s.}
\end{equation}
with
\begin{align*}
    \tilde w &\coloneqq p_\mathrm{con} + q (1-p_\mathrm{con}) \left( p_\mathrm{gen}^* + \frac{1}{4} \tilde c - \tilde d \right),\\
    \tilde x &\coloneqq \frac{1}{4} \left[ e^\Gamma -1  + q(1-p_\mathrm{con}) \left(-\tilde a + 4\tilde b - \frac{1}{4} \tilde c + \tilde d \right) \right],\\
    \tilde y &\coloneqq q (1-p_\mathrm{con}) \tilde c,\\
    \tilde z &\coloneqq e^\Gamma -1 + p_\mathrm{con} + q (1-p_\mathrm{con}) \left(p_\mathrm{gen}^* - \tilde a - \frac{1}{4} \tilde c \right),
\end{align*}
where $p_\mathrm{gen}^* = 1 - (1-p_{\mathrm{gen}})^n$, and $\tilde a$, $\tilde b$, $\tilde c$, and $\tilde d$ are given in Theorem \ref{theorem.A}. 
\end{theorem}
\begin{proof}
    See Appendix~\ref{app.solutions}.
\end{proof}

We note that both $A$ and $\overline F$ have been defined as random variables in Definitions~\ref{def.A} and \ref{def.F}. However, as shown in Theorems~\ref{theorem.A} and \ref{theorem.F}, these quantities are almost surely deterministic functions of the system parameters.
For clarity and convenience, we will adopt a slight abuse of notation and treat $A$ and $\overline F$ as deterministic functions.
This convention will be maintained throughout the remainder of the text.

%%%%%%%%%%%%%%%%%%%%%%%%%%%%%%%%%%%%%
%%%%%%%%%%% RESULTS %%%%%%%%%%%%
%%%%%%%%%%%%%%%%%%%%%%%%%%%%%%%%%%%%%
\section{Buffering system design}\label{sec.results}
In this section, we discuss our main findings after analysing the performance of the 1G$n$B system.
In Subsection~\ref{subsec.monotonic}, we study the impact of a general purification protocol on the system performance. In particular, it is shown that the availability and the average consumed fidelity are monotonic in the parameter $q$ that determines how frequently the system attempts purification.
In the remaining subsections, we investigate how the choice of purification policy impacts the performance of the buffering system, and we provide heuristic rules for the design of a good purification policy.

\subsection{Monotonic performance}\label{subsec.monotonic}
Each time a B memory successfully generates entanglement, there is the opportunity to purify the buffered link.
This is controlled by the parameter $q$, which is the probability that, after some fresh links are successfully generated, they are used to attempt purification (otherwise they are discarded).
If purification is never attempted ($q=0$), the fidelity of the buffered link will never be increased, although the buffered link will never be lost to failed purification.
If purification is always attempted ($q=1$), the availability and average consumed fidelity might be affected as follows:
\begin{itemize}
    \item Purifying more often means risking the loss of buffered entanglement more frequently, since purification can fail. This suggests availability may be decreasing in $q$.
    However, many purification protocols have a probability of success that is increasing in the fidelity of the buffered link, $F$. This means that, when purification is applied more frequently to maintain a high-fidelity link, subsequent purification attempts are more likely to succeed. Consequently, it is not clear that the availability is decreasing in $q$.
    \item The fidelity of the buffered link increases after applying several purification rounds. However, if purification is applied too greedily, we may lose a high-quality link and we would have to restart the system with a lower-quality link. If a consumption request then arrives, it would only be able to consume low-quality entanglement. Hence, it is not clear that the average consumed fidelity is increasing in $q$.
\end{itemize}
In the following, we address the previous discussion and show that, if purification is always attempted ($q=1$), the availability is actually minimised, while the average consumed fidelity is maximised.
%Despite the above intuition, we will now present some results that shows it to be incorrect.
%In particular,
More generally, we show that $A$ and $\overline{F}$ are both monotonic in $q$, given some reasonable conditions on the jump functions $J_k$.
The following results (Propositions \ref{prop.dAdq} and \ref{prop.dFdq}) may be used to answer an important question about the 1G$n$B system: \emph{how frequently should we purify the buffered state in order to maximise $A$ (or $\overline{F}$)?} 
That is, \emph{what value of $q$ optimises our performance metrics?}

\begin{proposition}\label{prop.dAdq}
    The availability is a non-increasing function of $q$, i.e.
    \begin{equation}
        \frac{\partial A}{\partial q} \leq 0.
    \end{equation}
\end{proposition}
\begin{proof}
    See Appendix~\ref{app.dAdq}.
\end{proof}

As previously explained, the monotonicity of the availability in $q$ is not a trivial result, and it has fundamental implications.
It allows us to derive upper and lower bounds that apply to 1G$n$B systems using \emph{any} purification policy.
\begin{corollary}
The availability is bounded as
\begin{equation}\label{eq.A_bounds}
    \frac{p_\mathrm{gen}^*\cdot(\gamma + p_\mathrm{con})}{
    \xi + \xi' \cdot p_\mathrm{gen}^* + \xi'' \cdot (p_\mathrm{gen}^*)^2}
    \leq A \leq \frac{p_\mathrm{gen}^*}{p_\mathrm{gen}^*+p_\mathrm{con}},
\end{equation}
with $p_\mathrm{gen}^*\coloneqq1-(1-p_\mathrm{gen})^n$, $\gamma \coloneqq e^\Gamma - 1$, $\xi \coloneqq \gamma p_\mathrm{con} + p_\mathrm{con}^2$, $\xi' \coloneqq 1 + 2\gamma + (2-\gamma) p_\mathrm{con} - 2 p_\mathrm{con}^2 $, and $\xi'' \coloneqq 2 (1-p_\mathrm{con})^2$. Moreover, the upper bound is tight, and for any purification policy is achieved when $q=0$.
\end{corollary}
\begin{proof}
    See Appendix~\ref{app.dAdq}.
\end{proof}

We refer to $p_\mathrm{gen}^*$ as the \emph{effective generation probability}, since it is the probability that at least one new link is generated in a single (multiplexed) attempt.

The upper bound from (\ref{eq.A_bounds}) only depends on the effective generation probability and the probability of consumption.
This bound is achievable with any purification policy: to maximise the availability, it suffices to never purify ($q=0$). A special case are deterministic policies (those with $p_k(F)=1$, $\forall k$), which achieve this bound for any $q$.
This upper bound coincides with the tight upper bound found in previous work for a 1G1B system \cite{Davies2023a}.
%In that work, the authors found that the maximum availability was $\lambda/(\lambda+\mu)$, where $\lambda$ was the (non-multiplexed) entanglement generation rate and $\mu$ was the consumption rate.
Note that the 1G1B analysis from ref. \cite{Davies2023a} was done in continuous time, where rates were used instead of probabilities. In this framework, the maximum availability was $\lambda/(\lambda+\mu)$, where $\lambda$ was the (non-multiplexed) entanglement generation rate and $\mu$ was the consumption rate.

Unlike the upper bound, we note that the lower bound from (\ref{eq.A_bounds}) has not yet shown to be tight.
We believe that the availability at $q=1$ of a policy that always fails purification ($c_k=d_k=0$, $\forall k$) constitutes a tight lower bound for any other purification policy.
We leave this proof as future work.

Figure \ref{fig.availability_bounds} shows the upper and lower bounds for the availability from (\ref{eq.A_bounds}) versus $p_\mathrm{gen}^*$ for two different noise levels.
As discussed, only the lower bound is affected by noise. In particular, we have observed that the gap between the bounds is reduced when the noise level increases.
Another remarkable feature is that, when $p_\mathrm{gen}^*$ approaches zero, both upper and lower bounds are equal to $p_\mathrm{gen}^*/p_\mathrm{con}$ to first order in $p_{\mathrm{gen}}$. Hence, in the limit of small effective generation probabilities, the availability also satisfies
\begin{equation}
    A \approx \frac{p_\mathrm{gen}^*}{p_\mathrm{con}}.
\end{equation}
%This can also be observed in the example from Figure~\ref{fig.availability_bounds}.

\begin{figure}[t!]
    \centering
    \includegraphics[width=0.9\linewidth]{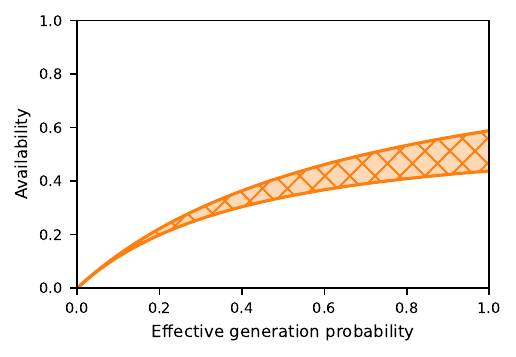}
    \caption{\textbf{The upper bound on the availability is tight and it converges to the lower bound in the limit of small generation probabilities.}
    Upper and lower bounds on the availability from (\ref{eq.A_bounds}), versus the effective generation probability $p_\mathrm{gen}^*=1-(1-p_\mathrm{gen})^n$.
    The availability can only take values within the shaded region.
    In this example we use $\Gamma=1$ and $p_\mathrm{con}=0.7$.
    }
    \label{fig.availability_bounds}
\end{figure}

%We now discuss the behavior of the average consumed fidelity.
\begin{proposition}\label{prop.dFdq}
    The average consumed fidelity is a non-decreasing function of $q$, i.e.,
    \begin{equation}
        \frac{\partial \overline F}{\partial q} \geq 0,
    \end{equation}
    if $J_k(F_{\mathrm{new}})\geq F_{\mathrm{new}}$, $\forall k\in\mathbb{N}$.
    %if $J_k(H_\mathrm{new})\geq H_\mathrm{new}$, $\forall k\in\mathbb{N}$.
\end{proposition}
\begin{proof}
    See Appendix~\ref{app.dFdq}.
\end{proof}

As previously explained, the monotonicity of $\overline F$ in $q$ is not a trivial result. In fact, this behaviour is only certain for purification policies composed of protocols that can increase the fidelity of a newly generated link. That is, when $k$ new links are generated, the protocol applied satisfies $J_k(F_{\mathrm{new}})\geq F_{\mathrm{new}}$.
This is a reasonable condition: otherwise, we would be applying purification protocols that decrease the fidelity of new links.

Proposition \ref{prop.dFdq} also allows us to derive useful upper and lower bounds for $\overline{F}$ that apply to 1G$n$B systems using any purification policy.
%\begin{equation}\label{eq.Fbounds}
    %\overline F_0 \leq \overline F \leq \overline F_0 + \Delta F,
%\end{equation}
\begin{corollary}
The average consumed fidelity is bounded as
\begin{equation}\label{eq.Fbounds}
    \frac{\gamma + 4F_\mathrm{new} p_\mathrm{con}}{4\gamma + 4p_\mathrm{con}}
    \leq \overline F
    \leq \frac{\gamma + 4 F_\mathrm{new} p_\mathrm{con} + 3 (1-p_\mathrm{con}) p_\mathrm{gen}^*}{4\gamma + 4p_\mathrm{con}},
\end{equation}
with $\gamma\coloneqq e^\Gamma -1$. Moreover, the lower bound is tight, and for any purification policy is achieved when $q=0$.
\end{corollary}
\begin{proof}
    See Appendix~\ref{app.dFdq}.
\end{proof}
%where $\overline F_0 \coloneqq (\gamma/4 + F_\mathrm{new} p_\mathrm{con})/(\gamma + p_\mathrm{con})$, $\gamma\coloneqq e^\Gamma -1$ and \mbox{$\Delta F \coloneqq (3/4) (1-p_\mathrm{con}) p_\mathrm{gen}^* / (\gamma + p_\mathrm{con})$}.

We see that that the tight lower bound from (\ref{eq.Fbounds}) does not depend on the number of memories $n$, the probability of successful entanglement generation $p_\mathrm{gen}$, or the purification policy.
This is because this bound corresponds to $q=0$. In such a case, no purification is applied, and the consumed fidelity only depends on the initial fidelity ($F_\mathrm{new}$) and the amount of decoherence experienced until consumption (given by $\Gamma$ and $p_\mathrm{con}$).

The bounds on $\overline F$ can be used to determine if the parameters of the system need an improvement to meet specific quality-of-service requirements. For example, let us consider Figure~\ref{fig.F_bounds}, which shows the bounds for $p_\mathrm{con}=0.7$ and two different values of $\Gamma$.
If noise is strong ($\Gamma=1$ in this example), we observe that values of $p_\mathrm{gen}^*$ below 0.5 yield $\overline F<1/2$, %which does not allow the consumption of entanglement (since the average consumed fidelity would be below $1/2$,
which means that, on average, the consumed link will not be entangled~\cite{Davies2023a}.
%of fidelities above $0.5$. Note that, for Bell-diagonal states, if the fidelity is below 0.5 there is no entanglement \cite{Davies2023a}.
Hence, if the consumption request rate is $p_\mathrm{con}=0.7$, we need to increase $p_\mathrm{gen}^*$ beyond 0.5 (by increasing the number of B memories, $n$, or the probability of successful entanglement generation, $p_\mathrm{gen}$) or to decrease the noise experienced in memory~G in order to provide a useful average state.
When the noise level is $\Gamma = 0.1$, Figure~\ref{fig.F_bounds} shows that $\overline F > 0.85$. Moreover, for $p_\mathrm{gen}^* > 0.3$, the upper bound is above $F_\mathrm{new}$, which means that a smart choice of purification policy may allow us to buffer entanglement with $\overline F > F_\mathrm{new}$.
Ultimately, this means that, in this regime, an entanglement buffer with faulty memories may be able to keep entanglement at higher fidelities than a perfect memory.

\begin{figure}[t!]
    \centering
    \includegraphics[width=0.9\linewidth]{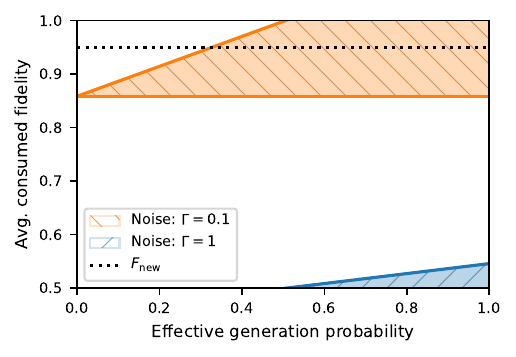}
    \caption{\textbf{The upper bound on the average consumed fidelity marks unachievable values for any purification policy.}
    Upper and lower bounds on the average consumed fidelity $\overline F$ from (\ref{eq.Fbounds}), versus the effective generation probability $p_\mathrm{gen}^*=1-(1-p_\mathrm{gen})^n$.
    $\overline F$ can only take values within the shaded region.
    In this example we use $p_\mathrm{con}=0.7$. %(there is a $70\%$ probability that a new consumption request arrived before each entanglement generation attempt)
    }
    \label{fig.F_bounds}
\end{figure}

\section{Choosing a purification policy}\label{sec.choosingpolicy}
In previous studies of entanglement buffering, the choice of purification policy was restricted by the properties of the system. For example, in ref. \cite{Davies2023a} the 1G1B system was studied, where only 2-to-1 purification protocols can be implemented, and the jump function was assumed to be linear in the fidelity of the buffered link. Other works include simplifying assumptions (e.g. in ref. \cite{Elsayed2023}, a buffer is studied that employs the purification protocol proposed in ref. \cite{Ruan2021}).
The 1G$n$B buffering system offers more freedom in the choice of purification protocols.
In a 1G$n$B buffer, each entanglement generation attempt is multiplexed and can generate up to $n$ new links at a time.
When $k \leq n$ new links are produced, any \mbox{$(k+1)$-to-$1$} purification protocol can in principle be implemented.
This provides an extra knob that can be used to tune the performance of the system to the desired values.
In this section, we investigate the impact that specific purification policies have on the system and we provide guidelines on how to choose a suitable purification policy.
Note that an exhaustive optimisation problem would be extremely computationally expensive to solve due to the large space of purification policies -- optimising over $a_k$, $b_k$, $c_k$, $d_k$ is not easy, since it is not certain that every combination of those parameters corresponds to an implementable purification circuit.

\subsection{Simple policies: identity, replacement, and concatenation}
% TRIVIAL: IDENTITY AND REPLACEMENT
There are two trivial deterministic policies ($p_k=1$, $\forall k$) that we will use as a baseline:
\begin{itemize}
    \item In the \emph{identity policy}, the system does not perform any operation on the buffered link, which yields an output fidelity $J_k(F) = F$, $\forall k>0$. This is equivalent to setting $q=0$. As discussed in Section~\ref{subsec.monotonic}, the identity policy therefore maximises the availability and minimises the average consumed fidelity.
    \item In the \emph{replacement policy}, the system replaces the buffered entangled link by a new link, yielding an output fidelity $J_k(F) = F_\mathrm{new}$, $\forall k>0$. This corresponds to $a_k=0$, $b_k=F_\mathrm{new}-1/4$, $c_k=0$, and $d_k=1$. Since this policy is deterministic, from the discussion in Section~\ref{subsec.monotonic} we find that the replacement policy also provides maximum availability for any value of $q$. Since $\overline F$ is maximized for $q=1$ (Proposition \ref{prop.dFdq}), we will only consider a replacement policy that always chooses to replace the link in memory when a new link is generated. That is, the replacement policy implicitly assumes $q=1$.
\end{itemize}

% DEJMPS
Another simple strategy is the \emph{DEJMPS policy}.
This policy consists in applying the well-known 2-to-1 DEJMPS purification protocol \cite{Deutsch1996} using the buffered link and a newly generated link as inputs. If more than one link is successfully generated, we use only one of them and discard the rest. We provide the purification coefficients $a_k$, $b_k$, $c_k$, and $d_k$ for this policy in Appendix~\ref{app.subsec.DEJMPS}.
% CONCATENATION
One of the main drawbacks of the DEJMPS policy is that it does not take full advantage of the multiplexed entanglement generation, as it only uses one of the newly generated links and discards the rest.
A technique that could improve the performance of the policy is \emph{concatenation}, which consists in applying DEJMPS to all links (the buffered one and the newly generated ones) consecutively until only one link remains, which will be stored in memory G.
Note that the concatenation of DEJMPS subroutines can be applied using different orders of the links (see Figure~\ref{fig.concatenation-sketch}). The order determines the output fidelity and probability of success \cite{VanMeter2008}, which affects the performance of the buffering system.
In what follows, we consider the \emph{concatenated DEJMPS} policy, where DEJMPS is applied sequentially to all the newly generated links and the buffered link is used in the last application of DEJMPS, as in Figure~\ref{fig.concatenation-sketch}a.
In our analysis, we found that different orderings provided qualitatively similar behaviour of our two performance metrics (see Appendix~\ref{app.subsec.orderings} for further details).

\begin{figure}[t!]
    \centering
    \includegraphics[width=0.9\linewidth]{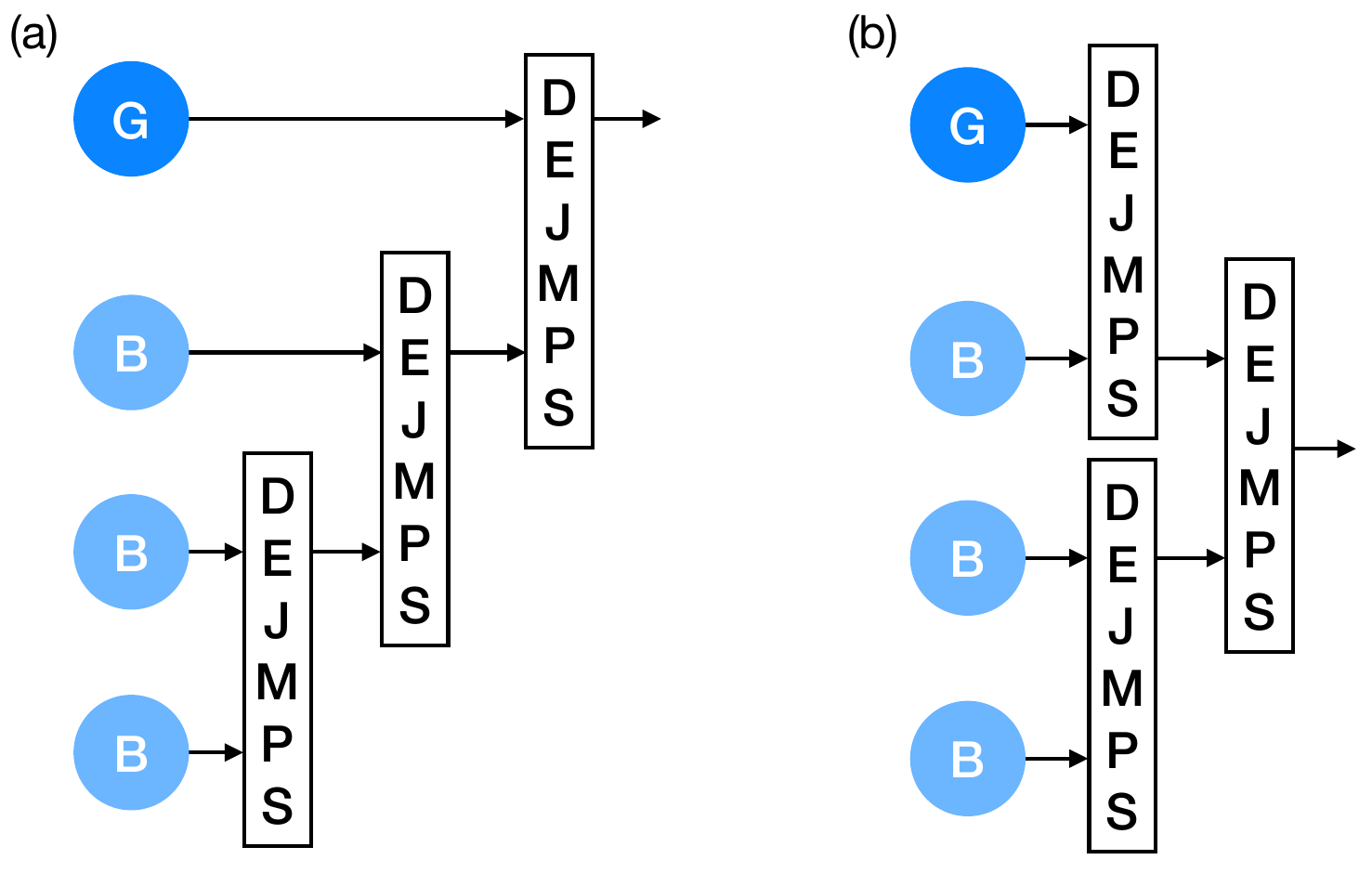}
    \caption{\textbf{The ordering in a concatenated policy matters.}
    Example of two different orderings when the buffered link (G) and three newly generated links (B) are used. We call ordering (a) ``concatenated DEJMPS''. Ordering (b) is often called ``nested''~\cite{Briegel1998}.
    }
    \label{fig.concatenation-sketch}
\end{figure}

Figure \ref{fig.AFplot_concatenation} shows the performance of several policies: identity, replacement, DEJMPS, and concatenated DEJMPS $\times N$. The latter is a policy that applies DEJMPS sequentially up to $N$ times and discards any extra links: if $k\leq N$ links are generated then $k$ concatenations are performed, and if $k>N$ links are generated, $N$ concatenations are performed. We note that concatenated DEJMPS $\times 1$ is just the same as the DEJMPS policy.
DEJMPS and concatenated DEJMPS are plotted for $q\in[0,1]$. The maximum average consumed fidelity is indicated with a dot, and it is achieved when $q=1$.
The first observation from this figure is that  a higher level of concatenation decreases the availability. This is because it requires multiple DEJMPS subroutines to succeed, which decreases the overall probability of successful purification. However, a higher level of concatenation can significantly increase the average consumed fidelity $\overline F$.
For example, the maximum $\overline F$ that DEJMPS can achieve is 0.915, while concatenated DEJMPS $\times 2$ leads to $\overline F = 0.937$ (for $q=1$).
Nevertheless, for the parameter values explored, we also find that increasing the number of concatenations beyond two often reduces both $A$ and $\overline F$.
This behaviour is shown more explicitly in Figure~\ref{fig.excessive_concatenation}, where we plot the maximum $\overline F$ versus the maximum number of concatenations $N$. In this example, the number of B memories is $n=10$, and therefore it is only possible to perform up to 10 concatenated applications of DEJMPS.
We observe that $\overline F$ is maximized for two concatenations. The same was observed for different parameter values -- in some edge cases, $\overline F$ increases with more concatenations, although the increase is marginal (see Appendix~\ref{app.subsec.increasingconcats} for further details).
In conclusion, this result shows that even if many new links are successfully generated in parallel, it can sometimes be beneficial to use only one or two of them for purification while discarding the rest.
% However, even when excess links are discarded and not used for purification, it is still beneficial to use a larger number of B memories $n$. This is because a higher level of multiplexing (i.e. larger $n$) improves the effective entanglement generation probability $p^*_{\mathrm{gen}}$, ultimately leading to an improvement in the availability $A$.

\begin{figure}[t!]
    \centering
    \includegraphics[width=0.9\linewidth]{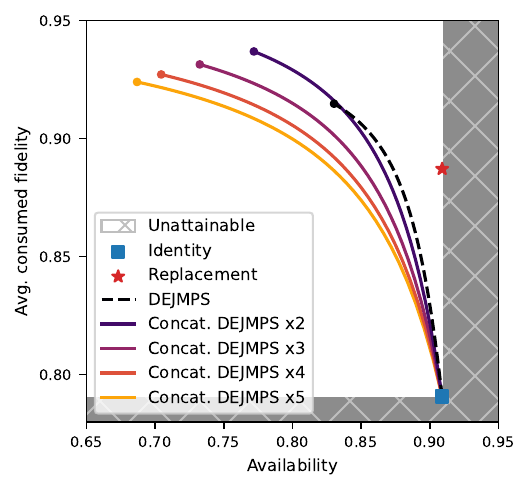}
    \caption{\textbf{Concatenating simple purification policies decreases $A$ but may increase $\overline F$.}
    Performance of 1G$n$B systems with different purification policies, in terms of availability $A$ and average consumed fidelity $\overline F$. The shaded area corresponds to unattainable values of $A$ and $\overline F$ (see (\ref{eq.A_bounds}) and (\ref{eq.Fbounds})).
    Lines and markers show the combinations of $A$ and $\overline F$ achievable by different purification policies: identity (square marker), replacement (star marker),
    DEJMPS (dashed line), and concatenated DEJMPS (solid lines).
    Concatenation can boost $\overline F$ (e.g. the maximum $\overline F$ of twice-concatenated DEJMPS is larger than DEJMPS), but excessive concatenation may eventually lead to a drop in $\overline F$.
    Parameter values used in this example: $n=10$, $p_\mathrm{gen}=0.5$, $\rho_\mathrm{new}$ is a Werner state with $F_\mathrm{new}=0.9$, $p_\mathrm{con}=0.1$, and $\Gamma=0.02$.
    }
    \label{fig.AFplot_concatenation}
\end{figure}

\begin{figure}[t!]
    \centering
    \includegraphics[width=0.9\linewidth]{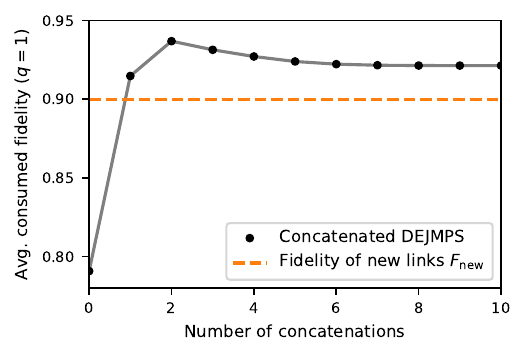}
    \caption{\textbf{Excessive concatenation worsens the performance.}
    Maximum average consumed fidelity $\overline F$ achieved by a purification policy that concatenates DEJMPS a limited number of times.
    Zero concatenations corresponds to an identity policy (no purification is performed). One concatenation corresponds to the DEJMPS policy.
    Excessive concatenation may decrease $\overline F$.
    Parameter values used in this example: $n=10$, $p_\mathrm{gen}=0.5$, $\rho_\mathrm{new}$ is a Werner state with $F_\mathrm{new}=0.9$, $p_\mathrm{con}=0.1$, and $\Gamma=0.02$.
    }
    \label{fig.excessive_concatenation}
\end{figure}

\subsection{Simple policies can outperform complex policies}
In the previous section, we found that implementing a simple 2-to-1 protocol, even when several links are generated in the B memories, can provide a better performance than using all of the newly generated links for purification with concatenated 2-to-1 protocols.
A follow-up question arises: \emph{what if we employ more sophisticated $(k+1)$-to-1 protocols instead of simply concatenating 2-to-1 protocols?} \emph{Can we then improve the performance of the buffer?}
This is the question that we explore now.

Much recent work has focused on the search for optimal purification protocols \cite{Rozpedek2018a,Krastanov2019,Jansen2022}, where optimal protocols are typically defined as those which maximise the output fidelity, or in some cases the success probability.
Here, we evaluate the performance of a 1G$n$B system with some of these protocols, and we find a surprising result: simple protocols like DEJMPS can vastly outperform these more complex protocols in terms of buffering performance.
In particular, we consider the bilocal Clifford protocols that maximise the output fidelity, given in ref.~\cite{Jansen2022}. We refer to this policy as the \emph{optimal bilocal Clifford (optimal-bC) policy}. In Appendix~\ref{app.subsec.optimalbiCliff}, we discuss the details of this policy and provide its purification coefficients $a_k$, $b_k$, $c_k$, and $d_k$.

Figure \ref{fig.simple_vs_optimal} shows the performance of the optimal-bC policy in comparison to DEJMPS and twice-concatenated DEJMPS.
The optimal-bC policy provides a significantly lower availability, $A$, without providing any advantage in average consumed fidelity, $\overline F$.
In other words, for any desired $A$, using DEJMPS or twice-concatenated DEJMPS always provides a larger $\overline F$ than the optimal-bC policy. If we want to increase $A$ as much as possible, the replacement policy is better than any other, as discussed earlier.
We say that the performance of DEJMPS, twice-concatenated DEJMPS and replacement forms the \emph{Pareto frontier}~\cite{Marler2004}, which informally is the set of best achievable values for $A$ and $\overline{F}$ for this collection of protocols.
%Due to the inherent tradeoff between $A$ and $\overline F$, it is convenient to look at the \emph{Pareto frontier} to compare the buffering performance when using different policies.
We tested different parameter combinations and found that the Pareto frontier was often made of DEJMPS, concatenated DEJMPS and replacement.
The reason for these simple policies to outperform the optimal-bC policy is that the optimal bilocal Clifford protocols maximise the output fidelity at the expense of a reduced probability of success. At some point, the sacrifice in the probability of success can outweigh the benefit of a larger output fidelity, thereby reducing the overall performance of the buffer in terms of both $A$ and $\overline F$.

\begin{figure}[t!]
    \centering
    \includegraphics[width=0.9\linewidth]{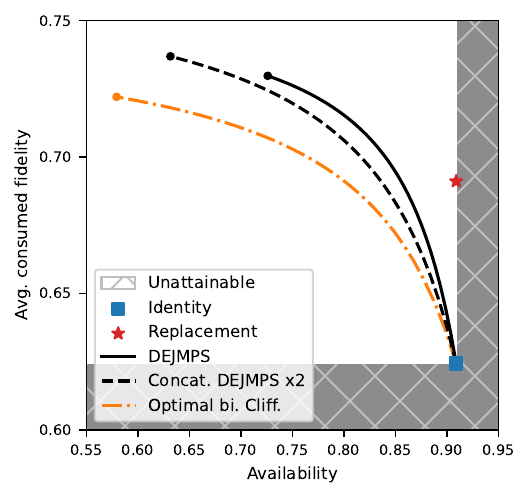}
    \caption{\textbf{Simple policies perform better despite discarding freshly generated entanglement.}
    Performance of 1G$n$B systems with different purification policies, in terms of availability $A$ and average consumed fidelity $\overline F$. The shaded area corresponds to unattainable values of $A$ and $\overline F$ (see (\ref{eq.A_bounds}) and (\ref{eq.Fbounds})).
    Lines and markers show the combinations of $A$ and $\overline F$ achievable by different purification policies: identity (square marker), replacement (star marker),
    DEJMPS (solid line), twice-concatenated DEJMPS (dashed line), and optimal-bC (dotted line).
    Parameter values used in this example: $n=5$, $p_\mathrm{gen}=0.8$, $\rho_\mathrm{new}$ is a Werner state with $F_\mathrm{new}=0.7$, $p_\mathrm{con}=0.1$, and $\Gamma=0.02$.
    }
    \label{fig.simple_vs_optimal}
\end{figure}

Our comparison between simple and optimal purification protocols is by no means an exhaustive study. However, it shows that purification protocols that maximise only the output fidelity (or probability of success) must not be blindly used in more complex systems involving many impacting factors such as decoherence and consumption, such as entanglement buffers.
In fact, we find that discarding some of the newly generated links and applying a 2-to-1 protocol can provide larger $A$ and $\overline F$ than using all of the links in a more sophisticated purification subroutine.
Note that this does not mean that multiplexed entanglement generation is not useful: even if we only employ 2-to-1 protocols, multiplexing boosts the effective entanglement generation rate, which allows for a more frequent purification of the buffered link.

Additionally, we also tested other complex policies that use (suboptimal) \mbox{$k$-to-1} protocols, such as the \emph{513 EC policy}, which uses a \mbox{5-to-1} protocol based on a $[[5,1,3]]$ quantum error correcting code.
In Appendix~\ref{app.subsec.513policy}, we explain this policy in detail and show that it can outperform DEJMPS and \mbox{twice-concatenated DEJMPS} in some parameter regions.

\subsection{Flags can improve performance}
As discussed in the previous sections, concatenating protocols multiple times does not necessarily improve the performance of the buffer (neither in terms of $A$ nor $\overline F$).
The reason is that, when concatenating, a single failure in one of the purification subroutines (in our examples, DEJMPS) leads to failure of the whole concatenated protocol.
This can be easily solved: instead of considering the concatenated protocol as a black box that only succeeds when all subroutines succeed, \emph{what if we condition the execution of each subroutine on the success/failure of previous subroutines?}
Consider for example the concatenated protocol from Figure~\ref{fig.concatenation-sketch}a.
If any of the DEJMPS subroutines fails, the whole protocol fails and the buffered link has to be discarded. However, we can fix this by raising a failure flag whenever any of the first two subroutines fails. If this flag is raised, the third subroutine is not executed and we leave the buffered link untouched.
The flagged version of a concatenated protocol has a larger probability of success, but can also have a lower output fidelity.
This means that it is not clear a priori what is the impact of flags on the buffer performance.
We now analyse a simple case in which we conclude that flags can be either beneficial or detrimental depending on the values of system parameters such as the level of noise $\Gamma$, and not only on the purification policy itself.

Let us consider a policy that operates as follows. For simplicity, we assume that newly generated states $\rho_\mathrm{new}$ are Werner states with fidelity $F_\mathrm{new}$.
When $k$ new links are generated and there is already a link stored in memory~G:
\begin{enumerate}
    \item If $k=1$, we apply the replacement protocol, which has coefficients $a_1=0$, $b_1 = F_\mathrm{new} - 1/4$, $c_1=0$, and $d_1=1$.
    \item If $k\geq2$, we apply the DEJMPS protocol to two of the fresh links and discard the rest. Then, we replace the link in memory with the output from the DEJMPS subroutine, without checking whether it was successful or not. This means that the output fidelity of the protocol is the same as the output fidelity from the DEJMPS subroutine. Since replacement is deterministic, the success probability of this protocol is also the same as the success probability of the DEJMPS subroutine. The purification coefficients for $k\geq2$ are therefore given by $a_k=0$, $b_k = a(\rho_\mathrm{new}) \cdot \left( F_\mathrm{new} - 1/4 \right) + b(\rho_\mathrm{new})$, $c_k=0$, and $d_k = c(\rho_\mathrm{new})  \cdot \left( F_\mathrm{new} - 1/4 \right) + d (\rho_\mathrm{new})$, where $a$, $b$, $c$, and $d$ are the coefficients of the DEJMPS protocol (given in Appendix~\ref{app.subsec.DEJMPS}).
\end{enumerate}
Now, let us consider a flagged variant of the previous policy, with coefficients $a_k'$, $b_k'$, $c_k'$, and $d_k'$. It works as follows:
\begin{enumerate}
    \item When $k=1$, we apply the replacement protocol.
    \item When $k\geq2$ links are generated, the DEJMPS protocol is applied to two of the fresh links, and the rest are discarded. Then, the link in memory is replaced with the output from the DEJMPS subroutine, but only if the subroutine succeeds (otherwise, the buffered link is left untouched).
    This protocol is now fully deterministic, since the buffered link is never removed from memory.
    Consequently, $c_k' = 0$, and $d_k' = 1$.
    The output fidelity of this protocol can be computed as the weighted average of the original fidelity of the link in memory and the output fidelity of the DEJMPS subroutine -- the first term must be weighted by the probability of failure of the subroutine, and the second term by the probability of success.
    Then, the remaining purification coefficients can be computed as $a_k' = 1 - c(\rho_\mathrm{new}) \cdot \left( F_\mathrm{new} - 1/4 \right) - d(\rho_\mathrm{new})$ and $b_k' = a(\rho_\mathrm{new}) \cdot \left( F_\mathrm{new} - 1/4 \right) + b(\rho_\mathrm{new})$, where $a$, $b$, $c$, and $d$ are the coefficients of the DEJMPS protocol (given in Appendix~\ref{app.subsec.DEJMPS}).
\end{enumerate}
By introducing the flags, we have created a protocol with probability of success $p_k'=1 \geq p_k$, where $p_k$ is the probability of success of the original protocol.
However, it can be shown that the output fidelity of the flagged protocol is $J_k'(F) \leq J_k(F)$, where $J_k$ is the jump function of the original protocol. This holds when DEJMPS can improve the fidelity of the newly generated links, i.e. when $J(F_\mathrm{new}) \geq F_\mathrm{new}$, where $J$ is the jump function of DEJMPS. The opposite regime is not interesting, since DEJMPS is decreasing the fidelity of the links and we would be better off not purifying.

As shown in the previous example, internal flags increase the probability of success of purification protocols, which should boost the availability of the buffer.
However, flags may have the side effect of reducing the output fidelity, and therefore it is not clear what is their impact on the average consumed fidelity.
In Figure~\ref{fig.flags}, we show the performance of a 1G$n$B system using the policy described above, versus the level of noise in memory~G.
We show $A$ (orange lines) and $\overline F$ (black lines) for the original policy (solid lines) and the flagged policy (dashed lines).
As expected, the availability is larger for the flagged policy.
The behaviour of $\overline F$ is more interesting.
When the level of noise is low, the flagged policy provides better performance, since it prevents high-quality entanglement from being lost to a failed purification.
However, when noise is strong, flagging becomes detrimental in terms of~$\overline F$: the buffer is likely to store low-quality entanglement due to the strong noise, and flags prevent the buffered link from being discarded earlier due to failed purification and being replaced by a fresh link.

\begin{figure}[t!]
    \centering
    \includegraphics[width=0.9\linewidth]{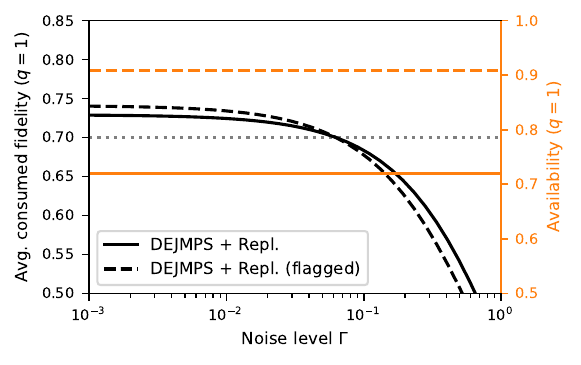}
    \caption{\textbf{Flagged protocols boost the availability but may decrease the average consumed fidelity.}
    Availability $A$ and average consumed fidelity $\overline F$ versus the noise level $\Gamma$, for a 'DEJMPS + Replacement' policy and its flagged version.
    In the first policy, the buffered link is lost when a DEJMPS subroutine fails.
    The second policy incorporates a flag that prevents this from happening -- it succeeds deterministically at the expense of a lower output fidelity.
    The flagged policy yields larger $A$, but may decrease $\overline F$ in some parameter regimes (e.g. when $\Gamma$ is large).
    Parameter values used in this example: $n=2$, $p_\mathrm{gen}=1$, $\rho_\mathrm{new}$ is a Werner state \cite{Werner1989} with $F_\mathrm{new}=0.7$, and $p_\mathrm{con}=0.1$.}
    \label{fig.flags}
\end{figure}

In conclusion, internal flags are a solid tool to improve the availability of entanglement buffers based on concatenated purification protocols.
However, they can decrease the average consumed fidelity in some parameter regimes.
Hence, flagged purification policies should not be assumed to be better than their non-flagged counterparts, and their performance should be carefully evaluated before being adopted.

%%%%%%%%%%%%%%%%%%%%%%%%%%%%%%%%%%%%%%%%
%%%%%%%%%%%%%%%%%%%%%%%%%%%%%%%%%%%%%%%%
%%%%%%%%%%%%%%%%%%%%%%%%%%%%%%%%%%%%%%%%
%\subsection{Noisy protocols}
%\memento{Lower probabilities of success have a stronger impact on availability, while lower jump functions affect the avg. cons. fid. more strongly.}

%%%%%%%%%%%%%%%%%%%%%%%%%%%%%%%%%%%%%
%%%%%%%%%%% OUTLOOK %%%%%%%%%%%%
%%%%%%%%%%%%%%%%%%%%%%%%%%%%%%%%%%%%%
\section{Outlook}\label{sec.outlook}
In this paper, we have studied the behavior of entanglement buffers with one long-lived memory and $n$ short-lived memories (1G$n$B system).
In particular, we have provided analytical expressions for the two main performance measures: the availability and the average consumed fidelity.
These expressions provide valuable insights, such as the fundamental limits to the performance of 1G$n$B systems discussed earlier.

Since our analytical solutions are not computationally expensive to evaluate, we expect our buffering setup to be easy to incorporate in more complex network architectures, such as quantum repeater chains or even large-scale quantum networks.
Additionally, larger buffering systems with multiple long-lived memories, e.g. an $m$G$n$B setup, can be implemented with multiple 1G$n$B systems in parallel.

%Regarding purification policies, a more comprehensive study is left as future work.
Due to the vast freedom in the choice of purification policy, there are multiple ways in which our analysis of purification strategies for entanglement buffers can be extended.
Notably, determining the optimal ordering in which simple protocols should be applied to newly generated links (e.g. concatenated, nested~\cite{Briegel1998}, or banded~\cite{VanMeter2008}) is left as future work.
Additionally, finding policies that optimize availability or average consumed fidelity remains an important open question.

%\begin{itemize}
%    \item Purification policies have been analyzed in the context of a quantum switch in ref. \cite{Jia2024}, although they only considered different concatenation orderings of a specific 2-to-1 protocol.
%\end{itemize}

%%%%%%%%%%%%%%%%%%%%%%%%%%%%%%%%%%%%%
%%%%%%%%%%% REFS, ACKS %%%%%%%%%%%%
%%%%%%%%%%%%%%%%%%%%%%%%%%%%%%%%%%%%%
\bibliographystyle{apsrev4-2}
\bibliography{references}

\begin{thebibliography}{59}%
\makeatletter
\providecommand \@ifxundefined [1]{%
 \@ifx{#1\undefined}
}%
\providecommand \@ifnum [1]{%
 \ifnum #1\expandafter \@firstoftwo
 \else \expandafter \@secondoftwo
 \fi
}%
\providecommand \@ifx [1]{%
 \ifx #1\expandafter \@firstoftwo
 \else \expandafter \@secondoftwo
 \fi
}%
\providecommand \natexlab [1]{#1}%
\providecommand \enquote  [1]{``#1''}%
\providecommand \bibnamefont  [1]{#1}%
\providecommand \bibfnamefont [1]{#1}%
\providecommand \citenamefont [1]{#1}%
\providecommand \href@noop [0]{\@secondoftwo}%
\providecommand \href [0]{\begingroup \@sanitize@url \@href}%
\providecommand \@href[1]{\@@startlink{#1}\@@href}%
\providecommand \@@href[1]{\endgroup#1\@@endlink}%
\providecommand \@sanitize@url [0]{\catcode `\\12\catcode `\$12\catcode `\&12\catcode `\#12\catcode `\^12\catcode `\_12\catcode `\%12\relax}%
\providecommand \@@startlink[1]{}%
\providecommand \@@endlink[0]{}%
\providecommand \url  [0]{\begingroup\@sanitize@url \@url }%
\providecommand \@url [1]{\endgroup\@href {#1}{\urlprefix }}%
\providecommand \urlprefix  [0]{URL }%
\providecommand \Eprint [0]{\href }%
\providecommand \doibase [0]{https://doi.org/}%
\providecommand \selectlanguage [0]{\@gobble}%
\providecommand \bibinfo  [0]{\@secondoftwo}%
\providecommand \bibfield  [0]{\@secondoftwo}%
\providecommand \translation [1]{[#1]}%
\providecommand \BibitemOpen [0]{}%
\providecommand \bibitemStop [0]{}%
\providecommand \bibitemNoStop [0]{.\EOS\space}%
\providecommand \EOS [0]{\spacefactor3000\relax}%
\providecommand \BibitemShut  [1]{\csname bibitem#1\endcsname}%
\let\auto@bib@innerbib\@empty
%</preamble>
\bibitem [{\citenamefont {Ekert}(1991)}]{Ekert1991}%
  \BibitemOpen
  \bibfield  {author} {\bibinfo {author} {\bibfnamefont {A.~K.}\ \bibnamefont {Ekert}},\ }\href@noop {} {\bibfield  {journal} {\bibinfo  {journal} {Phys. Rev. Lett.}\ }\textbf {\bibinfo {volume} {67}},\ \bibinfo {pages} {661} (\bibinfo {year} {1991})}\BibitemShut {NoStop}%
\bibitem [{\citenamefont {Bennett}\ \emph {et~al.}(1992)\citenamefont {Bennett}, \citenamefont {Brassard},\ and\ \citenamefont {Mermin}}]{Bennett1992}%
  \BibitemOpen
  \bibfield  {author} {\bibinfo {author} {\bibfnamefont {C.~H.}\ \bibnamefont {Bennett}}, \bibinfo {author} {\bibfnamefont {G.}~\bibnamefont {Brassard}},\ and\ \bibinfo {author} {\bibfnamefont {N.~D.}\ \bibnamefont {Mermin}},\ }\href@noop {} {\bibfield  {journal} {\bibinfo  {journal} {Phys. Rev. Lett.}\ }\textbf {\bibinfo {volume} {68}},\ \bibinfo {pages} {557} (\bibinfo {year} {1992})}\BibitemShut {NoStop}%
\bibitem [{\citenamefont {Lloyd}(2008)}]{Lloyd2008}%
  \BibitemOpen
  \bibfield  {author} {\bibinfo {author} {\bibfnamefont {S.}~\bibnamefont {Lloyd}},\ }\href@noop {} {\bibfield  {journal} {\bibinfo  {journal} {Science}\ }\textbf {\bibinfo {volume} {321}},\ \bibinfo {pages} {1463} (\bibinfo {year} {2008})}\BibitemShut {NoStop}%
\bibitem [{\citenamefont {Qian}\ \emph {et~al.}(2019)\citenamefont {Qian}, \citenamefont {Eldredge}, \citenamefont {Ge}, \citenamefont {Pagano}, \citenamefont {Monroe}, \citenamefont {Porto},\ and\ \citenamefont {Gorshkov}}]{Qian2019}%
  \BibitemOpen
  \bibfield  {author} {\bibinfo {author} {\bibfnamefont {K.}~\bibnamefont {Qian}}, \bibinfo {author} {\bibfnamefont {Z.}~\bibnamefont {Eldredge}}, \bibinfo {author} {\bibfnamefont {W.}~\bibnamefont {Ge}}, \bibinfo {author} {\bibfnamefont {G.}~\bibnamefont {Pagano}}, \bibinfo {author} {\bibfnamefont {C.}~\bibnamefont {Monroe}}, \bibinfo {author} {\bibfnamefont {J.~V.}\ \bibnamefont {Porto}},\ and\ \bibinfo {author} {\bibfnamefont {A.~V.}\ \bibnamefont {Gorshkov}},\ }\href@noop {} {\bibfield  {journal} {\bibinfo  {journal} {Phys. Rev. A}\ }\textbf {\bibinfo {volume} {100}},\ \bibinfo {pages} {042304} (\bibinfo {year} {2019})}\BibitemShut {NoStop}%
\bibitem [{\citenamefont {England}\ \emph {et~al.}(2019)\citenamefont {England}, \citenamefont {Balaji},\ and\ \citenamefont {Sussman}}]{England2019}%
  \BibitemOpen
  \bibfield  {author} {\bibinfo {author} {\bibfnamefont {D.~G.}\ \bibnamefont {England}}, \bibinfo {author} {\bibfnamefont {B.}~\bibnamefont {Balaji}},\ and\ \bibinfo {author} {\bibfnamefont {B.~J.}\ \bibnamefont {Sussman}},\ }\href@noop {} {\bibfield  {journal} {\bibinfo  {journal} {Phys. Rev. A}\ }\textbf {\bibinfo {volume} {99}},\ \bibinfo {pages} {023828} (\bibinfo {year} {2019})}\BibitemShut {NoStop}%
\bibitem [{\citenamefont {Wu}\ \emph {et~al.}(2023)\citenamefont {Wu}, \citenamefont {Guha},\ and\ \citenamefont {Zhuang}}]{Wu2023}%
  \BibitemOpen
  \bibfield  {author} {\bibinfo {author} {\bibfnamefont {B.-H.}\ \bibnamefont {Wu}}, \bibinfo {author} {\bibfnamefont {S.}~\bibnamefont {Guha}},\ and\ \bibinfo {author} {\bibfnamefont {Q.}~\bibnamefont {Zhuang}},\ }\href@noop {} {\bibfield  {journal} {\bibinfo  {journal} {Quantum Sci. Tech.}\ }\textbf {\bibinfo {volume} {8}},\ \bibinfo {pages} {035016} (\bibinfo {year} {2023})}\BibitemShut {NoStop}%
\bibitem [{\citenamefont {Brassard}\ \emph {et~al.}(2005)\citenamefont {Brassard}, \citenamefont {Broadbent},\ and\ \citenamefont {Tapp}}]{Brassard2005}%
  \BibitemOpen
  \bibfield  {author} {\bibinfo {author} {\bibfnamefont {G.}~\bibnamefont {Brassard}}, \bibinfo {author} {\bibfnamefont {A.}~\bibnamefont {Broadbent}},\ and\ \bibinfo {author} {\bibfnamefont {A.}~\bibnamefont {Tapp}},\ }\href@noop {} {\bibfield  {journal} {\bibinfo  {journal} {Found. Phys.}\ }\textbf {\bibinfo {volume} {35}},\ \bibinfo {pages} {1877} (\bibinfo {year} {2005})}\BibitemShut {NoStop}%
\bibitem [{\citenamefont {Broadbent}\ and\ \citenamefont {Tapp}(2008)}]{Broadbent2008}%
  \BibitemOpen
  \bibfield  {author} {\bibinfo {author} {\bibfnamefont {A.}~\bibnamefont {Broadbent}}\ and\ \bibinfo {author} {\bibfnamefont {A.}~\bibnamefont {Tapp}},\ }\href@noop {} {\bibfield  {journal} {\bibinfo  {journal} {ACM SIGACT News}\ }\textbf {\bibinfo {volume} {39}},\ \bibinfo {pages} {67} (\bibinfo {year} {2008})}\BibitemShut {NoStop}%
\bibitem [{\citenamefont {Chakraborty}\ \emph {et~al.}(2019)\citenamefont {Chakraborty}, \citenamefont {Rozpedek}, \citenamefont {Dahlberg},\ and\ \citenamefont {Wehner}}]{Chakraborty2019}%
  \BibitemOpen
  \bibfield  {author} {\bibinfo {author} {\bibfnamefont {K.}~\bibnamefont {Chakraborty}}, \bibinfo {author} {\bibfnamefont {F.}~\bibnamefont {Rozpedek}}, \bibinfo {author} {\bibfnamefont {A.}~\bibnamefont {Dahlberg}},\ and\ \bibinfo {author} {\bibfnamefont {S.}~\bibnamefont {Wehner}},\ }\href@noop {} {\bibfield  {journal} {\bibinfo  {journal} {arXiv preprint arXiv:1907.11630}\ } (\bibinfo {year} {2019})}\BibitemShut {NoStop}%
\bibitem [{\citenamefont {Ghaderibaneh}\ \emph {et~al.}(2022)\citenamefont {Ghaderibaneh}, \citenamefont {Gupta}, \citenamefont {Ramakrishnan},\ and\ \citenamefont {Luo}}]{Ghaderibaneh2022}%
  \BibitemOpen
  \bibfield  {author} {\bibinfo {author} {\bibfnamefont {M.}~\bibnamefont {Ghaderibaneh}}, \bibinfo {author} {\bibfnamefont {H.}~\bibnamefont {Gupta}}, \bibinfo {author} {\bibfnamefont {C.}~\bibnamefont {Ramakrishnan}},\ and\ \bibinfo {author} {\bibfnamefont {E.}~\bibnamefont {Luo}},\ }in\ \href@noop {} {\emph {\bibinfo {booktitle} {2022 IEEE International Conference on Quantum Computing and Engineering (QCE)}}}\ (\bibinfo {organization} {IEEE},\ \bibinfo {year} {2022})\ pp.\ \bibinfo {pages} {426--436}\BibitemShut {NoStop}%
\bibitem [{\citenamefont {Pouryousef}\ \emph {et~al.}(2023)\citenamefont {Pouryousef}, \citenamefont {Panigrahy},\ and\ \citenamefont {Towsley}}]{Pouryousef2022}%
  \BibitemOpen
  \bibfield  {author} {\bibinfo {author} {\bibfnamefont {S.}~\bibnamefont {Pouryousef}}, \bibinfo {author} {\bibfnamefont {N.~K.}\ \bibnamefont {Panigrahy}},\ and\ \bibinfo {author} {\bibfnamefont {D.}~\bibnamefont {Towsley}},\ }in\ \href@noop {} {\emph {\bibinfo {booktitle} {IEEE INFOCOM 2023-IEEE Conference on Computer Communications}}}\ (\bibinfo {organization} {IEEE},\ \bibinfo {year} {2023})\ pp.\ \bibinfo {pages} {1--10}\BibitemShut {NoStop}%
\bibitem [{\citenamefont {I{\~n}esta}\ and\ \citenamefont {Wehner}(2023)}]{inesta2023performance}%
  \BibitemOpen
  \bibfield  {author} {\bibinfo {author} {\bibfnamefont {{\'A}.~G.}\ \bibnamefont {I{\~n}esta}}\ and\ \bibinfo {author} {\bibfnamefont {S.}~\bibnamefont {Wehner}},\ }\href@noop {} {\bibfield  {journal} {\bibinfo  {journal} {Phys. Rev. A}\ }\textbf {\bibinfo {volume} {108}},\ \bibinfo {pages} {052615} (\bibinfo {year} {2023})}\BibitemShut {NoStop}%
\bibitem [{\citenamefont {Askarani}\ \emph {et~al.}(2021)\citenamefont {Askarani}, \citenamefont {Chakraborty},\ and\ \citenamefont {Do~Amaral}}]{Askarani2021}%
  \BibitemOpen
  \bibfield  {author} {\bibinfo {author} {\bibfnamefont {M.~F.}\ \bibnamefont {Askarani}}, \bibinfo {author} {\bibfnamefont {K.}~\bibnamefont {Chakraborty}},\ and\ \bibinfo {author} {\bibfnamefont {G.~C.}\ \bibnamefont {Do~Amaral}},\ }\href@noop {} {\bibfield  {journal} {\bibinfo  {journal} {New J. Phys.}\ }\textbf {\bibinfo {volume} {23}},\ \bibinfo {pages} {063078} (\bibinfo {year} {2021})}\BibitemShut {NoStop}%
\bibitem [{\citenamefont {Davies}\ \emph {et~al.}(2024)\citenamefont {Davies}, \citenamefont {I{\~n}esta},\ and\ \citenamefont {Wehner}}]{Davies2023a}%
  \BibitemOpen
  \bibfield  {author} {\bibinfo {author} {\bibfnamefont {B.}~\bibnamefont {Davies}}, \bibinfo {author} {\bibfnamefont {{\'A}.~G.}\ \bibnamefont {I{\~n}esta}},\ and\ \bibinfo {author} {\bibfnamefont {S.}~\bibnamefont {Wehner}},\ }\href@noop {} {\bibfield  {journal} {\bibinfo  {journal} {Quantum}\ }\textbf {\bibinfo {volume} {8}},\ \bibinfo {pages} {1458} (\bibinfo {year} {2024})}\BibitemShut {NoStop}%
\bibitem [{\citenamefont {Elsayed}\ \emph {et~al.}(2024)\citenamefont {Elsayed}, \citenamefont {KhudaBukhsh},\ and\ \citenamefont {Rizk}}]{Elsayed2023}%
  \BibitemOpen
  \bibfield  {author} {\bibinfo {author} {\bibfnamefont {K.~S.}\ \bibnamefont {Elsayed}}, \bibinfo {author} {\bibfnamefont {W.~R.}\ \bibnamefont {KhudaBukhsh}},\ and\ \bibinfo {author} {\bibfnamefont {A.}~\bibnamefont {Rizk}},\ }in\ \href@noop {} {\emph {\bibinfo {booktitle} {ICC 2024-IEEE International Conference on Communications}}}\ (\bibinfo {organization} {IEEE},\ \bibinfo {year} {2024})\ pp.\ \bibinfo {pages} {485--490}\BibitemShut {NoStop}%
\bibitem [{\citenamefont {Bennett}\ \emph {et~al.}(1996{\natexlab{a}})\citenamefont {Bennett}, \citenamefont {Brassard}, \citenamefont {Popescu}, \citenamefont {Schumacher}, \citenamefont {Smolin},\ and\ \citenamefont {Wootters}}]{Bennett1996}%
  \BibitemOpen
  \bibfield  {author} {\bibinfo {author} {\bibfnamefont {C.~H.}\ \bibnamefont {Bennett}}, \bibinfo {author} {\bibfnamefont {G.}~\bibnamefont {Brassard}}, \bibinfo {author} {\bibfnamefont {S.}~\bibnamefont {Popescu}}, \bibinfo {author} {\bibfnamefont {B.}~\bibnamefont {Schumacher}}, \bibinfo {author} {\bibfnamefont {J.~A.}\ \bibnamefont {Smolin}},\ and\ \bibinfo {author} {\bibfnamefont {W.~K.}\ \bibnamefont {Wootters}},\ }\href@noop {} {\bibfield  {journal} {\bibinfo  {journal} {Phys. Rev. Lett.}\ }\textbf {\bibinfo {volume} {76}},\ \bibinfo {pages} {722} (\bibinfo {year} {1996}{\natexlab{a}})}\BibitemShut {NoStop}%
\bibitem [{\citenamefont {Deutsch}\ \emph {et~al.}(1996)\citenamefont {Deutsch}, \citenamefont {Ekert}, \citenamefont {Jozsa}, \citenamefont {Macchiavello}, \citenamefont {Popescu},\ and\ \citenamefont {Sanpera}}]{Deutsch1996}%
  \BibitemOpen
  \bibfield  {author} {\bibinfo {author} {\bibfnamefont {D.}~\bibnamefont {Deutsch}}, \bibinfo {author} {\bibfnamefont {A.}~\bibnamefont {Ekert}}, \bibinfo {author} {\bibfnamefont {R.}~\bibnamefont {Jozsa}}, \bibinfo {author} {\bibfnamefont {C.}~\bibnamefont {Macchiavello}}, \bibinfo {author} {\bibfnamefont {S.}~\bibnamefont {Popescu}},\ and\ \bibinfo {author} {\bibfnamefont {A.}~\bibnamefont {Sanpera}},\ }\href@noop {} {\bibfield  {journal} {\bibinfo  {journal} {Phys. Rev. Lett.}\ }\textbf {\bibinfo {volume} {77}},\ \bibinfo {pages} {2818} (\bibinfo {year} {1996})}\BibitemShut {NoStop}%
\bibitem [{\citenamefont {D{\"u}r}\ \emph {et~al.}(1999)\citenamefont {D{\"u}r}, \citenamefont {Briegel}, \citenamefont {Cirac},\ and\ \citenamefont {Zoller}}]{Dur1999}%
  \BibitemOpen
  \bibfield  {author} {\bibinfo {author} {\bibfnamefont {W.}~\bibnamefont {D{\"u}r}}, \bibinfo {author} {\bibfnamefont {H.-J.}\ \bibnamefont {Briegel}}, \bibinfo {author} {\bibfnamefont {J.~I.}\ \bibnamefont {Cirac}},\ and\ \bibinfo {author} {\bibfnamefont {P.}~\bibnamefont {Zoller}},\ }\href@noop {} {\bibfield  {journal} {\bibinfo  {journal} {Phys. Rev. A}\ }\textbf {\bibinfo {volume} {59}},\ \bibinfo {pages} {169} (\bibinfo {year} {1999})}\BibitemShut {NoStop}%
\bibitem [{\citenamefont {Yan}\ \emph {et~al.}(2023)\citenamefont {Yan}, \citenamefont {Zhou}, \citenamefont {Zhong},\ and\ \citenamefont {Sheng}}]{Yan2023}%
  \BibitemOpen
  \bibfield  {author} {\bibinfo {author} {\bibfnamefont {P.-S.}\ \bibnamefont {Yan}}, \bibinfo {author} {\bibfnamefont {L.}~\bibnamefont {Zhou}}, \bibinfo {author} {\bibfnamefont {W.}~\bibnamefont {Zhong}},\ and\ \bibinfo {author} {\bibfnamefont {Y.-B.}\ \bibnamefont {Sheng}},\ }\href@noop {} {\bibfield  {journal} {\bibinfo  {journal} {Science China Physics, Mechanics \& Astronomy}\ }\textbf {\bibinfo {volume} {66}},\ \bibinfo {pages} {250301} (\bibinfo {year} {2023})}\BibitemShut {NoStop}%
\bibitem [{\citenamefont {Haldar}\ \emph {et~al.}(2024)\citenamefont {Haldar}, \citenamefont {Barge}, \citenamefont {Cheng}, \citenamefont {Chang}, \citenamefont {Kirby}, \citenamefont {Khatri}, \citenamefont {Wong},\ and\ \citenamefont {Lee}}]{Haldar2024}%
  \BibitemOpen
  \bibfield  {author} {\bibinfo {author} {\bibfnamefont {S.}~\bibnamefont {Haldar}}, \bibinfo {author} {\bibfnamefont {P.~J.}\ \bibnamefont {Barge}}, \bibinfo {author} {\bibfnamefont {X.}~\bibnamefont {Cheng}}, \bibinfo {author} {\bibfnamefont {K.-C.}\ \bibnamefont {Chang}}, \bibinfo {author} {\bibfnamefont {B.~T.}\ \bibnamefont {Kirby}}, \bibinfo {author} {\bibfnamefont {S.}~\bibnamefont {Khatri}}, \bibinfo {author} {\bibfnamefont {C.~W.}\ \bibnamefont {Wong}},\ and\ \bibinfo {author} {\bibfnamefont {H.}~\bibnamefont {Lee}},\ }\href@noop {} {\bibfield  {journal} {\bibinfo  {journal} {arXiv preprint arXiv:2401.13168}\ } (\bibinfo {year} {2024})}\BibitemShut {NoStop}%
\bibitem [{\citenamefont {Victora}\ \emph {et~al.}(2023)\citenamefont {Victora}, \citenamefont {Tserkis}, \citenamefont {Krastanov}, \citenamefont {de~la Cerda}, \citenamefont {Willis},\ and\ \citenamefont {Narang}}]{Victora2023}%
  \BibitemOpen
  \bibfield  {author} {\bibinfo {author} {\bibfnamefont {M.}~\bibnamefont {Victora}}, \bibinfo {author} {\bibfnamefont {S.}~\bibnamefont {Tserkis}}, \bibinfo {author} {\bibfnamefont {S.}~\bibnamefont {Krastanov}}, \bibinfo {author} {\bibfnamefont {A.~S.}\ \bibnamefont {de~la Cerda}}, \bibinfo {author} {\bibfnamefont {S.}~\bibnamefont {Willis}},\ and\ \bibinfo {author} {\bibfnamefont {P.}~\bibnamefont {Narang}},\ }\href@noop {} {\bibfield  {journal} {\bibinfo  {journal} {Phys. Rev. Research}\ }\textbf {\bibinfo {volume} {5}},\ \bibinfo {pages} {033171} (\bibinfo {year} {2023})}\BibitemShut {NoStop}%
\bibitem [{\citenamefont {Bradley}\ \emph {et~al.}(2019)\citenamefont {Bradley}, \citenamefont {Randall}, \citenamefont {Abobeih}, \citenamefont {Berrevoets}, \citenamefont {Degen}, \citenamefont {Bakker}, \citenamefont {Markham}, \citenamefont {Twitchen},\ and\ \citenamefont {Taminiau}}]{Bradley2019}%
  \BibitemOpen
  \bibfield  {author} {\bibinfo {author} {\bibfnamefont {C.~E.}\ \bibnamefont {Bradley}}, \bibinfo {author} {\bibfnamefont {J.}~\bibnamefont {Randall}}, \bibinfo {author} {\bibfnamefont {M.~H.}\ \bibnamefont {Abobeih}}, \bibinfo {author} {\bibfnamefont {R.~C.}\ \bibnamefont {Berrevoets}}, \bibinfo {author} {\bibfnamefont {M.~J.}\ \bibnamefont {Degen}}, \bibinfo {author} {\bibfnamefont {M.~A.}\ \bibnamefont {Bakker}}, \bibinfo {author} {\bibfnamefont {M.}~\bibnamefont {Markham}}, \bibinfo {author} {\bibfnamefont {D.~J.}\ \bibnamefont {Twitchen}},\ and\ \bibinfo {author} {\bibfnamefont {T.~H.}\ \bibnamefont {Taminiau}},\ }\href@noop {} {\bibfield  {journal} {\bibinfo  {journal} {Phys. Rev. X}\ }\textbf {\bibinfo {volume} {9}},\ \bibinfo {pages} {031045} (\bibinfo {year} {2019})}\BibitemShut {NoStop}%
\bibitem [{\citenamefont {Abobeih}\ \emph {et~al.}(2018)\citenamefont {Abobeih}, \citenamefont {Cramer}, \citenamefont {Bakker}, \citenamefont {Kalb}, \citenamefont {Markham}, \citenamefont {Twitchen},\ and\ \citenamefont {Taminiau}}]{Abobeih2018}%
  \BibitemOpen
  \bibfield  {author} {\bibinfo {author} {\bibfnamefont {M.~H.}\ \bibnamefont {Abobeih}}, \bibinfo {author} {\bibfnamefont {J.}~\bibnamefont {Cramer}}, \bibinfo {author} {\bibfnamefont {M.~A.}\ \bibnamefont {Bakker}}, \bibinfo {author} {\bibfnamefont {N.}~\bibnamefont {Kalb}}, \bibinfo {author} {\bibfnamefont {M.}~\bibnamefont {Markham}}, \bibinfo {author} {\bibfnamefont {D.~J.}\ \bibnamefont {Twitchen}},\ and\ \bibinfo {author} {\bibfnamefont {T.~H.}\ \bibnamefont {Taminiau}},\ }\href@noop {} {\bibfield  {journal} {\bibinfo  {journal} {Nat. Commun.}\ }\textbf {\bibinfo {volume} {9}},\ \bibinfo {pages} {2552} (\bibinfo {year} {2018})}\BibitemShut {NoStop}%
\bibitem [{\citenamefont {Pompili}\ \emph {et~al.}(2021)\citenamefont {Pompili}, \citenamefont {Hermans}, \citenamefont {Baier}, \citenamefont {Beukers}, \citenamefont {Humphreys}, \citenamefont {Schouten}, \citenamefont {Vermeulen}, \citenamefont {Tiggelman}, \citenamefont {dos Santos~Martins}, \citenamefont {Dirkse} \emph {et~al.}}]{Pompili2021}%
  \BibitemOpen
  \bibfield  {author} {\bibinfo {author} {\bibfnamefont {M.}~\bibnamefont {Pompili}}, \bibinfo {author} {\bibfnamefont {S.~L.}\ \bibnamefont {Hermans}}, \bibinfo {author} {\bibfnamefont {S.}~\bibnamefont {Baier}}, \bibinfo {author} {\bibfnamefont {H.~K.}\ \bibnamefont {Beukers}}, \bibinfo {author} {\bibfnamefont {P.~C.}\ \bibnamefont {Humphreys}}, \bibinfo {author} {\bibfnamefont {R.~N.}\ \bibnamefont {Schouten}}, \bibinfo {author} {\bibfnamefont {R.~F.}\ \bibnamefont {Vermeulen}}, \bibinfo {author} {\bibfnamefont {M.~J.}\ \bibnamefont {Tiggelman}}, \bibinfo {author} {\bibfnamefont {L.}~\bibnamefont {dos Santos~Martins}}, \bibinfo {author} {\bibfnamefont {B.}~\bibnamefont {Dirkse}}, \emph {et~al.},\ }\href@noop {} {\bibfield  {journal} {\bibinfo  {journal} {Science}\ }\textbf {\bibinfo {volume} {372}},\ \bibinfo {pages} {259} (\bibinfo {year} {2021})}\BibitemShut {NoStop}%
\bibitem [{\citenamefont {Wengerowsky}\ \emph {et~al.}(2018)\citenamefont {Wengerowsky}, \citenamefont {Joshi}, \citenamefont {Steinlechner}, \citenamefont {H{\"u}bel},\ and\ \citenamefont {Ursin}}]{Wengerowsky2018}%
  \BibitemOpen
  \bibfield  {author} {\bibinfo {author} {\bibfnamefont {S.}~\bibnamefont {Wengerowsky}}, \bibinfo {author} {\bibfnamefont {S.~K.}\ \bibnamefont {Joshi}}, \bibinfo {author} {\bibfnamefont {F.}~\bibnamefont {Steinlechner}}, \bibinfo {author} {\bibfnamefont {H.}~\bibnamefont {H{\"u}bel}},\ and\ \bibinfo {author} {\bibfnamefont {R.}~\bibnamefont {Ursin}},\ }\href@noop {} {\bibfield  {journal} {\bibinfo  {journal} {Nature}\ }\textbf {\bibinfo {volume} {564}},\ \bibinfo {pages} {225} (\bibinfo {year} {2018})}\BibitemShut {NoStop}%
\bibitem [{\citenamefont {Chen}\ \emph {et~al.}(2023)\citenamefont {Chen}, \citenamefont {Dhara}, \citenamefont {Heuck}, \citenamefont {Lee}, \citenamefont {Dai}, \citenamefont {Guha},\ and\ \citenamefont {Englund}}]{Chen2023}%
  \BibitemOpen
  \bibfield  {author} {\bibinfo {author} {\bibfnamefont {K.~C.}\ \bibnamefont {Chen}}, \bibinfo {author} {\bibfnamefont {P.}~\bibnamefont {Dhara}}, \bibinfo {author} {\bibfnamefont {M.}~\bibnamefont {Heuck}}, \bibinfo {author} {\bibfnamefont {Y.}~\bibnamefont {Lee}}, \bibinfo {author} {\bibfnamefont {W.}~\bibnamefont {Dai}}, \bibinfo {author} {\bibfnamefont {S.}~\bibnamefont {Guha}},\ and\ \bibinfo {author} {\bibfnamefont {D.}~\bibnamefont {Englund}},\ }\href@noop {} {\bibfield  {journal} {\bibinfo  {journal} {Phys. Rev. Applied}\ }\textbf {\bibinfo {volume} {19}},\ \bibinfo {pages} {054029} (\bibinfo {year} {2023})}\BibitemShut {NoStop}%
\bibitem [{\citenamefont {Mower}\ and\ \citenamefont {Englund}(2011)}]{Mower2011}%
  \BibitemOpen
  \bibfield  {author} {\bibinfo {author} {\bibfnamefont {J.}~\bibnamefont {Mower}}\ and\ \bibinfo {author} {\bibfnamefont {D.}~\bibnamefont {Englund}},\ }\href@noop {} {\bibfield  {journal} {\bibinfo  {journal} {Phys. Rev. A}\ }\textbf {\bibinfo {volume} {84}},\ \bibinfo {pages} {052326} (\bibinfo {year} {2011})}\BibitemShut {NoStop}%
\bibitem [{\citenamefont {Krutyanskiy}\ \emph {et~al.}(2024)\citenamefont {Krutyanskiy}, \citenamefont {Canteri}, \citenamefont {Meraner}, \citenamefont {Krcmarsky},\ and\ \citenamefont {Lanyon}}]{Krutyanskiy2024}%
  \BibitemOpen
  \bibfield  {author} {\bibinfo {author} {\bibfnamefont {V.}~\bibnamefont {Krutyanskiy}}, \bibinfo {author} {\bibfnamefont {M.}~\bibnamefont {Canteri}}, \bibinfo {author} {\bibfnamefont {M.}~\bibnamefont {Meraner}}, \bibinfo {author} {\bibfnamefont {V.}~\bibnamefont {Krcmarsky}},\ and\ \bibinfo {author} {\bibfnamefont {B.}~\bibnamefont {Lanyon}},\ }\href@noop {} {\bibfield  {journal} {\bibinfo  {journal} {PRX Quantum}\ }\textbf {\bibinfo {volume} {5}},\ \bibinfo {pages} {020308} (\bibinfo {year} {2024})}\BibitemShut {NoStop}%
\bibitem [{\citenamefont {Collins}\ \emph {et~al.}(2007)\citenamefont {Collins}, \citenamefont {Jenkins}, \citenamefont {Kuzmich},\ and\ \citenamefont {Kennedy}}]{Collins2007}%
  \BibitemOpen
  \bibfield  {author} {\bibinfo {author} {\bibfnamefont {O.}~\bibnamefont {Collins}}, \bibinfo {author} {\bibfnamefont {S.}~\bibnamefont {Jenkins}}, \bibinfo {author} {\bibfnamefont {A.}~\bibnamefont {Kuzmich}},\ and\ \bibinfo {author} {\bibfnamefont {T.}~\bibnamefont {Kennedy}},\ }\href@noop {} {\bibfield  {journal} {\bibinfo  {journal} {Phys. Rev. Lett.}\ }\textbf {\bibinfo {volume} {98}},\ \bibinfo {pages} {060502} (\bibinfo {year} {2007})}\BibitemShut {NoStop}%
\bibitem [{\citenamefont {Munro}\ \emph {et~al.}(2010)\citenamefont {Munro}, \citenamefont {Harrison}, \citenamefont {Stephens}, \citenamefont {Devitt},\ and\ \citenamefont {Nemoto}}]{Munro2010}%
  \BibitemOpen
  \bibfield  {author} {\bibinfo {author} {\bibfnamefont {W.}~\bibnamefont {Munro}}, \bibinfo {author} {\bibfnamefont {K.}~\bibnamefont {Harrison}}, \bibinfo {author} {\bibfnamefont {A.}~\bibnamefont {Stephens}}, \bibinfo {author} {\bibfnamefont {S.}~\bibnamefont {Devitt}},\ and\ \bibinfo {author} {\bibfnamefont {K.}~\bibnamefont {Nemoto}},\ }\href@noop {} {\bibfield  {journal} {\bibinfo  {journal} {Nature Photonics}\ }\textbf {\bibinfo {volume} {4}},\ \bibinfo {pages} {792} (\bibinfo {year} {2010})}\BibitemShut {NoStop}%
\bibitem [{\citenamefont {van Dam}\ \emph {et~al.}(2017)\citenamefont {van Dam}, \citenamefont {Humphreys}, \citenamefont {Rozpedek}, \citenamefont {Wehner},\ and\ \citenamefont {Hanson}}]{Dam2017}%
  \BibitemOpen
  \bibfield  {author} {\bibinfo {author} {\bibfnamefont {S.~B.}\ \bibnamefont {van Dam}}, \bibinfo {author} {\bibfnamefont {P.~C.}\ \bibnamefont {Humphreys}}, \bibinfo {author} {\bibfnamefont {F.}~\bibnamefont {Rozpedek}}, \bibinfo {author} {\bibfnamefont {S.}~\bibnamefont {Wehner}},\ and\ \bibinfo {author} {\bibfnamefont {R.}~\bibnamefont {Hanson}},\ }\href@noop {} {\bibfield  {journal} {\bibinfo  {journal} {Quantum Sci. Tech.}\ }\textbf {\bibinfo {volume} {2}},\ \bibinfo {pages} {034002} (\bibinfo {year} {2017})}\BibitemShut {NoStop}%
\bibitem [{\citenamefont {Barrett}\ and\ \citenamefont {Kok}(2005)}]{Barrett2005}%
  \BibitemOpen
  \bibfield  {author} {\bibinfo {author} {\bibfnamefont {S.~D.}\ \bibnamefont {Barrett}}\ and\ \bibinfo {author} {\bibfnamefont {P.}~\bibnamefont {Kok}},\ }\href@noop {} {\bibfield  {journal} {\bibinfo  {journal} {Phys. Rev. A}\ }\textbf {\bibinfo {volume} {71}},\ \bibinfo {pages} {060310} (\bibinfo {year} {2005})}\BibitemShut {NoStop}%
\bibitem [{\citenamefont {Togan}\ \emph {et~al.}(2010)\citenamefont {Togan}, \citenamefont {Chu}, \citenamefont {Trifonov}, \citenamefont {Jiang}, \citenamefont {Maze}, \citenamefont {Childress}, \citenamefont {Dutt}, \citenamefont {S{\o}rensen}, \citenamefont {Hemmer}, \citenamefont {Zibrov} \emph {et~al.}}]{Togan2010}%
  \BibitemOpen
  \bibfield  {author} {\bibinfo {author} {\bibfnamefont {E.}~\bibnamefont {Togan}}, \bibinfo {author} {\bibfnamefont {Y.}~\bibnamefont {Chu}}, \bibinfo {author} {\bibfnamefont {A.~S.}\ \bibnamefont {Trifonov}}, \bibinfo {author} {\bibfnamefont {L.}~\bibnamefont {Jiang}}, \bibinfo {author} {\bibfnamefont {J.}~\bibnamefont {Maze}}, \bibinfo {author} {\bibfnamefont {L.}~\bibnamefont {Childress}}, \bibinfo {author} {\bibfnamefont {M.~G.}\ \bibnamefont {Dutt}}, \bibinfo {author} {\bibfnamefont {A.~S.}\ \bibnamefont {S{\o}rensen}}, \bibinfo {author} {\bibfnamefont {P.~R.}\ \bibnamefont {Hemmer}}, \bibinfo {author} {\bibfnamefont {A.~S.}\ \bibnamefont {Zibrov}}, \emph {et~al.},\ }\href@noop {} {\bibfield  {journal} {\bibinfo  {journal} {Nature}\ }\textbf {\bibinfo {volume} {466}},\ \bibinfo {pages} {730} (\bibinfo {year} {2010})}\BibitemShut {NoStop}%
\bibitem [{\citenamefont {Bernien}\ \emph {et~al.}(2013)\citenamefont {Bernien}, \citenamefont {Hensen}, \citenamefont {Pfaff}, \citenamefont {Koolstra}, \citenamefont {Blok}, \citenamefont {Robledo}, \citenamefont {Taminiau}, \citenamefont {Markham}, \citenamefont {Twitchen}, \citenamefont {Childress} \emph {et~al.}}]{Bernien2013}%
  \BibitemOpen
  \bibfield  {author} {\bibinfo {author} {\bibfnamefont {H.}~\bibnamefont {Bernien}}, \bibinfo {author} {\bibfnamefont {B.}~\bibnamefont {Hensen}}, \bibinfo {author} {\bibfnamefont {W.}~\bibnamefont {Pfaff}}, \bibinfo {author} {\bibfnamefont {G.}~\bibnamefont {Koolstra}}, \bibinfo {author} {\bibfnamefont {M.~S.}\ \bibnamefont {Blok}}, \bibinfo {author} {\bibfnamefont {L.}~\bibnamefont {Robledo}}, \bibinfo {author} {\bibfnamefont {T.~H.}\ \bibnamefont {Taminiau}}, \bibinfo {author} {\bibfnamefont {M.}~\bibnamefont {Markham}}, \bibinfo {author} {\bibfnamefont {D.~J.}\ \bibnamefont {Twitchen}}, \bibinfo {author} {\bibfnamefont {L.}~\bibnamefont {Childress}}, \emph {et~al.},\ }\href@noop {} {\bibfield  {journal} {\bibinfo  {journal} {Nature}\ }\textbf {\bibinfo {volume} {497}},\ \bibinfo {pages} {86} (\bibinfo {year} {2013})}\BibitemShut {NoStop}%
\bibitem [{\citenamefont {Zhou}\ \emph {et~al.}(2024)\citenamefont {Zhou}, \citenamefont {Malik}, \citenamefont {Fertig}, \citenamefont {Bock}, \citenamefont {Bauer}, \citenamefont {van Leent}, \citenamefont {Zhang}, \citenamefont {Becher},\ and\ \citenamefont {Weinfurter}}]{Zhou2024}%
  \BibitemOpen
  \bibfield  {author} {\bibinfo {author} {\bibfnamefont {Y.}~\bibnamefont {Zhou}}, \bibinfo {author} {\bibfnamefont {P.}~\bibnamefont {Malik}}, \bibinfo {author} {\bibfnamefont {F.}~\bibnamefont {Fertig}}, \bibinfo {author} {\bibfnamefont {M.}~\bibnamefont {Bock}}, \bibinfo {author} {\bibfnamefont {T.}~\bibnamefont {Bauer}}, \bibinfo {author} {\bibfnamefont {T.}~\bibnamefont {van Leent}}, \bibinfo {author} {\bibfnamefont {W.}~\bibnamefont {Zhang}}, \bibinfo {author} {\bibfnamefont {C.}~\bibnamefont {Becher}},\ and\ \bibinfo {author} {\bibfnamefont {H.}~\bibnamefont {Weinfurter}},\ }\href@noop {} {\bibfield  {journal} {\bibinfo  {journal} {PRX Quantum}\ }\textbf {\bibinfo {volume} {5}},\ \bibinfo {pages} {020307} (\bibinfo {year} {2024})}\BibitemShut {NoStop}%
\bibitem [{\citenamefont {Jansen}\ \emph {et~al.}(2022)\citenamefont {Jansen}, \citenamefont {Goodenough}, \citenamefont {de~Bone}, \citenamefont {Gijswijt},\ and\ \citenamefont {Elkouss}}]{Jansen2022}%
  \BibitemOpen
  \bibfield  {author} {\bibinfo {author} {\bibfnamefont {S.}~\bibnamefont {Jansen}}, \bibinfo {author} {\bibfnamefont {K.}~\bibnamefont {Goodenough}}, \bibinfo {author} {\bibfnamefont {S.}~\bibnamefont {de~Bone}}, \bibinfo {author} {\bibfnamefont {D.}~\bibnamefont {Gijswijt}},\ and\ \bibinfo {author} {\bibfnamefont {D.}~\bibnamefont {Elkouss}},\ }\href@noop {} {\bibfield  {journal} {\bibinfo  {journal} {Quantum}\ }\textbf {\bibinfo {volume} {6}},\ \bibinfo {pages} {715} (\bibinfo {year} {2022})}\BibitemShut {NoStop}%
\bibitem [{\citenamefont {Bennett}\ \emph {et~al.}(1996{\natexlab{b}})\citenamefont {Bennett}, \citenamefont {DiVincenzo}, \citenamefont {Smolin},\ and\ \citenamefont {Wootters}}]{Bennett1996a}%
  \BibitemOpen
  \bibfield  {author} {\bibinfo {author} {\bibfnamefont {C.~H.}\ \bibnamefont {Bennett}}, \bibinfo {author} {\bibfnamefont {D.~P.}\ \bibnamefont {DiVincenzo}}, \bibinfo {author} {\bibfnamefont {J.~A.}\ \bibnamefont {Smolin}},\ and\ \bibinfo {author} {\bibfnamefont {W.~K.}\ \bibnamefont {Wootters}},\ }\href@noop {} {\bibfield  {journal} {\bibinfo  {journal} {Phys. Rev. A}\ }\textbf {\bibinfo {volume} {54}},\ \bibinfo {pages} {3824} (\bibinfo {year} {1996}{\natexlab{b}})}\BibitemShut {NoStop}%
\bibitem [{\citenamefont {Horodecki}\ \emph {et~al.}(1999)\citenamefont {Horodecki}, \citenamefont {Horodecki},\ and\ \citenamefont {Horodecki}}]{Horodecki1999}%
  \BibitemOpen
  \bibfield  {author} {\bibinfo {author} {\bibfnamefont {M.}~\bibnamefont {Horodecki}}, \bibinfo {author} {\bibfnamefont {P.}~\bibnamefont {Horodecki}},\ and\ \bibinfo {author} {\bibfnamefont {R.}~\bibnamefont {Horodecki}},\ }\href@noop {} {\bibfield  {journal} {\bibinfo  {journal} {Phys. Rev. A}\ }\textbf {\bibinfo {volume} {60}},\ \bibinfo {pages} {1888} (\bibinfo {year} {1999})}\BibitemShut {NoStop}%
\bibitem [{\citenamefont {Benjamin}\ \emph {et~al.}(2006)\citenamefont {Benjamin}, \citenamefont {Browne}, \citenamefont {Fitzsimons},\ and\ \citenamefont {Morton}}]{Benjamin2006}%
  \BibitemOpen
  \bibfield  {author} {\bibinfo {author} {\bibfnamefont {S.~C.}\ \bibnamefont {Benjamin}}, \bibinfo {author} {\bibfnamefont {D.~E.}\ \bibnamefont {Browne}}, \bibinfo {author} {\bibfnamefont {J.}~\bibnamefont {Fitzsimons}},\ and\ \bibinfo {author} {\bibfnamefont {J.~J.}\ \bibnamefont {Morton}},\ }\href@noop {} {\bibfield  {journal} {\bibinfo  {journal} {New J. Phys.}\ }\textbf {\bibinfo {volume} {8}},\ \bibinfo {pages} {141} (\bibinfo {year} {2006})}\BibitemShut {NoStop}%
\bibitem [{\citenamefont {Rozpedek}\ \emph {et~al.}(2019)\citenamefont {Rozpedek}, \citenamefont {Yehia}, \citenamefont {Goodenough}, \citenamefont {Ruf}, \citenamefont {Humphreys}, \citenamefont {Hanson}, \citenamefont {Wehner},\ and\ \citenamefont {Elkouss}}]{Rozpedek2019}%
  \BibitemOpen
  \bibfield  {author} {\bibinfo {author} {\bibfnamefont {F.}~\bibnamefont {Rozpedek}}, \bibinfo {author} {\bibfnamefont {R.}~\bibnamefont {Yehia}}, \bibinfo {author} {\bibfnamefont {K.}~\bibnamefont {Goodenough}}, \bibinfo {author} {\bibfnamefont {M.}~\bibnamefont {Ruf}}, \bibinfo {author} {\bibfnamefont {P.~C.}\ \bibnamefont {Humphreys}}, \bibinfo {author} {\bibfnamefont {R.}~\bibnamefont {Hanson}}, \bibinfo {author} {\bibfnamefont {S.}~\bibnamefont {Wehner}},\ and\ \bibinfo {author} {\bibfnamefont {D.}~\bibnamefont {Elkouss}},\ }\href@noop {} {\bibfield  {journal} {\bibinfo  {journal} {Phys. Rev. A}\ }\textbf {\bibinfo {volume} {99}},\ \bibinfo {pages} {052330} (\bibinfo {year} {2019})}\BibitemShut {NoStop}%
\bibitem [{\citenamefont {Lee}\ \emph {et~al.}(2022)\citenamefont {Lee}, \citenamefont {Bersin}, \citenamefont {Dahlberg}, \citenamefont {Wehner},\ and\ \citenamefont {Englund}}]{Lee2022}%
  \BibitemOpen
  \bibfield  {author} {\bibinfo {author} {\bibfnamefont {Y.}~\bibnamefont {Lee}}, \bibinfo {author} {\bibfnamefont {E.}~\bibnamefont {Bersin}}, \bibinfo {author} {\bibfnamefont {A.}~\bibnamefont {Dahlberg}}, \bibinfo {author} {\bibfnamefont {S.}~\bibnamefont {Wehner}},\ and\ \bibinfo {author} {\bibfnamefont {D.}~\bibnamefont {Englund}},\ }\href@noop {} {\bibfield  {journal} {\bibinfo  {journal} {npj Quantum Inf.}\ }\textbf {\bibinfo {volume} {8}},\ \bibinfo {pages} {75} (\bibinfo {year} {2022})}\BibitemShut {NoStop}%
\bibitem [{\citenamefont {Campbell}\ and\ \citenamefont {Benjamin}(2008)}]{Campbell2008}%
  \BibitemOpen
  \bibfield  {author} {\bibinfo {author} {\bibfnamefont {E.~T.}\ \bibnamefont {Campbell}}\ and\ \bibinfo {author} {\bibfnamefont {S.~C.}\ \bibnamefont {Benjamin}},\ }\href@noop {} {\bibfield  {journal} {\bibinfo  {journal} {Phys. Rev. Lett.}\ }\textbf {\bibinfo {volume} {101}},\ \bibinfo {pages} {130502} (\bibinfo {year} {2008})}\BibitemShut {NoStop}%
\bibitem [{\citenamefont {Jones}\ \emph {et~al.}(2016)\citenamefont {Jones}, \citenamefont {Kim}, \citenamefont {Rakher}, \citenamefont {Kwiat},\ and\ \citenamefont {Ladd}}]{Jones2016}%
  \BibitemOpen
  \bibfield  {author} {\bibinfo {author} {\bibfnamefont {C.}~\bibnamefont {Jones}}, \bibinfo {author} {\bibfnamefont {D.}~\bibnamefont {Kim}}, \bibinfo {author} {\bibfnamefont {M.~T.}\ \bibnamefont {Rakher}}, \bibinfo {author} {\bibfnamefont {P.~G.}\ \bibnamefont {Kwiat}},\ and\ \bibinfo {author} {\bibfnamefont {T.~D.}\ \bibnamefont {Ladd}},\ }\href@noop {} {\bibfield  {journal} {\bibinfo  {journal} {New J. Phys.}\ }\textbf {\bibinfo {volume} {18}},\ \bibinfo {pages} {083015} (\bibinfo {year} {2016})}\BibitemShut {NoStop}%
\bibitem [{\citenamefont {Dehaene}\ \emph {et~al.}(2003)\citenamefont {Dehaene}, \citenamefont {Van~den Nest}, \citenamefont {De~Moor},\ and\ \citenamefont {Verstraete}}]{Dehaene2003}%
  \BibitemOpen
  \bibfield  {author} {\bibinfo {author} {\bibfnamefont {J.}~\bibnamefont {Dehaene}}, \bibinfo {author} {\bibfnamefont {M.}~\bibnamefont {Van~den Nest}}, \bibinfo {author} {\bibfnamefont {B.}~\bibnamefont {De~Moor}},\ and\ \bibinfo {author} {\bibfnamefont {F.}~\bibnamefont {Verstraete}},\ }\href@noop {} {\bibfield  {journal} {\bibinfo  {journal} {Phys. Rev. A}\ }\textbf {\bibinfo {volume} {67}},\ \bibinfo {pages} {022310} (\bibinfo {year} {2003})}\BibitemShut {NoStop}%
\bibitem [{\citenamefont {D{\"u}r}\ and\ \citenamefont {Briegel}(2007)}]{Dur2007}%
  \BibitemOpen
  \bibfield  {author} {\bibinfo {author} {\bibfnamefont {W.}~\bibnamefont {D{\"u}r}}\ and\ \bibinfo {author} {\bibfnamefont {H.~J.}\ \bibnamefont {Briegel}},\ }\href@noop {} {\bibfield  {journal} {\bibinfo  {journal} {Rep. Prog. Phys.}\ }\textbf {\bibinfo {volume} {70}},\ \bibinfo {pages} {1381} (\bibinfo {year} {2007})}\BibitemShut {NoStop}%
\bibitem [{\citenamefont {Ruan}\ \emph {et~al.}(2021)\citenamefont {Ruan}, \citenamefont {Kirby}, \citenamefont {Brodsky},\ and\ \citenamefont {Win}}]{Ruan2021}%
  \BibitemOpen
  \bibfield  {author} {\bibinfo {author} {\bibfnamefont {L.}~\bibnamefont {Ruan}}, \bibinfo {author} {\bibfnamefont {B.~T.}\ \bibnamefont {Kirby}}, \bibinfo {author} {\bibfnamefont {M.}~\bibnamefont {Brodsky}},\ and\ \bibinfo {author} {\bibfnamefont {M.~Z.}\ \bibnamefont {Win}},\ }\href@noop {} {\bibfield  {journal} {\bibinfo  {journal} {Phys. Rev. A}\ }\textbf {\bibinfo {volume} {103}},\ \bibinfo {pages} {032425} (\bibinfo {year} {2021})}\BibitemShut {NoStop}%
\bibitem [{\citenamefont {Van~Meter}\ \emph {et~al.}(2008)\citenamefont {Van~Meter}, \citenamefont {Ladd}, \citenamefont {Munro},\ and\ \citenamefont {Nemoto}}]{VanMeter2008}%
  \BibitemOpen
  \bibfield  {author} {\bibinfo {author} {\bibfnamefont {R.}~\bibnamefont {Van~Meter}}, \bibinfo {author} {\bibfnamefont {T.~D.}\ \bibnamefont {Ladd}}, \bibinfo {author} {\bibfnamefont {W.~J.}\ \bibnamefont {Munro}},\ and\ \bibinfo {author} {\bibfnamefont {K.}~\bibnamefont {Nemoto}},\ }\href@noop {} {\bibfield  {journal} {\bibinfo  {journal} {IEEE/ACM Transactions On Networking}\ }\textbf {\bibinfo {volume} {17}},\ \bibinfo {pages} {1002} (\bibinfo {year} {2008})}\BibitemShut {NoStop}%
\bibitem [{\citenamefont {Briegel}\ \emph {et~al.}(1998)\citenamefont {Briegel}, \citenamefont {D{\"u}r}, \citenamefont {Cirac},\ and\ \citenamefont {Zoller}}]{Briegel1998}%
  \BibitemOpen
  \bibfield  {author} {\bibinfo {author} {\bibfnamefont {H.-J.}\ \bibnamefont {Briegel}}, \bibinfo {author} {\bibfnamefont {W.}~\bibnamefont {D{\"u}r}}, \bibinfo {author} {\bibfnamefont {J.~I.}\ \bibnamefont {Cirac}},\ and\ \bibinfo {author} {\bibfnamefont {P.}~\bibnamefont {Zoller}},\ }\href@noop {} {\bibfield  {journal} {\bibinfo  {journal} {Phys. Rev. Lett.}\ }\textbf {\bibinfo {volume} {81}},\ \bibinfo {pages} {5932} (\bibinfo {year} {1998})}\BibitemShut {NoStop}%
\bibitem [{\citenamefont {Rozpedek}\ \emph {et~al.}(2018)\citenamefont {Rozpedek}, \citenamefont {Schiet}, \citenamefont {Elkouss}, \citenamefont {Doherty}, \citenamefont {Wehner} \emph {et~al.}}]{Rozpedek2018a}%
  \BibitemOpen
  \bibfield  {author} {\bibinfo {author} {\bibfnamefont {F.}~\bibnamefont {Rozpedek}}, \bibinfo {author} {\bibfnamefont {T.}~\bibnamefont {Schiet}}, \bibinfo {author} {\bibfnamefont {D.}~\bibnamefont {Elkouss}}, \bibinfo {author} {\bibfnamefont {A.~C.}\ \bibnamefont {Doherty}}, \bibinfo {author} {\bibfnamefont {S.}~\bibnamefont {Wehner}}, \emph {et~al.},\ }\href@noop {} {\bibfield  {journal} {\bibinfo  {journal} {Phys. Rev. A}\ }\textbf {\bibinfo {volume} {97}},\ \bibinfo {pages} {062333} (\bibinfo {year} {2018})}\BibitemShut {NoStop}%
\bibitem [{\citenamefont {Krastanov}\ \emph {et~al.}(2019)\citenamefont {Krastanov}, \citenamefont {Albert},\ and\ \citenamefont {Jiang}}]{Krastanov2019}%
  \BibitemOpen
  \bibfield  {author} {\bibinfo {author} {\bibfnamefont {S.}~\bibnamefont {Krastanov}}, \bibinfo {author} {\bibfnamefont {V.~V.}\ \bibnamefont {Albert}},\ and\ \bibinfo {author} {\bibfnamefont {L.}~\bibnamefont {Jiang}},\ }\href@noop {} {\bibfield  {journal} {\bibinfo  {journal} {Quantum}\ }\textbf {\bibinfo {volume} {3}},\ \bibinfo {pages} {123} (\bibinfo {year} {2019})}\BibitemShut {NoStop}%
\bibitem [{\citenamefont {Marler}\ and\ \citenamefont {Arora}(2004)}]{Marler2004}%
  \BibitemOpen
  \bibfield  {author} {\bibinfo {author} {\bibfnamefont {R.~T.}\ \bibnamefont {Marler}}\ and\ \bibinfo {author} {\bibfnamefont {J.~S.}\ \bibnamefont {Arora}},\ }\href@noop {} {\bibfield  {journal} {\bibinfo  {journal} {Struct. Multidiscip. Optim.}\ }\textbf {\bibinfo {volume} {26}},\ \bibinfo {pages} {369} (\bibinfo {year} {2004})}\BibitemShut {NoStop}%
\bibitem [{\citenamefont {Werner}(1989)}]{Werner1989}%
  \BibitemOpen
  \bibfield  {author} {\bibinfo {author} {\bibfnamefont {R.~F.}\ \bibnamefont {Werner}},\ }\href@noop {} {\bibfield  {journal} {\bibinfo  {journal} {Phys. Rev. A}\ }\textbf {\bibinfo {volume} {40}},\ \bibinfo {pages} {4277} (\bibinfo {year} {1989})}\BibitemShut {NoStop}%
\bibitem [{\citenamefont {Sigman}\ and\ \citenamefont {Wolff}(1993)}]{Sigman1993}%
  \BibitemOpen
  \bibfield  {author} {\bibinfo {author} {\bibfnamefont {K.}~\bibnamefont {Sigman}}\ and\ \bibinfo {author} {\bibfnamefont {R.~W.}\ \bibnamefont {Wolff}},\ }\href@noop {} {\bibfield  {journal} {\bibinfo  {journal} {SIAM review}\ }\textbf {\bibinfo {volume} {35}},\ \bibinfo {pages} {269} (\bibinfo {year} {1993})}\BibitemShut {NoStop}%
\bibitem [{\citenamefont {Makowski}\ \emph {et~al.}(1989)\citenamefont {Makowski}, \citenamefont {Melamed},\ and\ \citenamefont {Whitt}}]{Makowski1989}%
  \BibitemOpen
  \bibfield  {author} {\bibinfo {author} {\bibfnamefont {A.}~\bibnamefont {Makowski}}, \bibinfo {author} {\bibfnamefont {B.}~\bibnamefont {Melamed}},\ and\ \bibinfo {author} {\bibfnamefont {W.}~\bibnamefont {Whitt}},\ }in\ \href@noop {} {\emph {\bibinfo {booktitle} {Proceedings of the 28th IEEE Conference on Decision and Control,}}}\ (\bibinfo {organization} {IEEE},\ \bibinfo {year} {1989})\ pp.\ \bibinfo {pages} {1084--1086}\BibitemShut {NoStop}%
\bibitem [{\citenamefont {Vlasiou}(2011)}]{Vlasiou2011}%
  \BibitemOpen
  \bibfield  {author} {\bibinfo {author} {\bibfnamefont {M.}~\bibnamefont {Vlasiou}},\ }\bibinfo {title} {Regenerative processes},\ in\ \href {https://doi.org/https://doi.org/10.1002/9780470400531.eorms0713} {\emph {\bibinfo {booktitle} {Wiley Encyclopedia of Operations Research and Management Science}}}\ (\bibinfo  {publisher} {John Wiley \& Sons, Ltd},\ \bibinfo {year} {2011})\BibitemShut {NoStop}%
\bibitem [{\citenamefont {Grimmett}\ and\ \citenamefont {Stirzaker}(2020)}]{Grimmett2020}%
  \BibitemOpen
  \bibfield  {author} {\bibinfo {author} {\bibfnamefont {G.}~\bibnamefont {Grimmett}}\ and\ \bibinfo {author} {\bibfnamefont {D.}~\bibnamefont {Stirzaker}},\ }\href@noop {} {\emph {\bibinfo {title} {Probability and random processes}}}\ (\bibinfo  {publisher} {Oxford university press},\ \bibinfo {year} {2020})\BibitemShut {NoStop}%
\bibitem [{\citenamefont {Laflamme}\ \emph {et~al.}(1996)\citenamefont {Laflamme}, \citenamefont {Miquel}, \citenamefont {Paz},\ and\ \citenamefont {Zurek}}]{Laflamme1996}%
  \BibitemOpen
  \bibfield  {author} {\bibinfo {author} {\bibfnamefont {R.}~\bibnamefont {Laflamme}}, \bibinfo {author} {\bibfnamefont {C.}~\bibnamefont {Miquel}}, \bibinfo {author} {\bibfnamefont {J.~P.}\ \bibnamefont {Paz}},\ and\ \bibinfo {author} {\bibfnamefont {W.~H.}\ \bibnamefont {Zurek}},\ }\href@noop {} {\bibfield  {journal} {\bibinfo  {journal} {Phys. Rev. Letters}\ }\textbf {\bibinfo {volume} {77}},\ \bibinfo {pages} {198} (\bibinfo {year} {1996})}\BibitemShut {NoStop}%
\bibitem [{\citenamefont {Stephens}\ \emph {et~al.}(2013)\citenamefont {Stephens}, \citenamefont {Huang}, \citenamefont {Nemoto},\ and\ \citenamefont {Munro}}]{Stephens2013}%
  \BibitemOpen
  \bibfield  {author} {\bibinfo {author} {\bibfnamefont {A.~M.}\ \bibnamefont {Stephens}}, \bibinfo {author} {\bibfnamefont {J.}~\bibnamefont {Huang}}, \bibinfo {author} {\bibfnamefont {K.}~\bibnamefont {Nemoto}},\ and\ \bibinfo {author} {\bibfnamefont {W.~J.}\ \bibnamefont {Munro}},\ }\href@noop {} {\bibfield  {journal} {\bibinfo  {journal} {Phys. Rev. A}\ }\textbf {\bibinfo {volume} {87}},\ \bibinfo {pages} {052333} (\bibinfo {year} {2013})}\BibitemShut {NoStop}%
\bibitem [{\citenamefont {Hermans}\ \emph {et~al.}(2023)\citenamefont {Hermans}, \citenamefont {Pompili}, \citenamefont {dos Santos~Martins}, \citenamefont {Rodriguez-Pardo~Montblanch}, \citenamefont {Beukers}, \citenamefont {Baier}, \citenamefont {Borregaard},\ and\ \citenamefont {Hanson}}]{Hermans2023}%
  \BibitemOpen
  \bibfield  {author} {\bibinfo {author} {\bibfnamefont {S.}~\bibnamefont {Hermans}}, \bibinfo {author} {\bibfnamefont {M.}~\bibnamefont {Pompili}}, \bibinfo {author} {\bibfnamefont {L.}~\bibnamefont {dos Santos~Martins}}, \bibinfo {author} {\bibfnamefont {A.}~\bibnamefont {Rodriguez-Pardo~Montblanch}}, \bibinfo {author} {\bibfnamefont {H.}~\bibnamefont {Beukers}}, \bibinfo {author} {\bibfnamefont {S.}~\bibnamefont {Baier}}, \bibinfo {author} {\bibfnamefont {J.}~\bibnamefont {Borregaard}},\ and\ \bibinfo {author} {\bibfnamefont {R.}~\bibnamefont {Hanson}},\ }\href@noop {} {\bibfield  {journal} {\bibinfo  {journal} {New J. Phys.}\ ,\ \bibinfo {pages} {013011}} (\bibinfo {year} {2023})}\BibitemShut {NoStop}%
\end{thebibliography}%

\section*{Code availability}
The code used to perform the analysis and generate all the plots shown in this paper can be found in the following GitHub repository: \href{https://github.com/AlvaroGI/buffering-1GnB}{https://github.com/AlvaroGI/buffering-1GnB}.
This repository also includes a discrete-event simulator of a 1G$n$B system that we used to validate our analytical results.

\section{Author Contributions}
BD and ÁGI conceived and defined the project.
BD and SK proved Theorems \ref{theorem.A} and \ref{theorem.F}.
ÁGI and BD proved Propositions \ref{prop.dAdq} and \ref{prop.dFdq}.
ÁGI carried out the analysis from Sections \ref{sec.results} and \ref{sec.choosingpolicy}, and coded the discrete-event simulation (used to validate analytical results).
ÁGI and BD wrote this manuscript.
SW supervised the project.

\section{Acknowledgements}
We thank C. Cicconetti, P. Kaku, and J. van Dam for discussions and feedback.
ÁGI acknowledges financial support from the Netherlands Organisation for Scientific Research (NWO/OCW), as part of the Frontiers of Nanoscience program.
BD acknowledges financial support from a KNAW Ammodo Award (SW).
SW acknowledges support from an NWO VICI grant.

%%%%%%%%%%%%%%%%%%%%%%%%%%%%%%%%%%%%%
%%%%%%%%%%% APPENDICES %%%%%%%%%%%%
%%%%%%%%%%%%%%%%%%%%%%%%%%%%%%%%%%%%%
\clearpage
\appendix
\counterwithin{theorem}{section}
\counterwithin{lemma}{section}
\counterwithin{proposition}{section}

\onecolumngrid

\clearpage
\section{A note on the viewpoint}
\label{app.viewpoint}
In this appendix, we provide three further ways to compute the performance metrics $A$ and $\overline{F}$.
The initial (and most natural) definitions of the performance metrics (see Definitions~\ref{def.A} and \ref{def.F}) consist in averages from the viewpoint of the network user, who consumed the links.
In Lemma~\ref{lem.metrics-single-cycle}, we show that the averaging may only be done over a single cycle of the renewal process.
In Lemma~\ref{lem.metrics-limiting-values}, we show that the performance metrics can also be computed as limiting values when time goes to infinity. Lastly, in Lemma~\ref{lem.time-average}, it is shown that one may compute the metrics by averaging over time, regardless of consumption arrival times.

We denote the arrival time of the $j$-th consumption request as $T_{\mathrm{con}}^{(j)}$. 
From now on, we write 
\begin{equation}
    \mathbbm{1}_{\mathrm{l.e.}}(F) \equiv \mathbbm{1}_{\text{link exists}}(F)  = 
    \begin{cases}
        1 \text{ if } F>0, \\ 
        0 \text{ if } F = 0.
    \end{cases}
\end{equation}
In the following, we let $\mathbb{N}_0$ denote the natural numbers containing zero, and $\mathbb{N}=\mathbb{N}_0\setminus \{0\}$. Recall that $F=\{F(t),t\in \mathbb{N}_0 \}$ is a discrete-time stochastic process. The value $F(t)$ is defined to be the fidelity at the \textit{beginning} of the time step $[t,t+1)$. Then, since consumption removes the link from the G memory, at each consumption time we have $F(T_{\mathrm{con}}^{(j)})=0$. However, the consumed fidelity at this time depends on the value of the fidelity at time $T^{(j)}_{\mathrm{con}}-1$. We therefore introduce some new notation to more easily treat this issue. 

In order to do this, we firstly note that associated with $F$ there is an equivalent continuous-time stochastic process $\{F_{\mathrm{cont}}(s):s\geq 0 \}$ that is obtained from $F$ with the following procedure: given $t\in \mathbb{N}$,
\begin{enumerate}[label=(\roman*)]
\item if $F(t)>0$, then for $s\in [t,t+1)$, $F_{\mathrm{cont}}(s)$ may be deduced by applying decoherence (\ref{eq.fidelity_decay}) to $F(t)$;
\item if $F(t)=0$, then $F_{\mathrm{cont}}(s) = 0$ for $s\in [t,t+1)$. 
\end{enumerate}
Conversely, $F$ may be obtained from $F_{\mathrm{cont}}$ by  taking its values at integer times.

From $F_{\mathrm{cont}}$, for $t\in \mathbb{N}$, we define another discrete time process $F^-$,
\begin{equation}
    F^{-} \coloneq \{ F_{\mathrm{cont}}(t^-): t\in \mathbb{N} \},
\end{equation}
where $t^-$ denotes taking the left-hand limit. In particular, the consumed fidelity $F^-(t)$ takes the value of the fidelity at the \textit{end} of the time step $[t-1,t)$.
The values of $F^-$ may also be deduced directly from $F$ as
\begin{equation}
    F^-(t) =  \begin{cases}
        e^{-\Gamma}\left(F(t-1)-\frac{1}{4}\right) + \frac{1}{4}, \text{ if } F(t-1)>0, \\ 
        0, \text{ if } F(t-1) = 0.
    \end{cases}
    \label{eq.def-F(t^-)}
\end{equation}
We note that the evolution of $\{F^-(t) ,t\in \mathbb{N} \}$ may be deduced directly from $\{F(t) ,t\in \mathbb{N} \}$ via (\ref{eq.def-F(t^-)}), and vice-versa. The value $F^-(t)$ may be interpreted as the state of the system `just before' time $t$, and $F(t)$ the state `just after'. Each completely captures the behaviour of the 1G$n$B system.
%See Figure \ref{} for a depiction of $F^-$, $F$, and $F_{\mathrm{cont}}$.

We then restate the original definitions \ref{def.A} and \ref{def.F} of availability, $A$, and average consumed fidelity, $\overline F$, below.
\begin{comment}
f a consumption request arrives at time $t$, then if no link is in the G memory at time $t-1$ (equivalently, $F(t-1) = 0$), the request is not served and the system proceeds with the entanglement generation attempt. If there is a link in the G memory at time $t-1$ (equivalently, $F(t-1) > 0$), then the link is immediately consumed and removed from the G memory (equivalently, $F(t)=0$). 
\end{comment}
\begin{definition}[Performance metrics, viewpoint of network user]
We have
\begin{equation}
    A = \lim_{m\rightarrow \infty}\frac{1}{m} \sum_{j = 1}^{m} \mathbbm{1}_{\mathrm{l.e.}}\left(F^-(T_{\mathrm{con}}^{(j)})\right),
    \label{eq.A-def-user-viewpoint}
\end{equation}
and 
\begin{equation}
    \overline{F} = \lim_{m\rightarrow \infty} \frac{\sum_{j = 1}^{m} F^-(T_{\mathrm{con}}^{(j)})\cdot \mathbbm{1}_{\mathrm{l.e.}}\left( F^-(T_{\mathrm{con}}^{(j)})\right)}{\sum_{j = 1}^{m} \mathbbm{1}_{\mathrm{l.e.}}\left(F^-(T_{\mathrm{con}}^{(j)})\right)}.
    \label{eq.Fbar-def-user-viewpoint}
\end{equation}
\label{def.metrics-viewpoint-user}
\end{definition}
%Recall that the arrival of consumption requests in each time step is assumed to be a Bernoulli process. In the following results, we see that the above metrics are equal to the long-time averages over the whole process. This follows from a version of the well-known PASTA property (Poisson Arrivals See Time Averages) in queuing theory~\cite{Makowski1989}.

We now present a second way to compute the performance metrics, which is the form that is used to derive the solutions for $A$ and $\overline{F}$ in Theorems \ref{theorem.A} and \ref{theorem.F} (see Appendix~\ref{app.solutions}). To show this result, we use the fact that $F(t)$ is a regenerative process. Informally, every time the link in the G memory is removed from the system, the process `starts again', in the sense that the stochastic properties from that point onwards are the same as when starting from any other time when the G memory is empty. This stems from the fact that entangled link generation and consumption request arrivals are assumed to be Markovian. 

\begin{definition}[Regenerative process, informal] A regenerative process $\{X(t),t \geq 0 \}$ is a stochastic process with the following properties: there exists a random variable $V_1>0$ such that
\begin{enumerate}[label=(\roman*)]
    \item $\{X(t+V_1),t\geq 0\}$ is independent of $\{X(t),t \leq V_1\}$ and $V_1$;
    \item $\{X(t+V_1),t\geq 0\}$ is stochastically equivalent to $\{ X(t),t\geq 0\}$ (i.e. these two processes have the same joint distributions).
\end{enumerate}
\end{definition}
For a formal definition of a regenerative process, see e.g. \cite{Sigman1993}. If the process is regenerative, it may also be shown that there is a sequence of regeneration cycles $V_0=0$, $\{V_k\}$ such that the sequence regenerates at each cycle, i.e. $\{X(t),t\geq 0\}$ and $\{X(t+V_k),t\geq 0\}$ are stochastically equivalent. 

We now show that our process $F$ is regenerative. Let us assume the system starts when a new link is freshly generated and moved to the G memory, such that $F(0) = F_{\mathrm{new}}$. The system then evolves as follows: the link in the G memory may undergo some purification rounds, between which it is subject to decoherence, and then is eventually removed from the G memory after time $T_{\mathrm{occ}}^{(1)}$ due to either purification failure or consumption. The time $T_{\mathrm{occ}}^{(1)}$ is the time during which the G memory is occupied. In particular,
\begin{equation}
    T_{\mathrm{occ}}^{(1)} \coloneqq \min \{t:F(t) = 0\}. 
\end{equation}
After the link is removed, the system will then attempt entanglement generation until a successful generation. Let the time from which the G memory is emptied until a new link is produced be $T_{\mathrm{gen}}^{(1)}$. By the assumption that entanglement generation attempts are independent and Bernoulli, $T_{\mathrm{gen}}^{(1)} \sim \mathrm{Geo}(1-(1-p_{\mathrm{gen}})^n)$. When a fresh link is generated at time $t=T_{\mathrm{occ}}^{(1)}+T_{\mathrm{gen}}^{(1)}$, we have $F(T_{\mathrm{occ}}^{(1)}+T_{\mathrm{gen}}^{(1)}) = F_{\mathrm{new}}$ and, from this time on, the process behaves equivalently to how it did from time $t=0$. Letting $V_1 = T_{\mathrm{occ}}^{(1)} + T_{\mathrm{gen}}^{(1)}$, we see that $F(t)$ is regenerative. All regeneration cycles $\{V_k\}$ may each be split into two phases: we have $V_k = T_{\mathrm{occ}}^{(k)} + T_{\mathrm{gen}}^{(k)}$, where $T_{\mathrm{occ}}^{(k)}$ is the time during which the memory is occupied, and $T_{\mathrm{gen}}^{(k)}$ is the time during which the memory is empty and entanglement generation is being attempted. We note that since $F^-$ is in one-to-one correspondence with $F$ via (\ref{eq.def-F(t^-)}), then $F^-$ is also regenerative with the same cycle lengths. 

For the following results, we note two important properties of the process $\{V_k\}$. Firstly, the mean cycle length $\mathbb{E}[V_1] = \mathbb{E}[T_{\mathrm{occ}}^{(1)}] + \mathbb{E}[T_{\mathrm{gen}}^{(1)}]$ is finite: this may be seen by the fact that $T_{\mathrm{gen}}^{(1)}$ is geometrically distributed (and therefore $\mathbb{E}[T_{\mathrm{gen}}^{(1)}]<\infty$) and that $T_{\mathrm{occ}}^{(1)}$ is bounded above by the time until the next consumption request, which is geometrically distributed, and so $\mathbb{E}[T_{\mathrm{occ}}^{(1)}]\leq \mathbb{E}[T_{\mathrm{con}}^{(1)}] < \infty $. The second important property is that the $\{V_k\}$ are \textit{aperiodic}, which means that $V_1$ takes values in a set of integers that have greatest common denominator equal to one. Again, this may be seen by the fact that consumption and entanglement generation are assumed to be geometric. If $p_{\mathrm{gen}}<1$, the value of $V_1$ has a non-zero probability of taking any value in $\mathbb{N}\setminus\{1\}$ and therefore satisfies this property. The same holds if $p_{\mathrm{gen}}=1$, and there is a non-zero probability of either no purification or successful purification. 
The cases where the $\{V_k\}$ are periodic may be accounted for separately:
\begin{enumerate}[label=(\Alph*)]
    \item If $p_{\mathrm{gen}} = 1$ and $p_{\mathrm{con}} = 1$, a link will deterministically be generated when in the empty state, and deterministically consumed in the following time step. The fidelity $F(t)$ then deterministically alternates between $0$ and $F_{\mathrm{new}}$, and the cycle length is always two. We therefore have 
    \begin{equation}
     A=\frac{1}{2},\;\;\;\;   \overline{F}= e^{-\Gamma}\left(F_{\mathrm{new}}-\frac{1}{4}\right) + \frac{1}{4}.
     \label{eq.metrics-values-edge-cases}
    \end{equation}
    \item If $p_{\mathrm{gen}} = 1$, $q=1$ and $c_k = d_k = 0$, then we have deterministic link generation, and the system always decides to purify. However, purification always fails. The fidelity then again deterministically alternates between $0$ and $F_{\mathrm{new}}$, and the cycle length is two. We note that even if purification is always attempted and always fails, then if a consumption request arrives, this will take priority over purification and the link will be consumed with fidelity $e^{-\Gamma}\left(F_{\mathrm{new}}-\frac{1}{4}\right) + \frac{1}{4}$. Then, $\overline F$ will also take this value. Moreover, by applying the PASTA property in discrete time \cite{Makowski1989}, we have $A=1/2$. Our metrics then take the values (\ref{eq.metrics-values-edge-cases}), as in case (A).
\end{enumerate}
We note that our formulae, as given in Theorems \ref{theorem.A} and \ref{theorem.F}, still hold for the above cases. The solutions for edge case (A) are obtained by inputting $p_{\mathrm{gen}} = 1$ and $p_{\mathrm{con}} = 1$. Edge case (B) can be dealt with in the same way: take $p_{\mathrm{gen}} = 1$, $q=1$ and the limit $c_k,d_k\rightarrow 0$. Note that the jump function (\ref{eq.Jk_def}) must still be well-defined, and so necessarily we must also take $a_k,b_k\rightarrow 0$. We then obtain (\ref{eq.metrics-values-edge-cases}). Although the proof in the general case may not be immediately applied in these cases, our formula still holds.

\begin{lemma}[Performance metrics, single cycle] Suppose that the 1G$n$B system parameters are not in edge cases (A) or (B). The performance metrics in Definition \ref{def.metrics-viewpoint-user} may be written in terms of the properties of a single cycle:
\begin{equation}
    A = \frac{\mathbb{E}[T_{\mathrm{occ}}^{(1)}]}{\mathbb{E}[T_{\mathrm{occ}}^{(1)}] + \mathbb{E}[T_{\mathrm{gen}}^{(1)}]} \text{    a.s.} \label{eq.A-single-cycle}
\end{equation}
and  
\begin{equation}
    \overline{F} = \mathbb{E}[F^-(T_{\mathrm{occ}}^{(1)})|C_1]\text{    a.s.}
\end{equation}
where $C_1$ is the event where the first link is removed due to consumption (and not failed purification), or equivalently $C_1 \equiv \{T_{\mathrm{occ}}^{(1)} = T_{\mathrm{con}}^{(1)}$\}.
\label{lem.metrics-single-cycle}
\end{lemma}
\begin{proof}
Let $F^-_{\infty}$ be a random variable with distribution given by
\begin{equation}
   \mathrm{P}\left(F^-_{\infty}\in B\right) =  \lim_{t\rightarrow \infty } \frac{1}{t} \sum_{s=1}^t \mathbb{1}_{B} \left(F^-(s)\right).
   \label{eq.def_F^-_infty}
\end{equation}
Then, as $F^-$ is a regenerative process with finite mean and aperiodic cycle length, by e.g. part $(a)$ of Theorem 1 from \cite{Vlasiou2011}, the above quantity exists and may be computed in terms of the properties of a single cycle as
\begin{equation}
    \mathrm{P}\left(F^-_{\infty}\in B\right) = \frac{1}{\mathrm{E}[V_1]} \mathbb{E}\left[\sum_{s=1}^{V_1} \mathbb{1}_{B} \left(F^-(s)\right) \right].
\end{equation}
Letting $B$ be the event where a link is present in the G memory, we then see that 
\begin{align}
    \mathrm{P}\left(F^-_{\infty}>0 \right) &= \frac{1}{\mathrm{E}[V_1]} \mathbb{E}\left[\sum_{s=1}^{V_1} \mathbb{1}_{\mathrm{l.e.}} \left(F^-(s)\right) \right] \\ &= \frac{1}{\mathbb{E}[T_{\mathrm{occ}}^{(1)}] + \mathbb{E}[T_{\mathrm{gen}}^{(1)}]}\cdot \mathbb{E}\left[ T_{\mathrm{occ}}^{(1)}\right].
    \label{eq.limP(F_infty)-one-cycle}
\end{align}
We now show that the above expression is equal to $A$. Since the interarrival times of consumption requests are i.i.d. and follow a geometric distribution, we make use of the PASTA property in discrete time \cite{Makowski1989} to see that the availability from the point of view of the consumer in Definition \ref{def.A} is equal to the time average as given above, i.e.
\begin{equation}
    A = \lim_{t\rightarrow \infty } \frac{1}{t} \sum_{s=1}^t \mathbb{1}_{\mathrm{l.e.}} \left(F^-(s)\right) = \mathrm{P}(F^-_{\infty} > 0), \text{    a.s.}
    \label{eqn.A-with-PASTA}
\end{equation}
Then, (\ref{eq.A-single-cycle}) is shown by combining (\ref{eqn.A-with-PASTA}) with (\ref{eq.limP(F_infty)-one-cycle}).

We now show the identity for $\overline{F}$. For this, we also use the regerative property. We define $W_0=0$ and $W_k$ to be the time at which the $k$-th cycle ends,
\begin{equation}
    W_k \coloneqq \sum_{j=1}^k V_j.
\end{equation}
Then, the sequence of times at which the link is removed from the G memory is 
\begin{equation}
    \{W_{k-1} + T_{\mathrm{occ}}^{(k)}\}_{k\geq 1}.
    \label{eq.W-sequence}
\end{equation}
We then define the subsequence
\begin{equation}
    \{W_{i_k-1} + T_{\mathrm{occ}}^{(i_k)}\}_{k\geq 1}
    \label{eq.W-subsequence-consumption}
\end{equation}
to be the times at which link removal is due to consumption (and not purification failure). We recall that in our model, when a consumption request arrives, it immediately removes the link from the G memory. Then, (\ref{eq.W-subsequence-consumption}) are precisely the times at which consumption requests arrive to find a link in the G memory. In particular,
\begin{equation}
    \{W_{i_k-1} + T_{\mathrm{occ}}^{(i_k)}\}_{k\geq 1} = \left\{T_{\mathrm{con}}^{(k)}:F^-\!\left(T_{\mathrm{con}}^{(k)} \right)>0\right\}_{k\geq 1},
    \label{eq.equality-seq-subseq}
\end{equation}
recalling that $\{T_{\mathrm{con}}^{(k)}\}$ is the sequence of arrival times for consumption requests.
Recalling Definition (\ref{def.metrics-viewpoint-user}) of $\overline{F}$, we then see that 
%\bd{Sounak's comments here for $m'$ random variable in SLLN}
\begin{align}
     \overline{F} &= \lim_{m\rightarrow \infty} \frac{\sum_{k = 1}^{m} F^-(T_{\mathrm{con}}^{(k)})\cdot \mathbbm{1}_{\mathrm{l.e.}}\left( F^-(T_{\mathrm{con}}^{(k)})\right)}{\sum_{k = 1}^{m} \mathbbm{1}_{\mathrm{l.e.}}\left( F^-(T_{\mathrm{con}}^{(k)}) \right)} \nonumber
     \\ &= \lim_{m\rightarrow \infty} \frac{\sum_{k = 1}^{M(m)} F^-\left( W_{i_k-1} + T_{\mathrm{occ}}^{(i_k)}\right) }{\sum_{k = 1}^{M(m)} 1}, \label{eq.change-variable-sum-M(m)}
\end{align}
where we have used the identity (\ref{eq.equality-seq-subseq}), and defined $M(m)\leq m$ as
\begin{equation*}
    M(m) = \left| \left\{ T_{\mathrm{con}}^{(k)} : F^-\left(T_{\mathrm{con}}^{(k)} \right)>0, \; k \leq m\right\} \right|.
\end{equation*}
Then, $M(m)$ is the number of consumption requests up to time $T_{\mathrm{con}}^{(m)}$ that arrive when a link is stored in memory. We now show that $ \lim_{m\rightarrow \infty} M(m) =  \infty $ a.s. so that we can apply SLLN to the above expression.
To see this, recall that $\{V_k\}_{k\geq 1}$ are the i.i.d. interarrival times of a renewal process $N(t) = \sup \{ k: W_k\leq t\}$. Since $|\mathbb{E}[V_1]|< \infty$, we have that $\lim_{t\rightarrow \infty}N(t) = \infty $ a.s. (see 10.1.2 of \cite{Grimmett2020}). Within each of these cycles, the link is removed from memory exactly once. The probability that this is due to consumption is bounded below by $p_{\mathrm{con}}>0$, because for each cycle it is possible to consume directly after link generation, which occurs with probability $p_{\mathrm{con}}$. Recalling the sequence of times when the link is removed due to consumption as given in (\ref{eq.equality-seq-subseq}), the number of these events may therefore be bounded below by a subsequence
\begin{equation}
    \{W_{j_k-1} + T_{\mathrm{occ}}^{(j_k)}\}_{k\geq 1} \subseteq \{W_{i_k-1} + T_{\mathrm{occ}}^{(i_k)}\}_{k\geq 1}
    \label{eq.consumed-links-subseq-geometric}
\end{equation}
such that the $j_k - j_{k-1}$ is geometrically distributed with parameter $\eta \geq p_{\mathrm{con}}$. We therefore see that 
\begin{equation}
     \lim_{k\rightarrow \infty} |\{W_{j_k-1} + T_{\mathrm{occ}}^{(j_k)}\}_{k\geq 1}| = \infty \text{ a.s.}
\end{equation}
and therefore by (\ref{eq.consumed-links-subseq-geometric}), the total number of times when the link is consumed diverges to infinity almost surely.
% \begin{equation}
%      \lim_{k\rightarrow \infty} |\{W_{i_k-1} + T_{\mathrm{occ}}^{(i_k)}\}_{k\geq 1}| = \infty \text{ a.s.}
% \end{equation}
From (\ref{eq.change-variable-sum-M(m)}), we then have
\begin{align*}
     \overline{F} &\stackrel{\text{a.s.}}{=} \lim_{M\rightarrow \infty} \frac{1}{M} \sum_{k = 1}^{M} F^-\left( W_{i_k-1} + T_{\mathrm{occ}}^{(i_k)}\right)  \\ &\stackrel{\text{a.s.}}{=} \mathbb{E}\left[ F^-\!(T_{\mathrm{occ}}^{(1)})| C_1 \right].
\end{align*}
\begin{comment}
\begin{align*}
     = \lim_{N\rightarrow \infty} \frac{\sum_{k = 1}^{N} F((W_{k-1} + T_{\mathrm{occ}}^{(k)})^-)\cdot \mathbbm{1}_{C_k}\left(F((W_{k-1} + T_{\mathrm{occ}}^{(k)})^-)\right) }{\sum_{k = 1}^{N} \mathbbm{1}_{C_k}\left(F\left((W_{k-1} + T_{\mathrm{occ}}^{(k)})^-\right)\right)},
\end{align*}
\end{comment}
where we have used the fact that the sequence $\{F^-(W_{i_k - 1} + T_{\mathrm{occ}}^{(i_k)})\}_{k\geq 1}$ is i.i.d. since the process is regenerative, and the strong law of large numbers.
\end{proof}
In the final lemma of this section, we see that the above metrics are equal to the time averages over the whole process. This follows from a version of the well-known PASTA property (Poisson Arrivals See Time Averages) in queuing theory~\cite{Makowski1989}, which we can employ because the arrival of consumption requests in each time step is assumed to be a Bernoulli process.
\begin{lemma}[Performance metrics, time average]
Suppose that the 1G$n$B system parameters are not in edge cases (A) or (B). The performance metrics in Definition \ref{def.metrics-viewpoint-user} may be computed using an average over time, i.e.
\begin{align}
    A &= \lim_{t\rightarrow \infty}\frac{1}{t} \sum_{s = 1}^{t} \mathbbm{1}_{\mathrm{l.e.}}\left(F^-(s)\right), \label{eq.A-time-average}
\end{align}
and
\begin{align}
    \overline{F} &= \lim_{t\rightarrow \infty} \frac{\sum_{s = 1}^t F^-(s)\cdot \mathbbm{1}_{\mathrm{l.e.}}\left(F^-(s)\right)}{\sum_{s = 1}^t \mathbbm{1}_{\mathrm{l.e.}}\left(F^-(s)\right)} \label{eq.Fbar-time-average}
\end{align}
\label{lem.time-average}
\end{lemma}
\begin{proof}
    The identity for $A$ is a direct application of the PASTA property in discrete time \cite{Makowski1989}, which we also saw in the proof of Lemma \ref{lem.first-cycle}.
    
    For the second equality, from (\ref{eq.Fbar-def-user-viewpoint}) we firstly rewrite $\overline F$ as
\begin{align}
    \overline{F} &= \lim_{m\rightarrow \infty} \frac{ \frac{1}{m}\sum_{j = 1}^{m} F^-\!\left(T_{\mathrm{con}}^{(j)}\right)}{\frac{1}{m} \sum_{j = 1}^{m} \mathbbm{1}_{\mathrm{l.e.}}\!\left(F^-\left(T_{\mathrm{con}}^{(j)}\right)\right)} 
    = \frac{\overline{F}_{\mathrm{tot}}}{A},
    \label{eq.F=Ftot/A}
\end{align}
where 
\begin{align*}
    \overline{F}_{\mathrm{tot}} \coloneqq \lim_{m\rightarrow \infty} \frac{1}{m}\sum_{j = 1}^{m} F^-\!\left(T_{\mathrm{con}}^{(j)}\right)
\end{align*}
is the average fidelity seen by users, without conditioning on the fidelity being nonzero. In (\ref{eq.F=Ftot/A}), we have removed the indicator function from the sum in the numerator by recalling that $F^-(T_{\mathrm{con}}^{(j)})=0$ if the $j$-th consumption request does not find a link in memory. Then, since $\overline{F}_{\mathrm{tot}} = \overline{F}\cdot A$ and by Lemma \ref{lem.first-cycle} both $\overline{F}$ and $A$ converge, the PASTA property can be applied and we have that
\begin{equation}
    \overline{F}_{\mathrm{tot}} = \lim_{t\rightarrow \infty} \frac{1}{t}\sum_{s=1}^{t} F^-(s), \text{ a.s.}
    \label{eq.F_tot-with-PASTA}
\end{equation}
Then,
\begin{align}
    \overline{F} &= \frac{\overline{F}_{\mathrm{tot}}}{A}= \frac{ \lim_{t\rightarrow \infty } \frac{1}{t}\sum_{s = 1}^{t} F^-(s)}{ \lim_{t\rightarrow \infty } \frac{1}{t} \sum_{s = 1}^{t} \mathbbm{1}_{\mathrm{l.e.}}\!\big(F^-(s)\big)} \\ &= \lim_{t\rightarrow \infty} \frac{ \sum_{s = 1}^{t} F^-(s)}{  \sum_{s = 1}^{t} \mathbbm{1}_{\mathrm{l.e.}}\!\big(F^-(s)\big)} \\ &= \lim_{t\rightarrow \infty} \frac{ \sum_{s = 1}^{t} F^-(s) \mathbbm{1}_{\mathrm{l.e.}}\!\big(F^-(s)\big)}{  \sum_{s = 1}^{t} \mathbbm{1}_{\mathrm{l.e.}}\!\big(F^-(s)\big)} \text{ a.s.}
\end{align}
\end{proof}

In the following, we show that our performance metrics may be computed as limiting values of properties of $F(t)$. Note that this was the definition used in ref~\cite{Davies2023a}. 
\begin{lemma}[Performance metrics, limiting values]
Suppose that the 1G$n$B system parameters are not in edge cases (A) or (B). Then, our performance metrics may be computed as 
\begin{align}
    A &= \lim_{t\rightarrow \infty} \mathrm{P}\!\left(F^-(t)>0\right)\text{    a.s.} \label{eq.A-limit-distribution} \\ \overline{F} &= \lim_{t\rightarrow\infty } \mathbb{E}\left[ F^-(t)|F^-(t)>0\right]\text{    a.s.}
    \label{eq.Fbar-limit-distribution}
\end{align}
\label{lem.metrics-limiting-values}
\end{lemma}
\begin{proof}
Since $F^-(t)$ is a regenerative process with finite mean 
%(as $\mathbb{E}[V_1] = \mathbb{E}[T_{\mathrm{occ}}^{(1)}] + \mathbb{E}[T_{\mathrm{gen}}^{(1)}] < \infty$) 
and an aperiodic cycle length, it follows that the limiting distribution is well-defined in the following sense. 
As in the proof of Lemma \ref{lem.first-cycle}, we let $F^-_{\infty}$ be a random variable with distribution given by (\ref{eq.def_F^-_infty}). Then, 
% Let $F_{\infty}$ be a random variable with distribution function given by
% \begin{equation}
%     \mathrm{P}(F_{\infty}\in B) = \lim_{t\rightarrow \infty} \frac{1}{t} \sum_{s=0}^t \mathbb{1}_{B}(F^-(s)),
% \end{equation}
% for $x\in \mathbb{R}$ and $B = (-\infty,x]$. 
by e.g. parts (a) and (b) of Theorem 1 of \cite{Vlasiou2011}, we have
\begin{equation}
    \lim_{t\rightarrow \infty } \mathrm{P}\left(F^-(t) \in B \right) = \mathrm{P}(F^-_{\infty}\in B).
\end{equation}
We therefore see that 
\begin{equation}
    \lim_{t\rightarrow \infty } \mathrm{P}\left(F^-(t) > 0\right) = \mathrm{P}(F^-_{\infty}>0) = A,
\end{equation}
where we have used the identity for $A$ which we saw in (\ref{eqn.A-with-PASTA}) in the proof of Lemma \ref{lem.first-cycle}. This shows (\ref{eq.A-limit-distribution}).

To show the identity for $\overline{F}$, we make use of the renewal-reward theorem (see e.g. 10.5.1 of \cite{Grimmett2020}). From the previous discussion, associated with the regenerative process $\{F(t) , t\in \mathbb{N} \}$ with cycle times $\{W_k\}$, there is a renewal process $N(t) = \sup \{ k: W_k\leq t\}$. We then define the reward $\tilde{R}_k$ as the sum of fidelity over the $k$-th cycle,
\begin{equation}
    \tilde{R}_k = \sum_{t= W_{k-1}+1}^{W_k} F^-(t).
\end{equation}
Then, the cumulative reward up to time $t$ is given by 
\begin{align}
    \tilde{C}(t) &= \sum_{s=1}^t F^-(s) \label{eq.C(t)_sum_F^-} \\  &=  \sum_{k=1}^{N(t)} \tilde{R}_k + E(t),
\end{align}
where we have defined
\begin{equation}
    E(t) = \sum_{s=W_{N(t)}+1}^{t} F^-(s)
\end{equation}
to be the remainder of the reward that is not contained in a full cycle. Then, we see that
\begin{align}
    \frac{\tilde{C}(t)}{t} &\leq  \frac{1}{t}\sum_{k=1}^{N(t)+1}\tilde{R}_k \nonumber
    \\ &= \frac{\sum_{k=1}^{N(t)+1}\tilde{R}_k}{N(t)+1}\cdot \frac{N(t)+1}{t}. \label{eq.Ctilde/T-split}
\end{align}
We will now  the strong law of large numbers (SLLN) for both terms in the above product.  In particular, the convergence of $(N(t)+1)/t$ may be seen by noticing that
\begin{equation}
   \frac{\sum_{k=1}^{N(t)} V_k}{N(t)}\cdot \frac{N(t)}{N(t)+1} < \frac{t}{N(t)+1} \leq \frac{\sum_{k=1}^{N(t)+1} V_k}{N(t)+1}
\end{equation}
and using SLLN shows that the upper and lower bound converge to $\mathbb{E}\left[ V_1 \right]$. From (\ref{eq.Ctilde/T-split}), we therefore see that
\begin{align}
    \lim_{t\rightarrow \infty} \frac{\tilde{C}(t)}{t}  &\leq  \frac{\mathbb{E}\left[ \tilde{R}_1 \right]}{\mathbb{E}\left[ V_1 \right]} \\ &= \frac{\mathbb{E}\left[ \sum_{t=1}^{V_1} F^-(t) \right]}{\mathbb{E}\left[ V_1 \right]} \text{  a.s.}
    \label{eq.C(t)/t-UB}
\end{align}
Similarly,
\begin{align}
    \lim_{t\rightarrow \infty} \frac{\tilde{C}(t)}{t} &\geq \lim_{t\rightarrow \infty} \frac{\sum_{k=1}^{N(t)}\tilde{R}_k}{N(t)+1}\cdot \frac{N(t)}{t} \\ &= \frac{\mathbb{E}\left[ \tilde{R}_1 \right]}{\mathbb{E}\left[ V_1 \right]}\text{     a.s.}
    \label{eq.C(t)/t-LB}
\end{align}
Combining (\ref{eq.F_tot-with-PASTA}), (\ref{eq.C(t)_sum_F^-}), (\ref{eq.C(t)/t-UB}) and (\ref{eq.C(t)/t-LB}), we therefore see that 
\begin{equation}
   \overline{F}_{\mathrm{tot}} = \lim_{t\rightarrow \infty} \frac{1}{t} \sum_{s=1}^t F^-(s)   =  \lim_{t\rightarrow \infty} \frac{\tilde{C}(t)}{t} = \frac{\mathbb{E}\left[ \sum_{t=1}^{V_1} F^-(t) \right]}{\mathbb{E}\left[ V_1 \right]}\text{    a.s.}
\end{equation}
Moreover, using part (b) of Theorem 1 from \cite{Vlasiou2011} , we see that 
\begin{equation}
    \lim_{t\rightarrow \infty } \mathbb{E}\left[ F^-(t) \right] = \frac{\mathbb{E}\left[ \sum_{t=1}^{V_1} F^-(t) \right]}{\mathbb{E}\left[ V_1 \right]},
\end{equation}
and therefore $\overline{F}_{\mathrm{tot}} = \lim_{t\rightarrow \infty } \mathbb{E}\left[ F^-(t) \right]$.
\begin{comment}
%OLD PROOF AVERAGE FIDELITY
By e.g. Theorem 1 from \cite{Vlasiou2011}, this distribution may be calculated in terms of a single cycle as 
\begin{equation}
    \mathbb{P}(F_{\infty}\in B) = \frac{1}{\mathbb{E}[V_1]} \mathbb{E}\left[ \sum_{t=0}^{V_1} \mathbb{1}_{B}(F^-(t)) \right].
\end{equation}
Similarly, we have that 
\begin{equation}
    \lim_{t\rightarrow \infty} \mathbb{E}\left[F(t^-)\right] = \frac{1}{\mathbb{E}[V_1]} \cdot \mathbb{E}\left[\sum_{t=0}^{V_1} F(t^-)\right].
\end{equation}
Now, by the renewal-reward theorem (see e.g. 10.5.1 of \cite{Grimmett2020}), we have  
\begin{equation}
    \overline{F}_{\mathrm{tot}} = \lim_{T\rightarrow \infty}\frac{1}{T}\sum_{t=0}^T F(t^-) = \mathbb{E}\left[ F_{\infty} \right]. \label{eq.Ftot-Finf}
\end{equation}
Then, noting that 
\begin{align}
    \mathbb{E}\left[F_{\infty}\right] &= \sum_f f \cdot \mathrm{P}(F_{\infty} = f) \\ &= \sum_f f \cdot \frac{1}{\mathbb{E}[V_1]} \mathbb{E}\left[ \sum_{t=0}^{V_1} \mathbb{1}_{\{F=f\}}(F(t^-)) \right] \\ &= \frac{1}{\mathbb{E}[V_1]} \cdot \mathbb{E}\left[\sum_{t=0}^{V_1} \sum_f f \cdot \mathbb{1}_{\{F=f\}}(F(t^-)) \right] \\ &= \frac{1}{\mathbb{E}[V_1]} \cdot \mathbb{E}\left[\sum_{t=0}^{V_1} F(t^-)\right]= \lim_{t\rightarrow \infty} \mathbb{E}\left[F(t)\right],
\end{align}
and combining with (\ref{eq.Ftot-Finf}), we have 
\begin{equation}
    \lim_{t\rightarrow \infty} \mathbb{E}\left[F(t)\right] = \overline{F}_{\mathrm{tot}}.
\end{equation}
\end{comment}
Then, we have 
\begin{align}
    \lim_{t\rightarrow \infty} \mathbb{E}\left[ F^-(t)|F^-(t)>0\right] &=  \lim_{t\rightarrow \infty} \frac{\mathbb{E}\left[ F^-(t)\mathbb{1}_{\mathrm{l.e.}} \left( F^-(t)\right)\right]}{\mathrm{P}(F^-(t)>0)} \\ &= \lim_{t\rightarrow \infty} \frac{\mathbb{E}\left[ F^-(t)\right]}{\mathrm{P}(F^-(t)>0)} \\ &=  \frac{\overline{F}_{\mathrm{tot}}}{A} = \overline{F}.
\end{align}
\end{proof}

%%%%%%%%%%%%%%%%%%%%%%%%%%%%%%%%%%%%%
%%%%%%%%%%%%%% metrics %%%%%%%%%%%%%%
%%%%%%%%%%%%%%%%%%%%%%%%%%%%%%%%%%%%%

\clearpage
\section{Derivation of formulae for performance metrics}
\label{app.solutions}
In this appendix, we prove Theorems \ref{theorem.A} and \ref{theorem.F}, which contain the formulae for the availability and the average consumed fidelity of the 1G$n$B system. 

For these derivations, we work with the following change of variable.
\begin{definition}[Shifted fidelity]
    The \textit{shifted fidelity} $H$ of the 1G$n$B system is given by 
    \begin{equation}
        H \coloneqq F - \frac{1}{4}, 
    \end{equation}
    where $F$ is the fidelity of the link in the G memory.
\end{definition}
This will simplify our calculations because under decoherence, the shifted fidelity changes due to a multiplicative exponential factor. In particular, given an initial value $h$ of the shifted fidelity, after $t$ time steps this reduces to
\begin{equation}
    h \rightarrow e^{-\Gamma t} h.
\end{equation}
We see that the shifted fidelity does not inherit linear terms under decoherence, in contrast to the fidelity, which decays according to (\ref{eq.fidelity_decay}).  This will simplify our derivations.

After successful $(k+1)$-to-$1$ purification, the value $h$ of the shifted fidelity undergoes a jump given by
\begin{equation}
    \tilde{J}_k(h) \coloneqq J_k\! \left(h+\frac{1}{4} \right) - \frac{1}{4} = \frac{a_k h + b_k}{c_k h + d_k} 
    \label{eq.jump-for-G}
\end{equation}
where we have used (\ref{eq.Jk_def}). Similarly, the probability of successful purification is
\begin{equation}
    \tilde{p}_k(h) \coloneqq p_k\! \left(h+\frac{1}{4} \right) = c_k h + d_k.
    \label{eq.prob-for-G}
\end{equation}
Therefore, $\tilde{J}_k$ and $\tilde{p}_k$ are the jump function and success probability of the corresponding purification events for the shifted fidelity.

Finally, we notice that the range for the fidelity $F \in [0,1]$ translates to $H\in \left[- \frac{1}{4}, \frac{3}{4} \right]$. In particular, we have $H<0$ if and only if there is no link in the G memory.

We have fully characterised the dynamics of the shifted fidelity in 1G$n$B (decoherence, purification, and link removal). Our two key performance metrics may then be rewritten in terms of $H$. Recall that with the assumption $F_{\mathrm{new}}>1/4$, and the depolarising decoherence model (\ref{eq.fidelity_decay}), a link exists at time $t$ if and only if $F(t)>1/4$, or equivalently $H(t)>0$.
Let us again denote the indicator function when acting on the shifted fidelity as 
\begin{align*}
\mathbb{1}_{\text{link exists}}(H) \equiv   \mathbb{1}_{\mathrm{l.e.}}(H) =    \begin{cases}
        1 \text{ if } H\geq 0, \\ 
        0 \text{ if } H<0.
    \end{cases}  
\end{align*}
Recalling Definition \ref{def.A}, the availability may then be written as 
\begin{equation}
 A = \lim_{m\rightarrow \infty}\frac{1}{m} \sum_{j = 1}^{m} \mathbbm{1}_{\mathrm{l.e.}}\left(H(T_{\mathrm{con}}^{(j)})\right).
\end{equation}
Recalling Definition \ref{def.F}, the average consumed fidelity may be rewritten as
\begin{align}
    \overline{F} &= \lim_{m\rightarrow \infty} \frac{\sum_{j = 1}^{m} F(T_{\mathrm{con}}^{(j)})\cdot  \mathbbm{1}_{\mathrm{l.e.}}\left(F(T_{\mathrm{con}}^{(j)})\right)}{\sum_{j = 1}^{m} \mathbbm{1}_{\mathrm{l.e.}}\left(F(T_{\mathrm{con}}^{(j)})\right) } \nonumber \\ &=\lim_{m\rightarrow \infty} \frac{\sum_{j = 1}^{m} \left(\frac{1}{4} + H(T_{\mathrm{con}}^{(j)})\right) \cdot \mathbbm{1}_{\mathrm{l.e.}}\left(H(T_{\mathrm{con}}^{(j)})\right)  }{\sum_{j = 1}^{m} \mathbbm{1}_{\mathrm{l.e.}}\left(H(T_{\mathrm{con}}^{(j)})\right)} \nonumber \\ &= \lim_{m\rightarrow \infty } \left[ \frac{1}{4} + \frac{\sum_{j = 1}^{m}  H(T_{\mathrm{con}}^{(j)}) \cdot \mathbbm{1}_{\mathrm{l.e.}}\left(H(T_{\mathrm{con}}^{(j)})\right) }{\sum_{j = 1}^{m} \mathbbm{1}_{\mathrm{l.e.}}\left(H(T_{\mathrm{con}}^{(j)})\right) }\right] \nonumber \\ &= \frac{1}{4} + \overline{H}. \label{eq.Fbar-Hbar-relation}
\end{align}
We have now written $\overline{F}$ in terms of $\overline{H}$, where
\begin{equation}
    \overline{H} \coloneqq \lim_{m \rightarrow \infty } \left[\frac{\sum_{j = 1}^{m}  H(T_{\mathrm{con}}^{(j)}) \cdot \mathbbm{1}_{\mathrm{l.e.}}\left(H(T_{\mathrm{con}}^{(j)})\right) }{\sum_{j = 1}^{m} \mathbbm{1}_{\mathrm{l.e.}}\left(H(T_{\mathrm{con}}^{(j)})\right) }\right]
\end{equation}
is the average consumed \textit{shifted} fidelity. Finding a formula for $\overline{F}$ then reduces to finding a formula for $\overline{H}$.

From now on, we will assume that the system starts with shifted fidelity $H(0) = H_{\mathrm{new}}$, where
\begin{equation}
    H_{\mathrm{new}} \coloneq F_{\mathrm{new}} - \frac{1}{4}
\end{equation}
is the state of the G memory immediately after transferring a freshly generated link into memory. Note that $H_{\mathrm{new}}$ is a constant, as newly generated links are assumed to be identical. The subsequent dynamics of the system will then be as follows: the link may undergo decoherence followed by purification a number of times, until the link is removed. The removal is due to either consumption or purification failure. After the link is removed, entanglement generation will be attempted until success, at which point a link is transferred to the G memory with shifted fidelity $H_{\mathrm{new}}$. See Figure~\ref{fig.one-cycle} for an illustration of this.

\begin{definition}
We define $T_0 = 0$, $\{T_i\}_{i=1}^{\infty}$ to be the times at which $H$ (equivalently, $F$) experiences a change that is due to purification, consumption or entanglement generation (or alternatively, any change that is not due to decoherence). Let $S_i \coloneqq T_i - T_{i-1}$ denote the times between each jump. 
\label{def.jump-times}
\end{definition}
We also refer to the $\{T_i\}$ as the \textit{jump times}. See Figure~\ref{fig.one-cycle} for a depiction. 

\begin{figure}[t]
    \centering
    \includegraphics[width=0.7\linewidth]{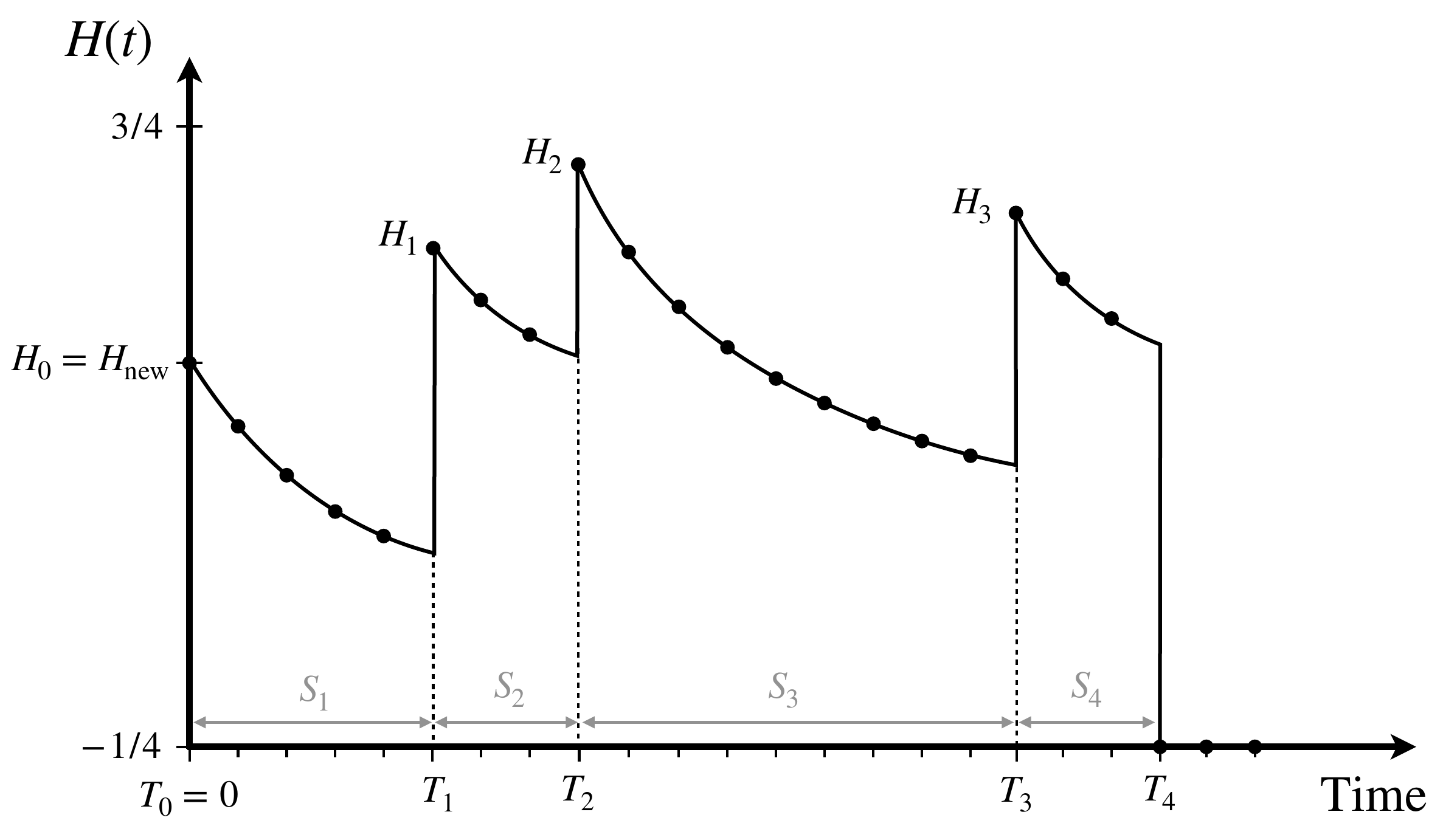}
    \caption{\textbf{Example dynamics of shifted fidelity in the first cycle of 1G$n$B.}
    We assume that $H(0) = H_{\mathrm{new}}$, or equivalently that a freshly generated link is transferred to memory at time $t=0$. The $\{T_i\}_{i\geq 0}$ are defined to be the times at which there are changes in the (shifted) fidelity that are not due to decoherence. We let $T_{\mathrm{occ}}$ be the first time at which the link is removed from the G memory. In the example, $T_{\mathrm{occ}} = T_4$.}
    \label{fig.one-cycle}
\end{figure}
Now, recall that both the time until entanglement generation and consumption are assumed to be geometrically distributed. Then, the distribution of $S_i$ is then given by
\begin{equation}
    S_i =
    \begin{cases}
        \min \left\{\tau^{(i)}_{\mathrm{pur}},\tau^{(i)}_{\mathrm{con}} \right\} &\text{ if } H(T_{i-1}) \geq 0 \\
         T_{\mathrm{gen}}^{(i)} &\text{ if } H(T_{i-1}) < 0,
    \end{cases}
    \label{eq.S_i-distribution}
\end{equation}
where $T_{\mathrm{gen}}^{(i)}$, $\tau_{\mathrm{pur}}^{(i)}$, and $\tau_{\mathrm{con}}^{(i)}$ are independent random variables with the following distributions
\begin{align}
    T^{(i)}_{\mathrm{gen}} &\sim \mathrm{Geo}\left(1-(1-p_{\mathrm{gen}})^n \right) \nonumber
    \\
    \tau_{\mathrm{pur}}^{(i)} &\sim \mathrm{Geo} \left(q(1-(1-p_{\mathrm{gen}})^n)\right) \nonumber
    \\ 
    \tau_{\mathrm{con}}^{(i)} &\sim \mathrm{Geo}(p_{\mathrm{con}}). \label{eq.dist-tau}
\end{align}
Here, starting at jump time $T_{i-1}$, $T^{(i)}_{\mathrm{gen}}$ is the time until a new link is generated and transferred to memory, $\tau_{\mathrm{pur}}^{(i)}$ is the time until there is a successful generation and the system decides to attempt purification, and $\tau_{\mathrm{con}}^{(i)}$ is the time until there is a consumption request. 
\begin{definition}
    For $i\geq 0$, we define $H_i \coloneqq H(T_i)$ to be the shifted fidelity at the jump times of the process. See Figure~\ref{fig.one-cycle} for an illustration.
    %We assume that these take the value of the shifted fidelity \textit{just after} the jump (see Figure~\ref{fig.one-cycle}) 
\end{definition}
Since we assume that the system starts with a freshly generated link in memory, we have $H_0 = H_{\mathrm{new}}$. We note that $\{H_i\}_{i\geq 0}$ is a Markov chain.
\begin{definition}
    Let $T_{\mathrm{occ}}$ be the first time at which the link in the G memory is removed from the system. In particular, $T_{\mathrm{occ}} = T_N$, where
    \begin{equation}
        N = \min\big\{i:H_{i} < 0 \big\}.
        \label{eq.def-N}
    \end{equation}
    \label{def.T_occ}
\end{definition}
Note that $N$ is finite a.s. since it is upper bounded by the time until the first consumption request arrives, which follows a geometric distribution.

In Appendix~\ref{app.viewpoint}, we saw that $F(t)$ is a regenerative process, meaning that it can be broken down into i.i.d. cycles $V_k = T_{\mathrm{occ}}^{(k)}+ T_{\mathrm{gen}}^{(k)}$, where $T_{\mathrm{occ}}^{(k)}$ are the times during which the G memory is occupied and $T_{\mathrm{gen}}^{(k)}$ are the times during which the G memory is empty. We note that in Definition \ref{def.T_occ}, $T_{\mathrm{occ}} = T_{\mathrm{occ}}^{(1)}$. From now on, we also refer to $T_{\mathrm{gen}}\equiv T_{\mathrm{gen}}^{(1)}$.

It follows straightforwardly that $H(t) = F(t) - 1/4$ is a regenerative process with the same cycles as $F(t)$. We saw in Lemma \ref{lem.metrics-single-cycle} that the performance metrics may be rewritten in terms of the statistical properties of one cycle. This result also holds for $\overline{H}$, which we restate below.
Recalling the notation introduced in Appendix \ref{app.viewpoint} for $F^-$, we will also use  the equivalent notation for $H^-$, i.e.
\begin{equation}
    H^-(t) = F^-(t) - \frac{1}{4}.
\end{equation}
\begin{lemma}[Performance metrics for $H$, single cycle]
The availability is given by
\begin{equation}
    A = \frac{\mathbb{E}[T_{\mathrm{occ}}]}{\mathbb{E}[T_{\mathrm{occ}}]+\mathbb{E}[T_{\mathrm{gen}}]} \text{ a.s.}
    \label{eq.A-first-cycle}
\end{equation}
and the average consumed (shifted) fidelity is given by
\begin{align}
    \overline{H} = \mathbb{E}\! \left[e^{-\Gamma S_N }H_{N-1}|\tau_{\mathrm{con}}^{(N)}\leq \tau_{\mathrm{pur}}^{(N)}\right]\text{ a.s.}
    \label{eq.Gbar-first-cycle}
\end{align}
where $C_1\equiv \{\tau_{\mathrm{con}}^{(N)}\leq \tau_{\mathrm{pur}}^{(N)}\}$ is the event that the link is consumed at time $T_{\mathrm{occ}}$.
\label{lem.first-cycle}
\end{lemma}
\begin{proof}
The identity (\ref{eq.A-first-cycle}) follows directly from Lemma \ref{lem.metrics-single-cycle}. In the same Lemma, we saw that 
\begin{equation}
    \overline{F} = \mathbb{E}[F(T_{\mathrm{occ}}^{(1)-})|C_1],
\end{equation}
where $C_1$ is the event that the first link is removed due to consumption, and we recall the notation
$$F^-(t) = e^{-\Gamma}\left(F(t-1)-\frac{1}{4}\right) + \frac{1}{4},$$
which is necessary to capture the fidelity when \textit{consumed} at time $t$, since the discrete-time stochastic process is defined such that $H(T_{\mathrm{occ}}^{(1)})=0$.
The value of $H^-(T_{\mathrm{occ}}^{(1)})$ is given by $e^{-\Gamma S_N} H_{N-1}$, where $H_{N-1}$ is the value of the shifted fidelity at the previous jump time (see Definition \ref{def.jump-times}) and $S_N$ is the time the link spends decohering in memory from that point until the link is removed from memory (see Definition \ref{def.T_occ}). For the conditioning, we recall from (\ref{eq.S_i-distribution}) that $C_1 \equiv \{\tau_{\mathrm{con}}^{(N)}\leq \tau_{\mathrm{pur}}^{(N)}\}$. 
\end{proof}
By properties of geometric random variables, we already know that 
\begin{equation*}
    \mathbb{E}[T^{(1)}_{\mathrm{gen}}] = \frac{1}{1 - (1-p_{\mathrm{gen}})^n}.
\end{equation*}
To solve for our two performance metrics, it is then sufficient to find formulae for $\mathbb{E}[T_{\mathrm{occ}}]$ and $\mathbb{E}\!\left[e^{-\Gamma S_N }H_{N-1}|\tau_{\mathrm{con}}^{(N)}\leq \tau_{\mathrm{pur}}^{(N)}\right]$.
This is what we  accomplish with the following results.
\begin{definition}
    For $i\leq N$, let $U_i$ denote the event that purification is attempted at the $i$th jump time, and $R_i \subseteq U_i$ denote the event that purification is attempted \textit{and} succeeds at the $i$th jump time. 
\end{definition}
\begin{lemma}
    Let $x$ and $y$ be given by
    \begin{equation}
        x \coloneqq \sum_{i=1}^{\infty} \mathbb{E}\left[H_i \prod_{j=1}^{i}\mathbbm{1}_{R_{j}}\right], \;\;\;\;\;  y \coloneqq \sum_{i=1}^{\infty}\mathrm{P}\! \left(N>i\right). 
        \label{eq.def-x-y}
    \end{equation}
Then,     
    \begin{equation}
        \mathbb{E}[T_{\mathrm{occ}}] = \frac{1+y}{p_{\mathrm{con}}+q \left(1- \left(1-p_{\mathrm{gen}} \right)^n\right)(1-p_{\mathrm{con}})}
        \label{eq.E(T_N)-rewrite}
    \end{equation}
    and
    \begin{equation}
        \mathbb{E}[e^{-\Gamma S_N }H_{N-1}|\tau^{(N)}_{\mathrm{con}}\leq \tau^{(N)}_{\mathrm{pur}}]=\frac{(H_{\mathrm{new}}+x)(p_{\mathrm{con}}+q\left(1- \left(1-p_{\mathrm{gen}} \right)^n\right)(1-p_{\mathrm{con}}))}{(1+y)\left( e^{\Gamma}-1+p_{\mathrm{con}} + q\left(1- \left(1-p_{\mathrm{gen}} \right)^n\right)(1-p_{\mathrm{con}})\right)}.
        \label{eq.acf-in-xy}
    \end{equation}
    %\agi{Maybe it would be more useful for the reader to combine Lemma 2 and Proposition 1, so they don't have to scan through the document to find a more explicit form of $x$ and $y$?} \bd{this would result in a very long derivation for this lemma, so I think best to keep it separate for clarity.}
    \label{lem.metrics-in-x-y}
\end{lemma}
\begin{proof}
Denoting $U_N^{\mathrm{c}} = \left\{\tau^{(N)}_{\mathrm{con}}\leq \tau^{(N)}_{\mathrm{pur}} \right\}$ and using properties of the conditional expectation,
we may write 
    \begin{align}
        \mathbb{E}[e^{-\Gamma S_N}H_{N-1}|\tau^{(N)}_{\mathrm{con}}\leq \tau^{(N)}_{\mathrm{pur}}] = \frac{\mathbb{E}[e^{-\Gamma S_N }H_{N-1} \mathbbm{1}_{U_N^{\mathrm{c}}}]}{\mathrm{P}(U_N^{\mathrm{c}})}.
        \label{eq.EG_N_fraction}
    \end{align}
The denominator $\mathrm{P}(U_N^{\mathrm{c}})$ may be rewritten as 
    \begin{align}
        \mathrm{P}(U_N^{\mathrm{c}}) &= \mathbb{E}[\mathbbm{1}_{U_N^{\mathrm{c}}}] \nonumber \\ &\stackrel{\text{i}}{=} \mathbb{E}\!\left[\mathbbm{1}_{U_1^{\mathrm{c}}}+ \sum_{i=2}^{\infty} \mathbbm{1}_{U_i^{\mathrm{c}}}\prod_{j=1}^{i-1} \mathbbm{1}_{R_j}\right] \nonumber \\ &\stackrel{\text{ii}}{=} \mathbb{E}\left[ \mathbbm{1}_{U_1^{\mathrm{c}}} \right] + \sum_{i=2}^{\infty} \mathbb{E}\left[\mathbbm{1}_{U_i^{\mathrm{c}}} |R_1, ..., R_{i-1} \right] \mathbb{E}\left[ \prod_{j=1}^{i-1}\mathbbm{1}_{R_j} \right] \nonumber \\ &\stackrel{\text{iii}}{=} \mathbb{E}\left[ \mathbbm{1}_{U_1^{\mathrm{c}}} \right] \left(1+ \sum_{i=1}^{\infty}\mathbb{E}\left[ \prod_{j=1}^{i}\mathbbm{1}_{R_j} \right] \right) \nonumber \\ &\stackrel{\text{iv}}{=} \mathrm{P}\!\left( U_1^{\mathrm{c}}\right)\left(1+ \sum_{i=1}^{\infty}\mathrm{P}\left(N>i \right)\right)
         \nonumber \\ &= \mathrm{P}\!\left( U_1^{\mathrm{c}}\right)\left(1+ y\right). \label{eq.P(U_N^c)}
    \end{align}
In the above, we have used the following steps:
\begin{enumerate}[label=\roman*.]
    \item One may partition the event $U_N^{\mathrm{c}}$ by conditioning on the value of $N$ as
    $$U_N^{\mathrm{c}} = \bigcup_{i=1}^\infty \left( U_i^{\mathrm{c}} \cap \left\{N=i \right\} \right). $$
    Now, notice that we have $U_i^{\mathrm{c}} \cap \{ N=i \}$ exactly when successful purification occurs $i-1$ times, and the link is consumed. Therefore,
    \begin{equation}
        U_i^{\mathrm{c}} \cap \{ N=i \} = U_i^{\mathrm{c}} \cap \left( \cap_{j = 1}^{i-1} R_j \right).
        \label{eq.event-U_i^c-N=i}
    \end{equation}
    Since the $U_i^{\mathrm{c}} \cap \left\{N=i \right\}$ are mutually exclusive, it follows from the above that 
    \begin{align}
        \mathbbm{1}_{U_N^{\mathrm{c}}} &= \sum_{i=1}^{\infty} \mathbbm{1}_{U_i^{\mathrm{c}} \cap \left\{N=i \right\}} \nonumber \\ &= \sum_{i=1}^{\infty} \mathbbm{1}_{ U_i^{\mathrm{c}} \cap \left( \cap_{j = 1}^{i-1} R_j \right)} \nonumber  \\  &= \mathbbm{1}_{U_1^{\mathrm{c}}} + \sum_{i=2}^{\infty} \mathbbm{1}_{U_i^{\mathrm{c}}} \prod_{j=1}^{i-1} \mathbbm{1}_{R_j}.
        \label{eq.indicator-rewrite}
    \end{align}
    \item We use linearity of taking the expectation and take the expectation inside the sum, which is possible by the monotone convergence theorem (see 5.6.12 of \cite{Grimmett2020}). Then, we express the joint probability in terms of conditional probabilities. %We then use the fact that $U_i^{\mathrm{c}}$ and $R_j$ are independent events for $i>j$. This is because $U_i$ is determined only by $\tau_{\mathrm{con}}^{(i)}$ and $\tau_{\mathrm{pur}}^{(i)}$, which are independent of the previous events.
    \item We have used the fact that
    \begin{align*}
        \mathbb{E} \left[\mathbbm{1}_{U_i^{\mathrm{c}}} |R_1, ..., R_{i-1} \right] &= \mathrm{P}(\tau^{(i)}_{\mathrm{con}}\leq \tau^{(i)}_{\mathrm{pur}} ) = \mathrm{P}(\tau^{(1)}_{\mathrm{con}}\leq \tau^{(1)}_{\mathrm{pur}} ) = \mathbb{E} \left[\mathbbm{1}_{U_1^{\mathrm{c}}} \right].
    \end{align*}
    \item Notice that $N>i$ if and only if the first $i$ jump times are due to successful purification. Therefore, 
    \begin{equation*}
        \{N>i\}  \equiv \cap_{j = 1}^i R_j.
    \end{equation*}
 We therefore have 
    \begin{equation*}
        \mathbb{E}\left[ \prod_{j=1}^{i} \mathbbm{1}_{R_j} \right] = \mathbb{E}\left[ \mathbbm{1}_{\{N>i\}} \right] = P (N>i).
    \end{equation*}
\end{enumerate}
Secondly, we rewrite the numerator of (\ref{eq.EG_N_fraction}) in a similar way: 
 \begin{align}
        \mathbb{E}[e^{-\Gamma S_N} H_{N-1} \mathbbm{1}_{U_N^{\mathrm{c}}}] &\stackrel{\text{i}}{=} \mathbb{E}\left[e^{-\Gamma S_N} H_{N-1} \cdot \left( \mathbbm{1}_{U_1^{\mathrm{c}}} + \sum_{i=2}^{\infty} \mathbbm{1}_{U_i^{\mathrm{c}}} \prod_{j=1}^{i-1} \mathbbm{1}_{R_j}\right)\right] 
        \nonumber \\ &\stackrel{\text{ii}}{=} \mathbb{E}\left[H_{\mathrm{new}}e^{-\Gamma {S_1}} \mathbbm{1}_{U_1^{\mathrm{c}}} +\sum_{i=1}^{\infty}  e^{-\Gamma S_{i+1}} H_i \mathbbm{1}_{U_{i+1}^{\mathrm{c}}} \prod_{j=1}^i\mathbbm{1}_{R_{j}} \right] \nonumber \\
        &\stackrel{\text{iii}}{=}H_{\mathrm{new}} \mathbb{E}\left[e^{-\Gamma {S_1}} \mathbbm{1}_{U_1^{\mathrm{c}}} \right] + \sum_{i=1}^{\infty} \mathbb{E}\left[e^{-\Gamma S_{i+1}} \mathbbm{1}_{U_{i+1}^{\mathrm{c}}} |R_1, ..., R_{i} \right] \mathbb{E}\left[H_i \prod_{j=1}^i \mathbbm{1}_{R_{i}}\right] \nonumber \\ &\stackrel{\text{iv}}{=} \mathbb{E}\left[e^{-\Gamma {S_1}} \mathbbm{1}_{U_1^{\mathrm{c}}} \right] \left(H_{\mathrm{new}}  + \sum_{i=1}^{\infty} \mathbb{E}\left[H_i \prod_{j=1}^i\mathbbm{1}_{R_{j}}\right] \right) \nonumber \\ &= \mathbb{E}\left[e^{-\Gamma {S_1}} \mathbbm{1}_{U_1^{\mathrm{c}}} \right] \left(H_{\mathrm{new}}  + x \right). \label{eq.E[e^(-GammaS_N)G_(N-1)ind]}
    \end{align}
    In the above, we have used the following steps:
\begin{enumerate}[label=\roman*.]
    \item We have again made use of (\ref{eq.indicator-rewrite}), and $H_0 \coloneqq H_{\mathrm{new}}$.
    \item Again making use of (\ref{eq.event-U_i^c-N=i}), we have noticed that the indicator function selects the value of $N$ as
    \begin{align*}
        e^{-\Gamma S_N} H_{N-1} \cdot  \mathbbm{1}_{U_i^{\mathrm{c}}} \prod_{j=1}^{i-1} \mathbbm{1}_{R_j} &= e^{-\Gamma S_N} H_{N-1} \cdot  \mathbbm{1}_{U_i^{\mathrm{c}} \cap \{N = i\}}  \\ &= e^{-\Gamma S_i} H_{i-1} \cdot  \mathbbm{1}_{U_i^{\mathrm{c}}} \prod_{j=1}^{i-1} \mathbbm{1}_{R_j}.
    \end{align*}
    \item We have used linearity of the expectation, and take the expectation inside the sum, which is possible by the monotone convergence theorem (see 5.6.12 of \cite{Grimmett2020}). Then, we express the joint probability in terms of conditional probabilities. %We then used the fact that $S_{i+1}$ and $U_{i+1}^{\mathrm{c}}$ are determined only by $\tau_{\mathrm{con}}^{(i+1)}$ and $\tau_{\mathrm{pur}}^{(i+1)}$, and are therefore independent of $R_j$, for $j<i$.
    \item We have again used the fact that, conditioned on $R_1, ..., R_{i-1}$, $e^{-\Gamma S_i} \mathbbm{1}_{U_i^{\mathrm{c}}}$ are identically distributed, for all $i\geq 1$.
    \end{enumerate}
We now directly evaluate the multiplying factor in the above expressions. 
Using the partition $U_1^{\mathrm{c}} = \bigcup_{i=1}^{\infty} \{U_1^{\mathrm{c}},S_1 = i\} $,
\begin{align}
    \mathbb{E}\left[e^{-\Gamma {S_1}} \mathbbm{1}_{U_1^{\mathrm{c}}} \right] &= \mathbb{E}\left[e^{-\Gamma {S_1}} \mathbbm{1}_{U_1^{\mathrm{c}}} \big| U_1 \right] \mathrm{P}(U_1) + \sum_{i = 1}^{\infty} \mathbb{E} \left[ e^{-\Gamma S_1} \mathbbm{1}_{U_1^{\mathrm{c}}} \big| U_1^{\mathrm{c}}, S_1 = i  \right] \mathrm{P}(U_1^{\mathrm{c}}, S_1 = i) \nonumber \\  &= 0 + \sum_{i = 1}^{\infty} e^{-i\Gamma} \mathrm{P}(U_1^{\mathrm{c}}, S_1 = i). \label{eq.E[e^-GammaS_1]}
\end{align}
Recalling that $ U_1^{\mathrm{c}} = \{\tau_{\mathrm{con}}^{(1)} \leq \tau_{\mathrm{pur}}^{(1)}\}$, we now evaluate
\begin{align}
    \mathrm{P}(U_1^{\mathrm{c}}, S_1 = i) &= \mathrm{P}(i = \tau_{\mathrm{con}}^{(1)}, i \leq \tau_{\mathrm{pur}}^{(1)}) \nonumber \\ &= \mathrm{P}(i = \tau_{\mathrm{con}}^{(1)}) \cdot \mathrm{P}(i \leq \tau_{\mathrm{pur}}^{(1)}) \nonumber \\  &= 
  (1-p_{\mathrm{con}})^{i-1} p_{\mathrm{con}} \cdot
  \left(1 - q(1-(1-p_{\mathrm{gen}})^n)\right)^{i-1}, \label{eq.P(U_1^c,S_1=i)} 
\end{align}
where we have used the fact that $\tau_{\mathrm{con}}^{(1)}$ and $\tau_{\mathrm{pur}}^{(1)}$ are independent, and have distributions as given in (\ref{eq.dist-tau}). Therefore, combining (\ref{eq.E[e^-GammaS_1]}) and (\ref{eq.P(U_1^c,S_1=i)}), it follows that 
\begin{align}
        \mathbb{E}\left[e^{-\Gamma {S_1}} \mathbbm{1}_{U_1^{\mathrm{c}}} \right]  &= \sum_{i=1}^{\infty} e^{-\Gamma i} (1-p_{\mathrm{con}})^{i-1}p_{\mathrm{con}} \cdot \left(1- q\left(1- \left(1-p_{\mathrm{gen}} \right)^n\right)\right)^{i-1} \nonumber \\ &= p_{\mathrm{con}} e^{-\Gamma} \sum_{i=0}^{\infty} e^{-\Gamma i} \left(1- q\left(1- \left(1-p_{\mathrm{gen}} \right)^n\right)\right)^{i}(1-p_{\mathrm{con}})^{i} \nonumber \\ &= \frac{p_{\mathrm{con}} e^{-\Gamma}}{1-e^{-\Gamma}\left(1- q\left(1- \left(1-p_{\mathrm{gen}} \right)^n\right)\right)(1-p_{\mathrm{con}})}, \label{eq.E[e^(-Gamma S_1)ind]}
\end{align}
where to obtain the first equality we have relabelled the summing index, and to obtain the second equality we have used the formula for a geometric series. By setting $\Gamma = 0$ in the above, we also obtain 
    \begin{align}
    \mathbb{E}\left[\mathbbm{1}_{U_1^{\mathrm{c}}}\right] = \mathrm{P}(U_1^{\mathrm{c}}) &= \frac{p_{\mathrm{con}} }{1-\left(1- q\left(1- \left(1-p_{\mathrm{gen}} \right)^n\right)\right)(1-p_{\mathrm{con}})} \nonumber \\ &=\frac{p_{\mathrm{con}} }{p_{\mathrm{con}} + q \left(1- \left(1-p_{\mathrm{gen}} \right)^n\right)(1-p_{\mathrm{con}})}. \label{eq.P(U_1c)}
    \end{align}
Then, combining (\ref{eq.P(U_N^c)}), (\ref{eq.E[e^(-GammaS_N)G_(N-1)ind]}) (\ref{eq.E[e^(-Gamma S_1)ind]}), (\ref{eq.P(U_1c)})  allows us to rewrite (\ref{eq.EG_N_fraction}) as
\begin{equation*}
        \mathbb{E}[e^{-\Gamma S_N }H_{N-1}|\tau^{(N)}_{\mathrm{con}}\leq \tau^{(N)}_{\mathrm{pur}}] =  \frac{(H_{\mathrm{new}}+x)(p_{\mathrm{con}}+q\left(1- \left(1-p_{\mathrm{gen}} \right)^n\right)(1-p_{\mathrm{con}}))}{(1+y)\left( e^{\Gamma}-\left(1- q\left(1- \left(1-p_{\mathrm{gen}} \right)^n\right)\right)(1-p_{\mathrm{con}})\right)}.
\end{equation*}
This shows (\ref{eq.acf-in-xy}). We now show (\ref{eq.E(T_N)-rewrite}) using a similar method: firstly, we again condition on the value of $N$. Recalling Definitions \ref{def.jump-times} and \ref{def.T_occ}, we may rewrite $T_N$ as 
\begin{align*}
    T_N &= \sum_{i=1}^{\infty} S_i \cdot \mathbbm{1}_{ \{N \geq i \} } \nonumber = S_1 + \sum_{i=2}^{\infty} S_i \prod_{j=1}^{i-1} \mathbbm{1}_{R_j},
\end{align*}
where we have again used $\{N\geq i \} \equiv \cap_{j=1}^{i-1} R_j$ to obtain the second equality. Taking expectations, it follows that 
\begin{align}
\mathbb{E}[T_N] &= \mathbb{E}\left[S_1 + \sum_{i=2}^{\infty} S_i
 \prod_{j=1}^{i-1} \mathbbm{1}_{R_j}\right] \nonumber \\
 &= \mathbb{E}\left[ S_1 \right] + \sum_{i=2}^{\infty}\mathbb{E}\left[S_i |R_1, ..., R_{i-1} \right] \mathbb{E}\left[ \prod_{j=1}^{i-1} \mathbbm{1}_{R_j} \right] \nonumber \\
 &= \mathbb{E}\left[ S_1 \right]\left( 1+ \sum_{i=1}^{\infty}\mathbb{E}\left[ \prod_{j=1}^{i} \mathbbm{1}_{R_j} \right] \right) \nonumber \\  &= \mathbb{E}\left[ S_1 \right]\left( 1+ y \right), \label{eq.E[T_N]-in-y}
\end{align}
where we have used the same reasoning as was used to obtain (\ref{eq.P(U_N^c)}) and (\ref{eq.E[e^(-GammaS_N)G_(N-1)ind]}). It now only remains to compute $\mathbb{E}[S_1]$. Recalling that $S_1 = \min \{ \tau_{\mathrm{con}}^{(1)} ,\tau_{\mathrm{pur}}^{(1)} \}$, we see that
\begin{align*}
    \mathrm{P}(S_1 >i) &= \mathrm{P}( \tau_{\mathrm{con}}^{(1)}>i,\tau_{\mathrm{pur}}^{(1)}>i) \\ &= \mathrm{P}( \tau_{\mathrm{con}}^{(1)}>i)(\tau_{\mathrm{pur}}^{(1)}>i) \\ &= (1- p_{\mathrm{con}})^i \cdot (1- q(1- (1 - p_{\mathrm{gen}})^n ))^i,
\end{align*}
where we have used the fact that $\tau_{\mathrm{con}}^{(1)} $ and  $ \tau_{\mathrm{pur}}^{(1)}$ are independent random variables, and their distributions which are given in (\ref{eq.dist-tau}). Then, we may rewrite the expectation as 
\begin{align*}
    \mathbb{E}[S_1] = \sum_{i=0}^{\infty} \mathrm{P}(S_1>i) &= \sum_{i=0}^{\infty} (1- p_{\mathrm{con}})^i \cdot (1- q(1- (1 - p_{\mathrm{gen}})^n ))^i \\ &= \frac{1}{1 - (1-p_{\mathrm{con}})(1- q(1- (1 - p_{\mathrm{gen}})^n ))} \\ &= \frac{1}{p_{\mathrm{con}} + q(1-p_{\mathrm{con}}) (1-(1-p_{\mathrm{gen}})^n)},
\end{align*}
where we have used the formula for a geometric series to evaluate the sum. Rearranging terms and combining the above with (\ref{eq.E[T_N]-in-y}), we may then write this as 
\begin{equation*}
    \mathbb{E}[T_N] = \frac{1+y}{p_{\mathrm{con}} + q(1-p_{\mathrm{con}}) (1-(1-p_{\mathrm{gen}})^n)},
\end{equation*}
which shows (\ref{eq.E(T_N)-rewrite}).
\end{proof}

\begin{lemma}
\label{lem.formulae-x-y}
    Let $x$ and $y$ be defined as in (\ref{eq.def-x-y}). Then, 
    \begin{equation}
       x = - H_{\mathrm{new}} + \frac{ \tilde{B} -\tilde{D}H_{\mathrm{new}} +H_{\mathrm{new}}}{(1-\tilde{A})(1-\tilde{D}) - \tilde{B}\tilde{C}}, \;\;\;\;\; y = -1 + \frac{1-\tilde{A} + \tilde{C}H_{\mathrm{new}}}{(1-\tilde{A})(1-\tilde{D}) - \tilde{B}\tilde{C}},
       \label{eq.x-y-in-ABCD}
    \end{equation}
    where $\tilde{A}$, $\tilde{B}$, $\tilde{C}$, $\tilde{D}$ are defined in Theorem \ref{theorem.A} in the main text.
\end{lemma}
\begin{proof}
    We firstly define the quantities 
\begin{equation}
   x_i \coloneqq  \mathbb{E}\left[H_i \prod_{j=1}^{i}\mathbbm{1}_{R_j} \right], \;\;\;\;\;
   y_i \coloneqq \mathrm{P}\left(N>i \right),
\end{equation}
which means that, recalling (\ref{eq.def-x-y}), $x$ and $y$ may be rewritten as 
\begin{equation}
    x = \sum_{i=1}^{\infty} x_i, \;\;\;\;\; y = \sum_{i=1}^{\infty} y_i.
\end{equation}
We now show that there is a recursive relationship between the $\{x_i\}$ and the $\{y_i\}$. We firstly rewrite $x_i$ by conditioning on the value of $H_{i-1}$. In particular, recalling that $\cap_{j=1}^{i} R_j = \{N>i\}$, we have 
\begin{align*}
    x_i &= \mathbb{E}\left[H_i \prod_{j=1}^{i}\mathbbm{1}_{R_j} \right] = \mathbb{E}\left[H_i \mathbbm{1}_{ \cap_{j=1}^i R_j} \right] = \mathbb{E}\left[H_i \mathbbm{1}_{ \{N>i\} } \right].
\end{align*}
Then, one may partition by conditioning on the value of $H_{i-1}$ in the following way:
\begin{equation*}
    \{N>i-1\} = \bigcup_h \{H_{i-1} = h, N > i-1 \}.
\end{equation*}
We may then rewrite $x_i$ as 
\begin{align}
    x_i  &= \mathbb{E}\left[H_i \mathbbm{1}_{\{N>i\} }\big| N\leq i-1 \right] \mathrm{P}(N\leq i-1) + \sum_{h} \mathbb{E}\left[H_i \mathbbm{1}_{\{N>i\}} \big| H_{i-1} = h, N > i-1\right] \cdot \mathrm{P} ( H_{i-1} = h, N > i-1 ) \nonumber\\ &= 0 + \sum_{h} \mathbb{E}\left[H_i \mathbbm{1}_{\{N>i\}} \big| H_{i-1} = h, N > i-1\right] \cdot \mathrm{P} ( H_{i-1} = h, N > i-1 ).
    \label{eq.x_i-expand}
\end{align}
We now focus on evaluating $\mathbb{E}\left[H_i \mathbbm{1}_{\{N>i\}} \big| H_{i-1} = h, N > i-1\right]$. We do this for $h>0$, as this is the only relevant range in the above formula. We firstly notice that this expression may be rewritten as
\begin{align}
    \mathbb{E}\left[H_i \mathbbm{1}_{\{N>i\}} \big| H_{i-1} = h, N > i-1\right] &= \mathbb{E}\left[H_i \prod_{j=1}^{i}\mathbbm{1}_{R_j} \big| H_{i-1} = h, \cap_{j=1}^{i-1} R_j \right] \nonumber \\ &= \mathbb{E}\left[H_i \mathbbm{1}_{R_i} \big| H_{i-1} = h, \cap_{j=1}^{i-1} R_j \right]  \nonumber \\  &= \mathbb{E}\left[H_i \mathbbm{1}_{R_i} \big| H_{i-1} = h  \right],
    \label{eq.x_i-simplify-conditional}
\end{align}
where to obtain the final equality we have used the Markovian property of the system: given the information that $H_{i-1} = h > 0$, this is sufficient to understand the future behaviour $\{H_k\}_{k \geq i}$. This follows from the fact that $\{H_i\}$ is a Markov chain.

Recall that $R_i$ is the event where the $i$th jump time is due to a purification round succeeding. Given that in the above expression we are conditioning on the value $H_{i-1}$, the random variables on which $H_i$ depends are therefore the time $S_i$ until the next round of purification, and the number of links $L_i$ that are used for this purification (recalling that this number determines which purification protocol is used). We must therefore take the expectation over these two random variables.
\begin{definition}
    For $i<N$ (the $i$-th successful purification round), let $L_i$ be the number of links that were produced in the bad memories just before time $T_i$.  
\end{definition}
We then expand the expectation (\ref{eq.x_i-simplify-conditional}) to condition on the values taken by $S_i$ and $L_i$:
\begin{align}
    \mathbb{E}\left[H_i \mathbbm{1}_{R_i} \big| H_{i-1} = h  \right] &= \sum_{t,k} \mathbb{E}\left[H_i \mathbbm{1}_{R_i} \big| H_{i-1} = h , S_i = t, L_i = k, R_i \right] \cdot \mathrm{P}( S_i = t, L_i = k, R_i | H_{i-1} = h) \nonumber \\  &= \sum_{t,k} \tilde{J}_k\! \left(e^{-\Gamma t} h \right) \cdot \mathrm{P}( S_i = t, L_i = k, R_i | H_{i-1} = h),
    \label{eq.x_i-J_k}
\end{align}
where, recalling (\ref{eq.jump-for-G}), $\tilde{J}_k$ is the jump function corresponding to the $(k+1)$-to-$1$ purification protocol from our purification policy. To evaluate (\ref{eq.x_i-J_k}), it now remains to compute the probability distribution in the weighted sum. We again condition, to find
\begin{align}
   \mathrm{P}( S_i = t, L_i = k, R_i | H_{i-1} = h) &= \mathrm{P}( R_i|U_i, S_i = t, L_i = k, H_{i-1} = h) \cdot \mathrm{P}(U_i, S_i = t, L_i = k| H_{i-1} = h) \nonumber \\ &= \tilde{p}_k ( e^{-\Gamma t} h) \cdot  \mathrm{P}(U_i, S_i = t, L_i = k| H_{i-1} = h),
   \label{eq.x_i-p_k}
\end{align}
where $\tilde{p}_k$ determines the probability of successful purification when employing the $(k+1)$-to-$ 1$ protocol, recalling its definition in (\ref{eq.prob-for-G}).
Now, recalling the distribution of $S_i$ from (\ref{eq.S_i-distribution}), 
\begin{align}
    \mathrm{P}\left( U_i,\: S_i = t,\: L_i = k| H_{i-1} = h\right) &= \mathrm{P}(\tau_{\mathrm{con}}^{(i)}>t,\: \tau_{\mathrm{pur}}^{(i)}=t, \: L_i = k) \nonumber \\ &= \mathrm{P}(\tau_{\mathrm{con}}^{(i)}>t) \cdot \mathrm{P}(\tau_{\mathrm{pur}}^{(i)}=t,\: L_i = k) \nonumber \\ &= (1-p_{\mathrm{con}})^t \cdot (1 - q(1-(1-p_{\mathrm{gen}})^n))^{t-1} \cdot q {n \choose k} p_{\mathrm{gen}}^k (1-p_{\mathrm{gen}})^{n-k},
    \label{eq.x_i-taus}
\end{align}
where we have used the fact that $\tau_{\mathrm{pur}}^{(i)}$ and $L_i$ are independent of $\tau_{\mathrm{con}}^{(i)}$. 

Combining (\ref{eq.x_i-J_k}), (\ref{eq.x_i-p_k}) and (\ref{eq.x_i-taus}) yields 
\begin{align}
     \mathbb{E}\left[H_i \mathbbm{1}_{R_i} \big| H_{i-1} = h  \right] &= \sum_{t,k} \tilde{J}_k\! \left( e^{-\Gamma t} h \right) \tilde{p}_k \! \left( e^{-\Gamma t} h \right) \cdot (1-p_{\mathrm{con}})^t \cdot (1 - q(1-(1-p_{\mathrm{gen}})^n))^{t-1} \cdot q {n \choose k} p_{\mathrm{gen}}^k (1-p_{\mathrm{gen}})^{n-k} \nonumber \\ &= \sum_{t,k} \left(a_k e^{-\Gamma t} h + b_k \right) \cdot (1-p_{\mathrm{con}})^t \cdot (1 - q(1-(1-p_{\mathrm{gen}})^n))^{t-1} \cdot q {n \choose k} p_{\mathrm{gen}}^k (1-p_{\mathrm{gen}})^{n-k} \nonumber \\ &= \sum_{t} \left(\tilde{a} e^{-\Gamma t} h + \tilde{b} \right) \cdot (1-p_{\mathrm{con}})^t \cdot (1 - q(1-(1-p_{\mathrm{gen}})^n))^{t-1} \cdot q, \label{eq.x_i-atilde-btilde}
\end{align}
where in the first step, we have made use of the expressions (\ref{eq.jump-for-G}) and (\ref{eq.prob-for-G}) that define the purification jump function and success probability for the shifted fidelity. In the second step, we have defined 
\begin{align*}
        \tilde{a} &= \sum_{k=1}^n a_k \cdot \binom{n}{k}(1-p_{\mathrm{gen}})^{n-k}p_{\mathrm{gen}}^k, \;\;\;\;\; \tilde{b} =  \sum_{k=1}^n b_k \cdot \binom{n}{k}(1-p_{\mathrm{gen}})^{n-k}p_{\mathrm{gen}}^k.
\end{align*}
Now, using the fact that (\ref{eq.x_i-atilde-btilde}) is a geometric series (starting from $t=1$), we obtain
\begin{align}
         \mathbb{E}\left[H_i  \mathbbm{1}_{R_i} | H_{i-1}=h \right] = \tilde{A} h + \tilde {B},
         \label{eq.x_i-EGi-Atilde-Btilde}
\end{align}
where 
\begin{equation}
    \tilde{A} = \frac{q  (1-p_{\mathrm{con}}) \tilde{a}}{e^{\Gamma} -   \left(1- q+ q\left(1-p_{\mathrm{gen}} \right)^n\right) (1-p_{\mathrm{con}})}, \;\;\;\;\; \tilde{B} = \frac{ q (1-p_{\mathrm{con}})\tilde{b}}{p_{\mathrm{con}} + q\left(1-\left(1-p_{\mathrm{gen}} \right)^n\right) (1-p_{\mathrm{con}})}.
    \label{eq.def-Atilde-Btilde}
\end{equation}
Combining (\ref{eq.x_i-expand}) and (\ref{eq.x_i-EGi-Atilde-Btilde}), we may then write 
\begin{align*}
    x_i &= \sum_{h}(\tilde{A} h + \tilde {B}) \cdot \mathrm{P} ( H_{i-1} = h,\: N > i-1 ) \\ &= \tilde{A}\cdot \mathbb{E}[H_{i-1} \mathbbm{1}_{N>i-1}] + \tilde{B}\cdot \mathrm{P}(N>i-1) \\ &= \tilde{A}x_{i-1} + \tilde{B} y_{i-1},
\end{align*}
which is our first recursion relation for $\{x_i\}$ and $\{y_i\}$. We now write down an analogous recursion relation for $y_i$. We use the same method as for the $x_i$. In particular,  using
\begin{align*}
    y_i = \mathrm{P}(N>i) = \mathbb{E}\left[ \mathbbm{1}_{N>i} \right],
\end{align*}
we again expand the expectation while conditioning on the value of $H_{i-1}$ in a step analogous to (\ref{eq.x_i-expand}):
\begin{align}
    y_i = \sum_{h>0} \mathbb{E}\left[ \mathbbm{1}_{\{N>i\}} | H_{i-1}=h,\: N>i-1\right] \cdot \mathrm{P}(H_{i-1}=h,\: N>i-1).
    \label{eq.y_i-expand}
\end{align}
In a step analogous to (\ref{eq.x_i-simplify-conditional}), we rewrite the conditional expectation as
\begin{align}
    \mathbb{E}\left[ \mathbbm{1}_{\{N>i\}} | H_{i-1}=h,\: N>i-1\right] = \mathbb{E}\left[ \mathbbm{1}_{R_i} | H_{i-1}=h\right]. 
    \label{eq.y_i-simplify-conditional}
\end{align}
We then expand the above with the distributions of $S_i$ and $L_i$ to obtain 
\begin{align*}
  \mathbb{E}\left[ \mathbbm{1}_{R_i} | H_{i-1}=h\right]  &= \sum_{t,k} 1 \cdot \mathrm{P}(S_i=t,L_i = k, R_i|H_{i-1}=h) \\ &= \sum_{t,k} 1 \cdot \tilde{p}_k(e^{- \Gamma t}h) \mathrm{P}(U_i,S_i=t,L_i = k|H_{i-1}=h) \\ &= \sum_{t,k} 1 \cdot \tilde{p}_k(e^{- \Gamma t}h) \cdot  (1-p_{\mathrm{con}})^t \cdot (1 - q(1-(1-p_{\mathrm{gen}})^n))^{t-1} \cdot q {n \choose k} p_{\mathrm{gen}}^k (1-p_{\mathrm{gen}})^{n-k},
\end{align*}
where in the second step we have used (\ref{eq.x_i-p_k}) and in the last step we have again used the conditional distribution in (\ref{eq.x_i-taus}).
Using again the definition for $\tilde{p}_k$ as given in (\ref{eq.prob-for-G}), one may simplify the above to obtain 
\begin{equation}
    \mathbb{E}\left[ \mathbbm{1}_{R_i} | H_{i-1}=h\right] = \sum_{t>0} (\tilde{c} e^{-\Gamma t} h + \tilde{d} ) \cdot  (1-p_{\mathrm{con}})^t \cdot (1 - q(1-(1-p_{\mathrm{gen}})^n))^{t-1} \cdot q ,
    \label{eq.y_i-ctilde-dtilde}
\end{equation}
where we have defined 
\begin{align*}
        \tilde{c} &= \sum_{k=1}^n c_k \cdot \binom{n}{k}(1-p_{\mathrm{gen}})^{n-k}p_{\mathrm{gen}}^k, \;\;\;\;\; \tilde{d} =  \sum_{k=1}^n d_k \cdot \binom{n}{k}(1-p_{\mathrm{gen}})^{n-k}p_{\mathrm{gen}}^k.
\end{align*}
One may again use a geometric series to evaluate (\ref{eq.y_i-ctilde-dtilde}), to obtain 
\begin{equation}
        \mathbb{E}\left[ \mathbbm{1}_{R_i} | H_{i-1}=h\right] = \tilde{C} h + \tilde{D},
        \label{eq.y_i-Ctilde-Dtilde}
\end{equation}
where 
\begin{equation}
    \tilde{C} = \frac{q  (1-p_{\mathrm{con}}) \tilde{c}}{e^{\Gamma} -   \left(1- q+ q\left(1-p_{\mathrm{gen}} \right)^n\right) (1-p_{\mathrm{con}})}, \;\;\;\;\; \tilde{D} = \frac{ q (1-p_{\mathrm{con}})\tilde{d}}{p_{\mathrm{con}} + q\left(1-\left(1-p_{\mathrm{gen}} \right)^n\right) (1-p_{\mathrm{con}})}.
    \label{eq.def-Ctilde-Dtilde}
\end{equation}
Combining (\ref{eq.y_i-expand}), (\ref{eq.y_i-simplify-conditional}) and (\ref{eq.y_i-Ctilde-Dtilde}) then yields
\begin{align*}
    y_i &= \sum_{h>0} (\tilde{C}h +\tilde{D})\cdot \mathrm{P}(H_{i-1}=h,\:N>i-1) \\ &= \tilde{C} \cdot \mathbb{E}\left[H_{i-1} \mathbb{1}_{\{N>i-1\}}\right]+ \tilde{D} \cdot P(N>i-1) \\ &= \tilde{C} x_{i-1} + \tilde{D}y_{i-1},
\end{align*}
which completes our second recursion relation for the $\{x_i\}$ and $\{y_i\}$. We now combine these to find expressions for $x$ and $y$. Given the initial values
\begin{align*}
    x_0 &= \mathbb{E}\left[H_0 \mathbb{1}_{N>0}\right] = \mathbb{E}\left[H_{\mathrm{new}} \cdot 1\right] = H_{\mathrm{new}}
\end{align*}
and 
\begin{equation*}
    y_0 = \mathrm{P}( N>0) = 1,
\end{equation*}
it follows that 
\begin{align*}
    x &= \sum_{i=1}^{\infty} ( \tilde{A}x_{i-1} + \tilde{B}y_{i-1} ) = \tilde{A}(x+H_{\mathrm{new}}) + \tilde{B}(y+1)\\
    y &= \sum_{i=1}^{\infty} ( \tilde{C}x_{i-1} + \tilde{D}y_{i-1} ) = \tilde{C}(x+H_{\mathrm{new}}) + \tilde{D}(y+1).
\end{align*}
We therefore have a linear system of equations for $x$ and $y$, which may be written as
\begin{equation}
    \left(
    \begin{matrix}
    x \\ y
    \end{matrix}
    \right) = \left(\begin{matrix}
    \tilde{A} & \tilde{B} \\ \tilde{C} & \tilde{D}
    \end{matrix} \right) \left(
    \begin{matrix}
    x \\ y
    \end{matrix}
    \right) + \left(
    \begin{matrix}
    \tilde{A}H_{\mathrm{new}} + \tilde{B} \\ \tilde{C}H_{\mathrm{new}} + \tilde{D}
    \end{matrix}
    \right),
\end{equation}
\begin{comment}
or equivalently
\begin{equation}
    \left(\begin{matrix}
    1-\tilde{A} & -\tilde{B} \\ -\tilde{C} & 1-\tilde{D}
    \end{matrix} \right) \left(
    \begin{matrix}
    x \\ y
    \end{matrix}
    \right) = \left(
    \begin{matrix}
    \tilde{A}H_\mathrm{new} + \tilde{B} \\ \tilde{C}H_\mathrm{new} + \tilde{D}
    \end{matrix}
    \right)
\end{equation}
\end{comment}
which has solution
\begin{equation}
    \left(
    \begin{matrix}
    x \\ y
    \end{matrix}
    \right) = \frac{1}{(1-\tilde{A})(1-\tilde{D}) - \tilde{B}\tilde{C}}\left(\begin{matrix}
    1-\tilde{D} & \tilde{B} \\ \tilde{C} & 1-\tilde{A}
    \end{matrix} \right) 
    \left(
    \begin{matrix}
    \tilde{A}H_{\mathrm{new}} + \tilde{B} \\ \tilde{C}H_{\mathrm{new}} + \tilde{D}
    \end{matrix}
    \right),
\end{equation}
providing us with the formulae for $x$ and $y$. These may be simplified in the following way:
\begin{align*}
    x &= \frac{ (1-\tilde{D})(\tilde{A}H_{\mathrm{new}} + \tilde{B}) + \tilde{B}(\tilde{C}H_{\mathrm{new}} + \tilde{D})}{(1-\tilde{A})(1-\tilde{D}) - \tilde{B}\tilde{C}} = - H_{\mathrm{new}} + \frac{ \tilde{B} -\tilde{D}H_{\mathrm{new}} +H_{\mathrm{new}}}{(1-\tilde{A})(1-\tilde{D}) - \tilde{B}\tilde{C}} \\
    y &= \frac{\tilde{C}(\tilde{A}H_{\mathrm{new}} + \tilde{B})+(1-\tilde{A})(\tilde{C}H_{\mathrm{new}} + \tilde{D})}{(1-\tilde{A})(1-\tilde{D}) - \tilde{B}\tilde{C}} = -1 + \frac{1-\tilde{A} + \tilde{C}H_{\mathrm{new}}}{(1-\tilde{A})(1-\tilde{D}) - \tilde{B}\tilde{C}},
\end{align*}
which are in the final form for $x$ and $y$, as given in (\ref{eq.x-y-in-ABCD}).
\end{proof}

\begin{proof}[Proof of Theorems \ref{theorem.A} and \ref{theorem.F}]
We combine Lemmas \ref{lem.first-cycle}, \ref{lem.metrics-in-x-y} and \ref{lem.formulae-x-y}. From Lemma \ref{lem.first-cycle}, we recall that our performance metrics may be written in terms of properties of the first cycle. From Lemma \ref{lem.metrics-in-x-y}, we recall that these may be written in terms of $x$ and $y$. Finally, in Lemma \ref{lem.formulae-x-y} we have found formulae for $x$ and $y$. In order to write down the availability, we firstly combine (\ref{eq.E(T_N)-rewrite}) and (\ref{eq.x-y-in-ABCD}), to find
\begin{align*}
    \mathbb{E}[T_{\mathrm{occ}}] = \mathbb{E}[T_N] &=  \frac{1+y}{p_{\mathrm{con}} + q(1-p_{\mathrm{con}}) (1-(1-p_{\mathrm{gen}})^n)} \\ &= \frac{1-\tilde{A} + \tilde{C}H_{\mathrm{new}}}{\left( (1- \tilde{A})(1- \tilde{D}) - \tilde{B}\tilde{C} \right) \tilde{P}} \\ &= \frac{1-\tilde{A} + \tilde{C}(F_{\mathrm{new}} - \frac{1}{4})}{\left( (1- \tilde{A})(1- \tilde{D}) - \tilde{B}\tilde{C} \right) \tilde{P}}
\end{align*}
where $\tilde{P} \coloneqq p_{\mathrm{con}} + q(1-p_{\mathrm{con}}) (1-(1-p_{\mathrm{gen}})^n)$. This suffices to show Theorem \ref{theorem.A}.

In order to write down the average consumed fidelity, we combine (\ref{eq.acf-in-xy}) and (\ref{eq.x-y-in-ABCD}), to obtain
\begin{align*}
     \overline{H} &\stackrel{\text{a.s.}}{=}\frac{\big[\tilde{B} -\tilde{D}H_{\mathrm{new}} +H_{\mathrm{new}}\big] \cdot \big[p_{\mathrm{con}}+q\left(1- \left(1-p_{\mathrm{gen}} \right)^n\right)(1-p_{\mathrm{con}})\big]}{\big[1 -\tilde{A} + \tilde{C} H_{\mathrm{new}}\big]\cdot \big[ e^{\Gamma}-\left(1- q\left(1- \left(1-p_{\mathrm{gen}} \right)^n\right)\right)(1-p_{\mathrm{con}})\big]} \\ &= \frac{q(1-p_{\mathrm{con}})(\tilde{b} - \tilde{d}H_{\mathrm{new}}) + H_{\mathrm{new}} \left(p_{\mathrm{con}}+q\left(1- \left(1-p_{\mathrm{gen}} \right)^n\right)(1-p_{\mathrm{con}})\right)}{q(1-p_{\mathrm{con}})(\tilde{c}H_{\mathrm{new}} -\tilde{a}) + e^{\Gamma}-\left(1- q\left(1- \left(1-p_{\mathrm{gen}} \right)^n\right)\right)(1-p_{\mathrm{con}})} 
     % \\ &= \frac{q(1-p_{\mathrm{con}})(\tilde{b} - \tilde{d}H_{\mathrm{new}}) + H_{\mathrm{new}} \left(p_{\mathrm{con}}+q p_{\mathrm{gen}}^*(1-p_{\mathrm{con}})\right)}{q(1-p_{\mathrm{con}})(\tilde{c}H_{\mathrm{new}} -\tilde{a}) + e^{\Gamma}-\left(1- q p_{\mathrm{gen}}^*\right)(1-p_{\mathrm{con}})}, \\ &= \frac{ \left[p_{\mathrm{con}} + q (1-p_{\mathrm{con}})\left(p_{\mathrm{gen}}^* - \tilde{d} \right) \right] \cdot H_{\mathrm{new}} + q (1-p_{\mathrm{con}})\tilde{b}}{ \left[q (1-p_{\mathrm{con}}) \tilde{c} \right] \cdot H_{\mathrm{new}} + e^{\Gamma} - 1 + p_{\mathrm{con}} + q (1-p_{\mathrm{con}})\left(p_{\mathrm{gen}}^* - \tilde{a}\right)  }
\end{align*}
where we have used the formulae (\ref{eq.def-Atilde-Btilde}) and (\ref{eq.def-Ctilde-Dtilde}) for $\tilde{A}$, $\tilde{B}$, $\tilde{C}$, and $\tilde{D}$. The above may be rewritten as
\begin{align}\label{eq.app.G.oldform}
    \overline{H} &\stackrel{\text{a.s.}}{=} \frac{q(1-p_{\mathrm{con}})(\tilde{b} - \tilde{d}H_{\mathrm{new}}) + H_{\mathrm{new}} \left(p_{\mathrm{con}}+q p_{\mathrm{gen}}^*(1-p_{\mathrm{con}})\right)}{q(1-p_{\mathrm{con}})(\tilde{c}H_{\mathrm{new}} -\tilde{a}) + e^{\Gamma}-\left(1- q p_{\mathrm{gen}}^*\right)(1-p_{\mathrm{con}})}\\ &= \frac{ \left[p_{\mathrm{con}} + q (1-p_{\mathrm{con}})\left(p_{\mathrm{gen}}^* - \tilde{d} \right) \right] \cdot H_{\mathrm{new}} + q (1-p_{\mathrm{con}})\tilde{b}}{ \left[q (1-p_{\mathrm{con}}) \tilde{c} \right] \cdot H_{\mathrm{new}} + e^{\Gamma} - 1 + p_{\mathrm{con}} + q (1-p_{\mathrm{con}})\left(p_{\mathrm{gen}}^* - \tilde{a}\right) },
\end{align}
where we have let $p_{\mathrm{gen}}^* = 1 - (1-p_{\mathrm{gen}})^n$ be the effective probability of link generation.
We now convert the above to $\overline{F}$. Recalling that $H_{\mathrm{new}} = F_{\mathrm{new}}-1/4$, it follows that
\begin{align*}
    \overline{F} &= \overline{H} + \frac{1}{4} \\ &\stackrel{\text{a.s.}}{=}\frac{ \left[p_{\mathrm{con}} + q (1-p_{\mathrm{con}})\left(p_{\mathrm{gen}}^* - \tilde{d} \right) \right] \cdot (F_{\mathrm{new}}-\frac{1}{4}) + q (1-p_{\mathrm{con}})\tilde{b}}{ \left[q (1-p_{\mathrm{con}}) \tilde{c} \right] \cdot (F_{\mathrm{new}}-\frac{1}{4}) + e^{\Gamma} - 1 + p_{\mathrm{con}} + q (1-p_{\mathrm{con}})\left(p_{\mathrm{gen}}^* - \tilde{a}\right) } + \frac{1}{4} \\ &= \frac{\left[ p_{\mathrm{con}} + q (1-p_{\mathrm{con}})\left(p_{\mathrm{gen}}^* +\frac{\tilde{c}}{4} - \tilde{d}  \right) 
 \right]\cdot F_{\mathrm{new}} + \frac{1}{4}\left[ e^{\Gamma}-1 + q(1-p_{\mathrm{con}}) \left( -\tilde{a} + 4\tilde{b} - \frac{\tilde{c}}{4} +\tilde{d} \right)\right]  }{\left[q (1-p_{\mathrm{con}})\tilde{c}\right]\cdot F_{\mathrm{new}} + e^{\Gamma} - 1 + p_{\mathrm{con}} + q (1-p_{\mathrm{con}})\left(p_{\mathrm{gen}}^* - \tilde{a} - \frac{\tilde{c}}{4}\right)},
\end{align*}
which is our formula for the average consumed fidelity in terms of the system parameters, as is given in Theorem~\ref{theorem.F}.
\end{proof}

%%%%%%%%%%%%%%%%%%%%%%%%%%%%%%%%%%%%%
%%%%%%%%%%%%%% abcd %%%%%%%%%%%%%%
%%%%%%%%%%%%%%%%%%%%%%%%%%%%%%%%%%%%%
\clearpage
\section{Purification coefficients $a_k$, $b_k$, $c_k$, and $d_k$}\label{app.abcd}
Here, we discuss the values that the coefficients $a_k$, $b_k$, $c_k$, and $d_k$ of a purification protocol $k$ are allowed to take.
Note that these coefficients are in general functions of the newly generated state, $\rho_\mathrm{new}$, although here we do not write this dependence explicitly for brevity.
Then, in Subsection~\ref{app.subsec.DEJMPS}, we provide explicit expressions for the coefficients of the DEJMPS policy discussed in the main text.

The probability of success of the protocol is given by
\begin{equation}
    p_k(F) = c_k \left(F-\frac{1}{4}\right) + d_k,
\end{equation}
where the fidelity of the buffered state $F$ can take values between $1/4$ (fully depolarized state) and $1$ (perfect Bell pair).
Since $p_k$ is a probability, we must enforce $0 \leq p_k \leq 1$. At $F=1/4$, this yields
\begin{equation}\label{eq.bounds.d}
    0 \leq d_k \leq 1.
\end{equation}
At $F=1$, it yields
\begin{equation}\label{eq.bounds.c}
    -\frac{4}{3}d_k \leq c_k \leq \frac{4}{3}(1-d_k).
\end{equation}
Combining (\ref{eq.bounds.d}) and (\ref{eq.bounds.c}) yields
\begin{equation}\label{eq.bounds.c.abs}
    -\frac{4}{3} \leq c_k \leq \frac{4}{3}.
\end{equation}

The jump function (output fidelity) of the protocol is given by
\begin{equation}
    J_k(F) = \frac{1}{4} + \frac{a_k \left(F-\frac{1}{4}\right)+b_k}{c_k \left(F-\frac{1}{4}\right)+d_k}.
\end{equation}
This output fidelity must also be between $1/4$ (fully depolarized state) and $1$ (perfect Bell pair). This condition can be written as $0 \leq a_k \left(F-\frac{1}{4}\right)+b_k \leq \frac{3}{4}p_k$.
At $F=1/4$, this yields
\begin{equation}\label{eq.bounds.b}
    0 \leq b_k \leq \frac{3}{4} d_k.
\end{equation}
Combining (\ref{eq.bounds.b}) and (\ref{eq.bounds.d}) yields
\begin{equation}\label{eq.bounds.b.abs}
    0 \leq b_k \leq \frac{3}{4}.
\end{equation}
Similarly, the condition on the jump function at $F=1$ can be written as
\begin{equation}\label{eq.bounds.a}
    -\frac{4}{3}b_k \leq a_k \leq -\frac{4}{3}b_k + \frac{3}{4}c_k + d_k.
\end{equation}
Combining (\ref{eq.bounds.a}) with (\ref{eq.bounds.d}), (\ref{eq.bounds.c}), and (\ref{eq.bounds.b}), we find
\begin{equation}\label{eq.bounds.a.abs}
    -1 \leq a_k \leq 1.
\end{equation}

\subsection{DEJMPS and concatenated DEJMPS policies}\label{app.subsec.DEJMPS}
As explained in the main text, the DEJMPS policy applies the well-known 2-to-1 DEJMPS purification protocol~\cite{Deutsch1996} to the buffered link and one of the newly generated links (and discarding the rest). This policy is given by the following purification coefficients:
\begin{equation}\label{eq.coefs-DEJMPS}
\begin{split}
    a_k &= \frac{1}{6} \left(
    5\rho_{00} + \rho_{11} + \rho_{22} - 3\rho_{33}
    \right),\\
    b_k &= \frac{1}{24} \left(
    3\rho_{00} -3 \rho_{11} -3 \rho_{22} +5 \rho_{33}
    \right),\\
    c_k &= \frac{2}{3} \left(
    \rho_{00} - \rho_{11} - \rho_{22} + \rho_{33}
    \right),\\
    d_k &= \frac{1}{2} \left(
    \rho_{00} + \rho_{11} + \rho_{22} + \rho_{33}
    \right),\\
\end{split}
\end{equation}
$\forall k \in \{1,...,n\}$, where $\rho_{ii}$ are the diagonal elements of $\rho_\mathrm{new}$ in the Bell basis $\left\{ \ket{\phi^+}, \ket{\phi^-}, \ket{\psi^+}, \ket{\psi^-} \right\}$. Note that we define the Bell states as follows:
\begin{equation}
    \ket{\phi^+} = \frac{\ket{00}+\ket{11}}{\sqrt{2}}, \;
    \ket{\phi^-} = \frac{\ket{00}-\ket{11}}{\sqrt{2}}, \;
    \ket{\psi^+} = \frac{\ket{01}+\ket{10}}{\sqrt{2}}, \;
    \ket{\psi^-} = \frac{\ket{01}-\ket{10}}{\sqrt{2}}.
\end{equation}

Regarding a concatenated or a nested DEJMPS policy, one can find the purification coefficients by applying (\ref{eq.coefs-DEJMPS}) recursively. For each round of DEJMPS, the coefficients $\rho_{ii}$ in (\ref{eq.coefs-DEJMPS}) are the diagonal elements of the output state from the previous application of DEJMPS. These diagonal elements are given by~\cite{Deutsch1996}
\begin{equation}
    \begin{split}
        \rho_{00} = \frac{\sigma_{00} \sigma'_{00} + \sigma_{33}\sigma'_{33}}{P},\\
        \rho_{11} = \frac{\sigma_{00} \sigma'_{33} + \sigma_{33}\sigma'_{00}}{P},\\
        \rho_{22} = \frac{\sigma_{11} \sigma'_{11} + \sigma_{22}\sigma'_{22}}{P},\\
        \rho_{33} = \frac{\sigma_{11} \sigma'_{22} + \sigma_{11}\sigma'_{22}}{P},
    \end{split}
\end{equation}
with $P=(\sigma_{00}+\sigma_{33})(\sigma'_{00}+\sigma'_{33}) + (\sigma_{11}+\sigma_{22})(\sigma'_{11}+\sigma'_{22})$,
where $\sigma_{ii}$ and $\sigma'_{ii}$ are the Bell diagonal elements of the two input states, $\sigma$ and $\sigma'$.

\subsection{Optimal bilocal Clifford policy}\label{app.subsec.optimalbiCliff}
In the main text, we compare the concatenated versions of the DEJMPS policy to the optimal bilocal Clifford (optimal-bC) policy. When there is a buffered link in memory and $k$ new links are generated, the optimal-bC policy operates as follows:
\begin{itemize}
    \item When $k=1$, DEJMPS is applied, using the buffered link and the newly generated link as inputs.
    \item When $k>1$, the optimal $k$-to-1 purification protocol from ref.~\cite{Jansen2022} is applied to all $k$ new links. Then, the resulting state is used for DEJMPS, together with the buffered link.
    This is illustrated in Fig.~\ref{fig.optimal-bC-sketch}b.
\end{itemize}

The reason why we apply an optimal bilocal Clifford protocol followed by DEJMPS is because these bilocal Clifford protocols have been shown to be optimal when the input states are identical. Hence, they ensure that the second link used in the final DEJMPS subroutine has maximum fidelity (see ref.~\cite{Jansen2022} for a comparison of the output fidelity using the optimal protocol versus concatenated DEJMPS).
This combined protocol (optimal $k$-to-1 followed by DEJMPS, Fig.~\ref{fig.optimal-bC-sketch}b) is not necessarily the $(k+1)$-to-1 protocol that yields the largest output fidelity. However, one would expect it to provide better buffering performance than a simple concatenation of DEJMPS subroutines (Fig.~\ref{fig.optimal-bC-sketch}a) -- nevertheless, in the main text we show that this intuition is incorrect.

Let us now show how to compute the purification coefficients $a_k$, $b_k$, $c_k$, and $d_k$ of the optimal-bC policy:
\begin{itemize}
    \item When $k=1$ new links are generated, the purification coefficients $a_1$, $b_1$, $c_1$, and $d_1$ are given by (\ref{eq.coefs-DEJMPS}), as in the DEJMPS policy.

    \item When $k>1$, we first apply the optimal bilocal Clifford protocol, which outputs a state $\sigma_k$, with diagonal elements in the Bell basis $\sigma_{k,ii}$. The probability of success of this subroutine is $\theta_k$. Then, the state $\sigma_k$ is used as input for a final DEJMPS subroutine. Using (\ref{eq.coefs-DEJMPS}), we obtain
    \begin{equation}
    \begin{split}
        a_k &= \frac{1}{6\theta_k} \left(
        5\sigma_{k,00} + \sigma_{k,11} + \sigma_{k,22} - 3\sigma_{k,33}
        \right),\\
        b_k &= \frac{1}{24\theta_k} \left(
        3\sigma_{k,00} -3 \sigma_{k,11} -3 \sigma_{k,22} +5 \sigma_{k,33}
        \right),\\
        c_k &= \frac{2}{3\theta_k} \left(
        \sigma_{k,00} - \sigma_{k,11} - \sigma_{k,22} + \sigma_{k,33}
        \right),\\
        d_k &= \frac{1}{2\theta_k} \left(
        \sigma_{k,00} + \sigma_{k,11} + \sigma_{k,22} + \sigma_{k,33}
        \right).
    \end{split}
    \end{equation}
\end{itemize}

In the example discussed in the main text, we consider $n=4$. We also consider the newly generated links to be Werner states~\cite{Werner1989} with fidelity $F_\mathrm{new}$, i.e.,
\begin{equation}
    \rho_\mathrm{new} = F_\mathrm{new} \ketbra{\phi^+} + \frac{1-F_\mathrm{new}}{3} \ketbra{\phi^-} + \frac{1-F_\mathrm{new}}{3} \ketbra{\psi^+} + \frac{1-F_\mathrm{new}}{3} \ketbra{\psi^-}.
\end{equation}
Under these assumptions, the values of $\sigma_{k,00}$ (fidelity of the output state from the optimal bilocal Clifford protocol) and $\theta_k$ are given explicitly in ref.~\cite{Jansen2022}:
    \begin{equation}
    \begin{split}
        \theta_2 &= \frac{8}{9} F_\mathrm{new}^2 - \frac{4}{9} F_\mathrm{new} + \frac{5}{9},\\
        \theta_3 &= \frac{32}{27} F_\mathrm{new}^3 - \frac{4}{9} F_\mathrm{new}^2 + \frac{7}{27},\\
        \theta_4 &= \frac{32}{27} F_\mathrm{new}^4 - \frac{4}{9} F_\mathrm{new}^2 + \frac{4}{27} F_\mathrm{new} + \frac{1}{9},
    \end{split}
    \end{equation}
    
    \begin{equation}
    \begin{split}
        \sigma_{2,00} &= \frac{1}{\theta_2} \cdot \left( \frac{10}{9} F_\mathrm{new}^2 - \frac{2}{9} F_\mathrm{new} + \frac{1}{9} \right),\\
        \sigma_{3,00} &= \frac{1}{\theta_3} \cdot \left( \frac{28}{27} F_\mathrm{new}^3 - \frac{1}{9} F_\mathrm{new} + \frac{2}{27} \right),\\
        \sigma_{4,00} &= \frac{1}{\theta_4} \cdot \left( \frac{8}{9} F_\mathrm{new}^4 + \frac{8}{27} F_\mathrm{new}^3 - \frac{2}{9} F_\mathrm{new}^2 + \frac{1}{27} \right),
    \end{split}
    \end{equation}
where $F_\mathrm{new}$ is the fidelity of the newly generated Werner states.
The rest of the diagonal elements of $\sigma_k$ can be found using the code provided in our repository (\href{https://github.com/AlvaroGI/buffering-1GnB}{https://github.com/AlvaroGI/buffering-1GnB}; this code is based on the methods from ref.~\cite{Jansen2022}). For $F_\mathrm{new}=0.7$, which we use in the example from the main text, we have
\begin{equation}
    \begin{cases}
        \sigma_{2,11} &= 0.20589\\
        \sigma_{2,22} &= 0.02941\\
        \sigma_{2,33} &= 0.02941
    \end{cases}
    \;\; , \;\;\;
    \begin{cases}
        \sigma_{3,11} &= 0.14287\\
        \sigma_{3,22} &= 0.03571\\
        \sigma_{3,33} &= 0.03571
    \end{cases}
    \;\; \mathrm{and} \;\;\;
    \begin{cases}
        \sigma_{4,11} &= 0.04545\\
        \sigma_{4,22} &= 0.04545\\
        \sigma_{4,33} &= 0.04545
    \end{cases}
    \;\; .
    \end{equation}

The calculations from ref.~\cite{Jansen2022} can also be used to obtain $\theta_k$ and $\sigma_k$ for $k>4$, although their methods become infeasible for $k>8$ due to the large computational cost, as discussed in their paper.

\begin{figure}[t!]
    \centering
    \includegraphics[width=0.45\linewidth]{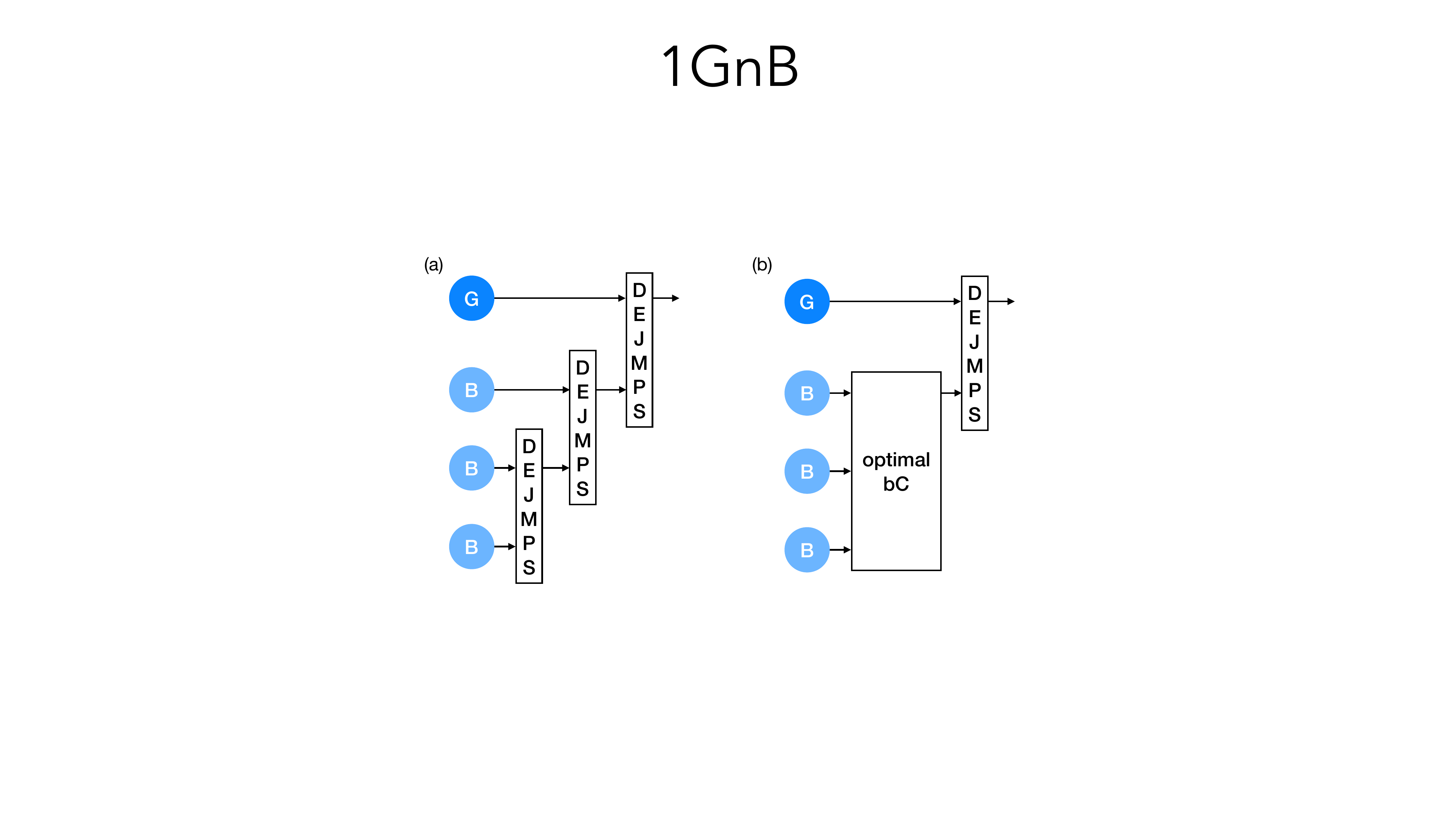}
    \caption{\textbf{The optimal bilocal Clifford policy applies an optimal protocol followed by DEJMPS.}
    Illustration of two purification policies: (a) concatenated DEJMPS and (b) optimal bilocal Clifford.
    }
    \label{fig.optimal-bC-sketch}
\end{figure}

%%%%%%%%%%%%%%%%%%%%%%%%%%%%%%%%%%%%%
%%%%%%%%%%%%%% dA/dq %%%%%%%%%%%%%%
%%%%%%%%%%%%%%%%%%%%%%%%%%%%%%%%%%%%%
\clearpage
\section{Buffering with the 513 EC policy}\label{app.subsec.513policy}
In this appendix, we compare the performance of a 1G$n$B system that uses a concatenated DEJMPS policy to a system that uses the \emph{513 EC policy}.
When there is a buffered link in memory and $k$ new links are generated, the 513 EC policy operates as follows:
\begin{itemize}
    \item When $k=1$, DEJMPS is applied, using the buffered link and the newly generated link as inputs.
    \item When $k=5$, the purification protocol based on the $[[5,1,3]]$ quantum error-correcting code~\cite{Laflamme1996} from ref.~\cite{Stephens2013} is applied to all $k$ new links. Then, the output state is twirled into a Werner state (that is, it is transformed into Werner form while preserving the fidelity) and used for DEJMPS, together with the buffered link.
    \item Otherwise, twice-concatenated DEJMPS is applied.
\end{itemize}

This policy is heavily based on twice-concatenated DEJMPS, with the main difference being that, when $k=5$, a different protocol is applied.
Note that, when $k=5$, we apply some twirling after the purification step to be able to use the results reported in ref.~\cite{Stephens2013}, where they provide the output fidelity and the success probability of the protocol but not the full density matrix of the output state.

The purification coefficients of the 513 EC policy can be computed as follows:
\begin{itemize}
    \item When $k \neq 5$, this policy applies DEJMPS or concatenated DEJMPS. Hence, $a_k$, $b_k$, $c_k$, and $d_k$ can be found as explained in Appendix~\ref{app.subsec.DEJMPS}.
    \item When $k=5$, the purification coefficients are given by the output fidelity and probability of success of the 513 EC protocol reported in Figure~3 from ref.~\cite{Stephens2013}. Since we apply this protocol followed by twirling and DEJMPS, we can use (\ref{eq.coefs-DEJMPS}) to compute the purification coefficients of the whole protocol:
    \begin{equation}
        \begin{split}
            a_5 &= \frac{1}{6\theta} \left(
            5\sigma_{00} + \sigma_{11} + \sigma_{22} - 3\sigma_{33}
            \right),\\
            b_5 &= \frac{1}{24\theta} \left(
            3\sigma_{00} -3 \sigma_{11} -3 \sigma_{22} +5 \sigma_{33}
            \right),\\
            c_5 &= \frac{2}{3\theta} \left(
            \sigma_{00} - \sigma_{11} - \sigma_{22} + \sigma_{33}
            \right),\\
            d_5 &= \frac{1}{2\theta} \left(
            \sigma_{00} + \sigma_{11} + \sigma_{22} + \sigma_{33}
            \right),
        \end{split}
    \end{equation}
    where $\sigma$ is the output state of the 513 EC protocol after twirling: $\sigma_{00}$ is the output fidelity from the 513 protocol (reported in Figure~3 from ref.~\cite{Stephens2013}), and $\sigma_{11} = \sigma_{22} = \sigma_{33} = (1-\sigma_{00})/3$ (since we twirl the output state); and $\theta$ is the probability of success of the 513 EC protocol (reported in Figure~3 from ref.~\cite{Stephens2013}).
    
\end{itemize}

Figure \ref{fig.simple_vs_513} shows the performance of the 513 EC policy versus DEJMPS and twice-concatenated DEJMPS. In this example, twice-concatenated DEJMPS also includes twirling before the final round of DEJMPS, to make the comparison with the 513 EC policy fairer.
In Fig. \ref{fig.simple_vs_513}a, we assume $F_\mathrm{new}=0.86$ (according to Figure~3 from ref.~\cite{Stephens2013}, this corresponds to $\theta = 0.869$ and $\sigma_{00}=0.864$), and in Fig. \ref{fig.simple_vs_513}a, we assume $F_\mathrm{new}=0.95$ (according to Figure~3 from ref.~\cite{Stephens2013}, this corresponds to $\theta = 0.981$ and $\sigma_{00}=0.978$).
Similar to the optimal-bC policies discussed in the main text, we observe that the 513 EC policy can be outperformed by DEJMPS, twice-concatenated DEJMPS, and replacement (Figure \ref{fig.simple_vs_513}a).
In some parameter regions, the 513 EC may provide better performance (Figure \ref{fig.simple_vs_513}b), although this behavior may not be achievable experimentally, as it requires both large $p_\mathrm{gen}$ and large $F_\mathrm{new}$ -- in commonly used entanglement generation protocols, there is a tradeoff between these two parameters, see e.g. ref.~\cite{Hermans2023}.

\begin{figure}[t!]
    \centering
    \includegraphics[width=0.9\linewidth]{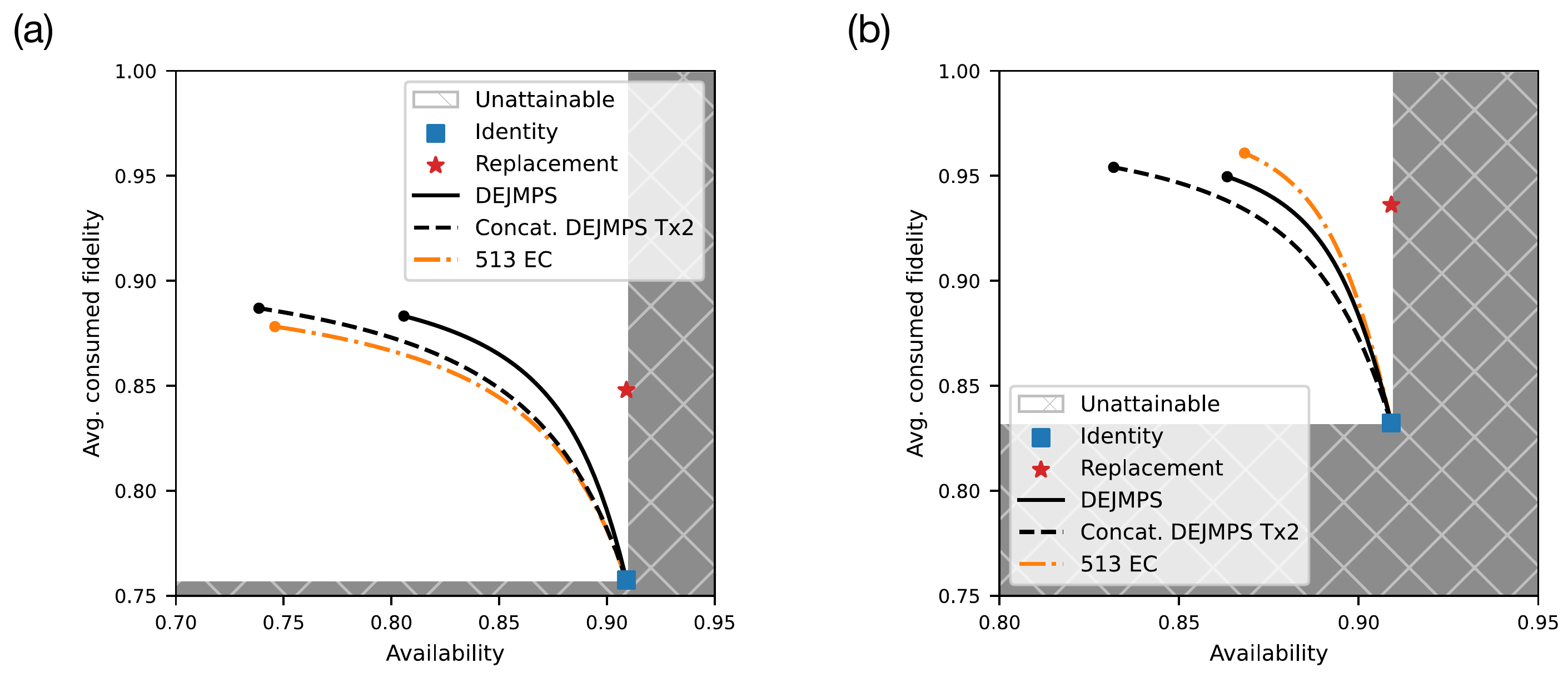}
    \caption{\textbf{The 513 EC policy may perform better than DEJMPS when new links are generated with a very large fidelity.}
    Performance of 1G$n$B systems with different purification policies, in terms of availability $A$ and average consumed fidelity $\overline F$.
    In (a), newly generated links are Werner states with fidelity $F_\mathrm{new}=0.86$, while in (b) we assume $F_\mathrm{new}=0.95$.
    The shaded area corresponds to unattainable values of $A$ and $\overline F$ (see (\ref{eq.A_bounds}) and (\ref{eq.Fbounds})).
    Replacement (star marker) and identity (square marker) policies provide maximum availability.
    Lines represent the achievable values when using one of the following policies: DEJMPS (solid line), twice-concatenated DEJMPS with twirling (dashed line), and 513 EC (dotted line).
    Parameter values used in this example: $n=5$, $p_\mathrm{gen}=1$, $p_\mathrm{con}=0.1$, and $\Gamma=0.02$.
    }
    \label{fig.simple_vs_513}
\end{figure}

%%%%%%%%%%%%%%%%%%%%%%%%%%%%%%%%%%%%%
%%%%%%%%%%%%%% dA/dq %%%%%%%%%%%%%%
%%%%%%%%%%%%%%%%%%%%%%%%%%%%%%%%%%%%%
\clearpage
\section{Monotonicity of the availability (proof of Proposition \ref{prop.dAdq}) and bounds}\label{app.dAdq}
In this appendix, we show that the availability of the 1G$n$B system (given in Theorem \ref{theorem.A}) is monotonically decreasing with increasing probability of purification $q$ (Proposition \ref{prop.dAdq}). This means that the availability is maximized when no purification is performed. If any purification is performed, the availability can only decrease, until reaching its minimum value at $q=1$. Using these ideas, we compute upper and lower bounds for the availability in Section \ref{app.sec.Abounds}.

\begin{proof}[Proof of Proposition \ref{prop.dAdq}]
We start by taking the partial derivative of the availability:
\begin{equation}\label{eq.dAdq}
    \frac{\partial A}{\partial q} = \frac{\mathbb{E}[T_\mathrm{gen}]}{\left(\mathbb{E}[T_\mathrm{gen}] + \mathbb{E}[T_\mathrm{occ}]\right)^2}
    \frac{\partial \mathbb{E}[T_\mathrm{occ}]}{\partial q},
\end{equation}
where we have used (\ref{eq.A_1GnB}), (\ref{eq.Tgen}), and (\ref{eq.Tocc}). Since the first term in \ref{eq.dAdq} is always positive, the sign of $\partial A/\partial q$ is the same as the sign of $\partial \mathbb{E}[T_\mathrm{occ}] / \partial q$. Hence, we only need to show that $\partial \mathbb{E}[T_\mathrm{occ}] / \partial q \leq 0$.
Next, we write $\mathbb{E}[T_\mathrm{occ}]$ explicitly in terms of $q$:
\begin{equation}\label{eq.ETocc-explicit}
    \mathbb{E}[T_\mathrm{occ}] = \frac{\varepsilon + \varepsilon' q}{\delta + \delta' q + \delta'' q^2},
\end{equation}
where
\begin{equation}\label{eq.eps-and-delta}
\begin{split}
    \varepsilon &\coloneqq \gamma + p_\mathrm{con},\\
    \varepsilon' &\coloneqq (1-p_\mathrm{con}) \left(p_\mathrm{gen}^*-\tilde a +  H_\mathrm{new} \tilde c\right),\\
    \delta &\coloneqq \varepsilon p_\mathrm{con},\\
    \delta' &\coloneqq (1-p_\mathrm{con}) \left( \gamma p_\mathrm{gen}^* +2p_\mathrm{con}p_\mathrm{gen}^* - p_\mathrm{con} \tilde a - (\gamma + p_\mathrm{con}) \tilde d \right),\\
    \delta'' &\coloneqq (1-p_\mathrm{con})^2 \left( (p_\mathrm{gen}^*)^2 - p_\mathrm{gen}^*\tilde a - p_\mathrm{gen}^* \tilde d + \tilde a \tilde d - \tilde b \tilde c \right),
\end{split}
\end{equation}
with $\gamma \coloneqq e^\Gamma-1$, $p_\mathrm{gen}^* \coloneqq 1 - (1-p_\mathrm{gen})^n$ and $H_\mathrm{new} \coloneqq F_\mathrm{new}-\frac{1}{4}$.
The derivative of $\mathbb{E}[T_\mathrm{occ}]$ can be written as
\begin{equation}\label{eq.dETocc}
    \frac{\partial \mathbb{E}[T_\mathrm{occ}]}{\partial q} = \frac
    {\varepsilon \left(\varepsilon' p_\mathrm{con} - \delta'\right)
    - 2\varepsilon\delta'' q
    - \varepsilon'\delta'' q^2}
    {\left( \delta + \delta' q + \delta'' q^2 \right)^2}.
\end{equation}

To prove that $\partial \mathbb{E}[T_\mathrm{occ}] / \partial q \leq 0$, we will now show that all three terms in the numerator are negative.
\\

{\sc First term from (\ref{eq.dETocc})} - The first term can be expanded as follows:
\begin{equation}\label{eq.first-term-epsilons}
\begin{split}
    \varepsilon \left(\varepsilon' p_\mathrm{con} - \delta'\right)
    &= -\varepsilon (1-p_\mathrm{con}) \left( \gamma (p_\mathrm{gen}^*-\tilde d) + p_\mathrm{con}(p_\mathrm{gen}^* -\tilde d - H_\mathrm{new} \tilde c) \right)
    \geq 0,
\end{split}
\end{equation}
where, in the last step, we have used the following: ($i$) $0 \leq p_\mathrm{con} \leq 1$, ($ii$) $\gamma \coloneqq e^\Gamma - 1 \geq 0$, ($iii$) $\tilde d + H_\mathrm{new} \tilde c \leq p_\mathrm{gen}^*$, and ($iv$) $\tilde d \leq p_\mathrm{gen}^*$.
Inequality ($iii$) can be shown as follows:
\begin{equation}\label{eq.dgnewcx}
    \begin{split}
        \tilde d + H_\mathrm{new} \tilde c
        &= \sum_{k=1}^{n} (d_k + H_\mathrm{new} c_k) {n \choose k} (1-p_\mathrm{gen})^{n-k} p_\mathrm{gen}^k\\
        &\leq \sum_{k=1}^{n} {n \choose k} (1-p_\mathrm{gen})^{n-k} p_\mathrm{gen}^k\\
        &= \sum_{k=0}^{n} {n \choose k} (1-p_\mathrm{gen})^{n-k} p_\mathrm{gen}^k
        - (1-p_\mathrm{gen})^n\\
        &= 1 - (1-p_\mathrm{gen})^n = p_\mathrm{gen}^*,
    \end{split}
\end{equation}
where we have used the definition of $\tilde c$ and $\tilde d$ from Theorem \ref{theorem.A} and the fact that $d_k + H_\mathrm{new} c_k \leq 1$ (this is the success probability of purification protocol $k$ when the link in memory has fidelity $F_\mathrm{new}$).
Inequality ($iv$) can be shown in a similar way:
\begin{equation}\label{eq.dx}
    \begin{split}
        \tilde d 
        &= \sum_{k=1}^{n} d_k {n \choose k} (1-p_\mathrm{gen})^{n-k} p_\mathrm{gen}^k\\
        &\leq \sum_{k=1}^{n} {n \choose k} (1-p_\mathrm{gen})^{n-k} p_\mathrm{gen}^k\\
        &= \sum_{k=0}^{n} {n \choose k} (1-p_\mathrm{gen})^{n-k} p_\mathrm{gen}^k
        - (1-p_\mathrm{gen})^n\\
        &= 1 - (1-p_\mathrm{gen})^n = p_\mathrm{gen}^*,
    \end{split}
\end{equation}
where we have used $d_k\leq 1$ (upper bound from (\ref{eq.bounds.d})).
\\

{\sc Second term from (\ref{eq.dETocc})} - Regarding the second term in the numerator of (\ref{eq.dETocc}), we first note that, since $p_\mathrm{con} \geq 0$ and $\gamma \geq 0$, then $\varepsilon \geq 0$. Moreover, $q\geq0$ by definition. Consequently, the second term in the numerator of (\ref{eq.dETocc}) is negative if and only if $\delta'' \geq 0$, which in turn is equivalent to $(p_\mathrm{gen}^*)^2 - p_\mathrm{gen}^* \tilde a - p_\mathrm{gen}^* \tilde d + \tilde a \tilde d - \tilde b \tilde c \geq 0$. This can be shown as follows:
\begin{equation}\label{eq.deltaaa-nonnegative}
    \begin{split}
        (p_\mathrm{gen}^*)^2 - p_\mathrm{gen}^* \tilde a - p_\mathrm{gen}^* \tilde d + \tilde a \tilde d - \tilde b \tilde c
        &\stackrel{\mathrm{i}}{\geq} (p_\mathrm{gen}^*)^2 - p_\mathrm{gen}^* \tilde a - p_\mathrm{gen}^* \tilde d + \tilde a \tilde d - \tilde b \frac{4}{3}(p_\mathrm{gen}^*-\tilde d)\\
        &= \left( p_\mathrm{gen}^* - \frac{4}{3} \tilde b - \tilde a \right) (p_\mathrm{gen}^* - \tilde d)\\
        &\stackrel{\mathrm{ii}}{\geq} 0,
    \end{split}
\end{equation}
with these steps:
\begin{enumerate}[label=\roman*.]
    \item We use $\tilde b \geq 0$ (which follows from the lower bound in (\ref{eq.bounds.b.abs})) and $\tilde c \leq ({p_\mathrm{gen}^*}-\tilde d)/H_\mathrm{new}$ (shown in (\ref{eq.dgnewcx})). This last inequality must hold for any $H_\mathrm{new} \in [0, 3/4]$, and therefore $\tilde c \leq 4({p_\mathrm{gen}^*}-\tilde d)/3$.
    \item To show that the first factor is non-negative, we use $\tilde a + 4\tilde b /3 \leq \tilde d + H_\mathrm{new} \tilde c \leq {p_\mathrm{gen}^*}$. The first inequality can be shown using the definitions of $\tilde a$, $\tilde b$, $\tilde c$, and $\tilde d$ from Theorem~\ref{theorem.A} and the upper bound from (\ref{eq.bounds.a}); while the second inequality was shown in (\ref{eq.dgnewcx}).
    The second factor (${p_\mathrm{gen}^*}-\tilde d$) is also non-negative, as shown in (\ref{eq.dx}).
\end{enumerate}

{\sc Third term from (\ref{eq.dETocc})} - Lastly, the third term in the numerator of (\ref{eq.dETocc}) is negative if and only if $\varepsilon'\geq0$, since we just showed that $\delta''\geq0$.
Moreover, $\varepsilon'\geq0 \Leftrightarrow {p_\mathrm{gen}^*}-\tilde a +  H_\mathrm{new} \tilde c \geq0$. The latter can be shown as follows:
\begin{equation}\label{eq.varepsilonn-nonnegative}
    \begin{split}
        {p_\mathrm{gen}^*} - \tilde a +  H_\mathrm{new} \tilde c
        &\stackrel{\mathrm{i}}{=} {p_\mathrm{gen}^*} + \sum_{k=1}^{n} (-a_k +  H_\mathrm{new} c_k) {n \choose k} (1-p_\mathrm{gen})^{n-k} p_\mathrm{gen}^k\\
        &\stackrel{\mathrm{ii}}{\geq} {p_\mathrm{gen}^*} + \sum_{k=1}^{n} \left(-\left(\frac{3}{4}-H_\mathrm{new}\right) c_k - d_k \right) {n \choose k} (1-p_\mathrm{gen})^{n-k} p_\mathrm{gen}^k\\
        &\stackrel{\mathrm{iii}}{\geq} {p_\mathrm{gen}^*} + \sum_{k=1}^{n} \left(\frac{4}{3} H_\mathrm{new} (1-d_k) - 1 \right) {n \choose k} (1-p_\mathrm{gen})^{n-k} p_\mathrm{gen}^k\\
        &\stackrel{\mathrm{iv}}{\geq} {p_\mathrm{gen}^*} - \sum_{k=1}^{n} {n \choose k} (1-p_\mathrm{gen})^{n-k} p_\mathrm{gen}^k\\
        &\stackrel{\mathrm{v}}{=} 0,
    \end{split}
\end{equation}
with these steps:
\begin{enumerate}[label=\roman*.]
    \item We use the definitions of $\tilde a$ and $\tilde c$ from Theorem \ref{theorem.A}.
    \item We use $a_k \leq 3c_k/4 + d_k$, which can be shown using the upper bound from (\ref{eq.bounds.a}) in combination with the lower bound from (\ref{eq.bounds.b.abs}).
    \item We use $c_k \leq 4(1-d_k)/3$ (upper bound from (\ref{eq.bounds.c})).
    \item We note that $H_\mathrm{new}(1-d_k)\geq0$, since $H_\mathrm{new}\geq 0$ (by definition) and $d_k\leq1$ (as shown in (\ref{eq.bounds.d})).
    \item We recall the definition ${p_\mathrm{gen}^*} \coloneqq 1 - (1-p_\mathrm{gen})^n = \sum_{k=0}^{n} {n \choose k} (1-p_\mathrm{gen})^{n-k} p_\mathrm{gen}^k - (1-p_\mathrm{gen})^n$.
\end{enumerate}

We have now shown that all three terms in the numerator of (\ref{eq.dETocc}) are negative. Therefore, $\partial \mathbb{E}[T_\mathrm{occ}] / \partial q \leq 0$ and, consequently, $\partial A / \partial q \leq 0$.
\end{proof}

\subsection{Upper and lower bounding the availability}\label{app.sec.Abounds}
Since $\partial A/\partial q \leq 0$, the availability is upper bounded by the value it takes when $q=0$.
From (\ref{eq.ETocc-explicit}), we have
\begin{equation}
    \mathbb{E}[T_\mathrm{occ}]\bigr|_{q=0} = \frac{1}{p_\mathrm{con}}.
\end{equation}
Combining this with (\ref{eq.A_1GnB}), we obtain
\begin{equation}\label{eq.A_lowerbound}
    A \leq A\bigr|_{q=0} = \frac{{p_\mathrm{gen}^*}}{{p_\mathrm{gen}^*}+p_\mathrm{con}},
\end{equation}
with ${p_\mathrm{gen}^*}\coloneqq 1 - (1-p_\mathrm{gen})^n$.

To evaluate $A$ at $q=1$, we first use (\ref{eq.A_1GnB}) and (\ref{eq.ETocc-explicit}) to write it as follows:
\begin{equation}\label{eq.Aq1}
    A \geq A\bigr|_{q=1} = \frac{{p_\mathrm{gen}^*} \eta}{{p_\mathrm{gen}^*}\eta + \Delta},
\end{equation}
with $\eta \coloneqq \varepsilon + \varepsilon'$, $\Delta \coloneqq \delta + \delta' + \delta''$, with $\varepsilon,\varepsilon',\delta,\delta',\delta''$ defined in (\ref{eq.eps-and-delta}).

The solution from (\ref{eq.Aq1}) constitutes a lower bound for the availability. However, $\eta$ and $\Delta$ implicitly depend on the parameters of the purification policy, $a_k$, $b_k$, $c_k$, and $d_k$, $k \in \{0,...,n\}$.
Next, we find a more general and meaningful lower bound that applies to any purification policy.

We start by noting that
\begin{itemize}
    \item $\varepsilon\geq0$ (since $p_\mathrm{con} \geq 0$ and $\gamma \geq 0$),
    
    \item $\varepsilon'\geq0$ (as shown in (\ref{eq.varepsilonn-nonnegative})),
    
    \item $\delta\geq0$ (since $\varepsilon\geq0$),
    
    \item $\delta'\geq0$ (this can be shown using the fact that $\tilde d \leq {p_\mathrm{gen}^*}$, shown in (\ref{eq.dx}), and $\tilde a \leq {p_\mathrm{gen}^*}$, which can be shown in a similar way as (\ref{eq.dx}) and using (\ref{eq.bounds.a.abs})),
    
    \item and $\delta''\geq0$ (as shown in (\ref{eq.deltaaa-nonnegative})).
\end{itemize}
As a consequence, none of the factors in (\ref{eq.Aq1}) can be negative: ${p_\mathrm{gen}^*}\geq 0$ (by definition), $\eta\geq 0$, and $\Delta \geq 0$. This means that we can find a lower bound for $A\bigr|_{q=1}$ by lower bounding $\eta$ and upper bounding $\Delta$.
We first lower bound~$\eta$:
\begin{equation}\label{eq.lowerb-eta}
    \eta = \gamma + p_\mathrm{con} + (1-p_\mathrm{con}) \left({p_\mathrm{gen}^*}-\tilde a +  H_\mathrm{new} \tilde c\right)
    \geq \gamma + p_\mathrm{con},
\end{equation}
where we have used ${p_\mathrm{gen}^*}-\tilde a +  H_\mathrm{new} \tilde c \geq 0$, which was shown in (\ref{eq.varepsilonn-nonnegative}).

Regarding $\Delta$, we proceed as follows:
\begin{equation}\label{eq.upperb-Delta}
    \begin{split}
        \Delta &= \delta+\delta'+\delta''\\
        &= (\gamma + p_\mathrm{con})p_\mathrm{con}
        +(1-p_\mathrm{con}) \left( (\gamma + 2p_\mathrm{con}) {p_\mathrm{gen}^*} - p_\mathrm{con} (\tilde a + \tilde d) - \gamma \tilde d \right)
        +(1-p_\mathrm{con})^2 \left( (p_\mathrm{gen}^*)^2 - {p_\mathrm{gen}^*} (\tilde a + \tilde d) + \tilde a \tilde d - \tilde b \tilde c \right)\\
        &\stackrel{\mathrm{i}}{\leq} (\gamma + p_\mathrm{con})p_\mathrm{con}
        +(1-p_\mathrm{con}) \left( \gamma {p_\mathrm{gen}^*} +2p_\mathrm{con}{p_\mathrm{gen}^*} - p_\mathrm{con} \tilde a \right)
        +(1-p_\mathrm{con})^2 \left( (p_\mathrm{gen}^*)^2 - {p_\mathrm{gen}^*}\tilde a + \tilde a \tilde d - \tilde b \tilde c \right)\\
        &\stackrel{\mathrm{ii}}{\leq} (\gamma + p_\mathrm{con})p_\mathrm{con}
        +(1-p_\mathrm{con}) \left( \gamma {p_\mathrm{gen}^*} +2p_\mathrm{con}{p_\mathrm{gen}^*} + p_\mathrm{con} {p_\mathrm{gen}^*} \right)
        +(1-p_\mathrm{con})^2 \left( (p_\mathrm{gen}^*)^2 + (p_\mathrm{gen}^*)^2 + \tilde a \tilde d - \tilde b \tilde c \right)\\
        &= (\gamma + p_\mathrm{con})p_\mathrm{con}
        +(1-p_\mathrm{con}) \left( \gamma + 3p_\mathrm{con} \right){p_\mathrm{gen}^*}
        +(1-p_\mathrm{con})^2 \left( 2(p_\mathrm{gen}^*)^2 + \tilde a \tilde d - \tilde b \tilde c \right)\\
        &\stackrel{\mathrm{iii}}{\leq} (\gamma + p_\mathrm{con})p_\mathrm{con}
        +(1-p_\mathrm{con}) \left( \gamma + 3p_\mathrm{con}\right){p_\mathrm{gen}^*}
        +(1-p_\mathrm{con})^2 \left( 2(p_\mathrm{gen}^*)^2 + {p_\mathrm{gen}^*} \right)\\
        &= (\gamma + p_\mathrm{con})p_\mathrm{con}
        +(1-p_\mathrm{con}) \left(1 + \gamma + 2p_\mathrm{con}\right){p_\mathrm{gen}^*}
        + 2(1-p_\mathrm{con})^2 (p_\mathrm{gen}^*)^2
    \end{split}
\end{equation}
with these steps:
\begin{enumerate}[label=\roman*.]
    \item We use $- (\gamma + p_\mathrm{con}) \tilde d \leq 0$ and
    $- {p_\mathrm{gen}^*} \tilde d \leq 0$, which follows from $\tilde d\geq 0$ (shown in (\ref{eq.bounds.d})).
    \item Using the lower bound from (\ref{eq.bounds.a.abs}) and following a similar derivation as in (\ref{eq.dx}), one can show that $\tilde a \geq -{p_\mathrm{gen}^*}$. This implies that $- p_\mathrm{con} \tilde a \leq p_\mathrm{con} {p_\mathrm{gen}^*}$ and $- {p_\mathrm{gen}^*}\tilde a \leq (p_\mathrm{gen}^*)^2$.
    \item We use $\tilde a \tilde d - \tilde b \tilde c \leq {p_\mathrm{gen}^*}$. This can be shown as follows:
        \begin{equation}
            \begin{split}
                \tilde a \tilde d - \tilde b \tilde c
                \leq -\frac{4}{3}\tilde b \tilde d + \frac{3}{4} \tilde c \tilde d + \tilde d ^2 - \tilde b \tilde c
                \leq -\frac{4}{3}\tilde b \tilde d + (1-\tilde d) \tilde d + \tilde d ^2 - \tilde b \tilde c
                \leq \tilde d
                \leq {p_\mathrm{gen}^*},
            \end{split}
        \end{equation}
    where we have used the upper bound from (\ref{eq.bounds.a}) in the first step; $\tilde d \geq 0$ (see (\ref{eq.bounds.d})) and the upper bound from (\ref{eq.bounds.c}) in the second step; $\tilde b \geq 0$ (see (\ref{eq.bounds.b.abs})) and the lower bound from (\ref{eq.bounds.c}) in the third step; and (\ref{eq.dx}) in the last step.
\end{enumerate}

Lastly, combining (\ref{eq.Aq1}) with the bounds from (\ref{eq.lowerb-eta}) and (\ref{eq.upperb-Delta}), we obtain
\begin{equation}\label{eq.Aq1bound}
    A \geq A\bigr|_{q=1} = \frac{{p_\mathrm{gen}^*} \eta}{{p_\mathrm{gen}^*}\eta + \Delta} \geq
    \frac{{p_\mathrm{gen}^*} \left(\gamma + p_\mathrm{con}\right)}{\xi + \xi' {p_\mathrm{gen}^*} + \xi'' (p_\mathrm{gen}^*)^2},
\end{equation}
with $\xi \coloneqq \gamma p_\mathrm{con} + p_\mathrm{con}^2$, $\xi' \coloneqq 1 + 2\gamma + (2-\gamma) p_\mathrm{con} - 2 p_\mathrm{con}^2 $, and $\xi'' \coloneqq 2 (1-p_\mathrm{con})^2$.
This lower bound is general and applies to every 1G$n$B system, no matter which purification policy it employs.

%%%%%%%%%%%%%%%%%%%%%%%%%%%%%%%%%%%%%
%%%%%%%%%%%%%% dF/dq %%%%%%%%%%%%%%
%%%%%%%%%%%%%%%%%%%%%%%%%%%%%%%%%%%%%
\clearpage
\section{Monotonicity of the average consumed fidelity (proof of Proposition \ref{prop.dFdq}) and bounds}\label{app.dFdq}
In this appendix, we show that the average consumed fidelity of the 1G$n$B system (given in Theorem \ref{theorem.F}) is monotonically increasing with increasing probability of purification $q$ (Proposition \ref{prop.dFdq}), as long as the purification policy is made of protocols that can increase the fidelity of newly generated links (i.e., $J_k(F_\mathrm{new}) \geq F_\mathrm{new}$, $\forall k \in \{1, ..., n\}$). This means that the average consumed fidelity is maximized when purification is performed every time a new link is generated ($q=1$). Using these ideas, we compute upper and lower bounds for the average consumed fidelity in Section \ref{app.sec.Fbounds}.

\begin{proof}[Proof of Proposition \ref{prop.dFdq}]
Recalling from (\ref{eq.Fbar-Hbar-relation}) that $\overline{F} = \overline{H} + 1/4$, showing the monotonicity of $\overline{H}$ is equivalent to showing the monotonicity of $\overline{F}$. We firstly rewrite the formula for $\overline{H}$ as given in (\ref{eq.app.G.oldform}),
\begin{equation*}
    \overline{H} = 
    \frac{ q (1-p_{\mathrm{con}})\left[ \tilde{b} - \tilde{d}H_{\mathrm{new}} + H_{\mathrm{new}}p_{\mathrm{gen}}^*  \right] + H_{\mathrm{new}} p_{\mathrm{con}} }{ q (1-p_{\mathrm{con}}) \left[ \tilde{c}H_{\mathrm{new}} -\tilde{a} + p_{\mathrm{gen}}^* \right] + e^{\Gamma}-1+p_{\mathrm{con}}},
    \end{equation*}
    where $p_{\mathrm{gen}}^* = 1 - (1-p_{\mathrm{gen}})^n$.
    Now consider functions of the form $g(x) = \frac{\alpha x + \beta}{\gamma x + \delta}$. This is non-decreasing if and only if 
    \begin{align*}
        \dv{g}{x} = \frac{\alpha \delta - \beta \gamma }{(\gamma x + \delta )^2} &\geq 0 \\ \Leftrightarrow \alpha \delta - \beta \gamma &\geq 0.
    \end{align*}
    We therefore see that $\overline{H}$ is non-decreasing in $q$ if and only if
    \begin{equation*}
        (1-p_{\mathrm{con}})\left[ \tilde{b} - \tilde{d}H_{\mathrm{new}} + H_{\mathrm{new}}p_{\mathrm{gen}}^*  \right](e^{\Gamma}-1+p_{\mathrm{con}})  - H_{\mathrm{new}} p_{\mathrm{con}} (1-p_{\mathrm{con}}) \left[ \tilde{c}H_{\mathrm{new}} -\tilde{a} + p_{\mathrm{gen}}^* \right] \geq 0, 
    \end{equation*}
    or equivalently
    \begin{equation}
        (e^{\Gamma}-1)\left(\tilde{b} - \tilde{d}H_{\mathrm{new}} + H_{\mathrm{new}} p_{\mathrm{gen}}^* \right) + p_{\mathrm{con}}\left(\tilde{b} - \tilde{d}H_{\mathrm{new}} - \tilde{c} H_{\mathrm{new}}^2 + \tilde{a}H_{\mathrm{new}} \right) \geq 0 
        \label{eqn:jump_function_non-decreasing_condition}
    \end{equation}
    We now show this by considering the two parts of the expression:
    \begin{enumerate}[label=(\alph*)]
        \item $\tilde{b} - \tilde{d}H_{\mathrm{new}} + H_{\mathrm{new}} p_{\mathrm{gen}}^* \geq 0$ \newline Recall that the jump functions $\tilde{J}_k$ map the shifted fidelity $h$ as
        \begin{equation}
            \tilde{J}_k(h) = \frac{a_k h+b_k}{c_k h+d_k}.
        \end{equation}
        When the input state is completely mixed ($h=0$), the probability of successful purification is 
        \begin{align*}
            \tilde{p}_k(0) = d_k,
        \end{align*}
        and so we must have $0\leq  d_k \leq 1$. If $d_k >0$, the output fidelity when inputting a completely mixed state then satisfies 
        \begin{align*}
            \tilde{J}_k(0) &=  \frac{b_k}{d_k} \geq 0
        \end{align*}
        which implies $b\geq 0$. If $d=0$, the output fidelity as the input state approaches the completely mixed state is 
        \begin{equation*}
            \lim_{h\rightarrow 0} \frac{a_k h+b_k}{c_k h},
        \end{equation*}
        and since this is bounded, it must be the case that $b=0$. Therefore,
        \begin{equation*}
            \tilde{b} = \sum_{k=1}^n b_k \cdot \binom{n}{k}(1-p_{\mathrm{gen}})^{n-k}p_{\mathrm{gen}}^k \geq 0
        \end{equation*}
        and 
        \begin{align*}
            \tilde{d} &= \sum_{k=1}^n d_k \cdot \binom{n}{k}(1-p_{\mathrm{gen}})^{n-k}p_{\mathrm{gen}}^k \\ &\leq \sum_{k=1}^n \binom{n}{k}(1-p_{\mathrm{gen}})^{n-k}p_{\mathrm{gen}}^k = 1- (1-p_{\mathrm{gen}})^n = p_{\mathrm{gen}}^*.
        \end{align*}
        Combining the above, we obtain 
        \begin{align*}
            \tilde{b} - \tilde{d}H_{\mathrm{new}} + H_{\mathrm{new}} p_{\mathrm{gen}}^* \geq H_{\mathrm{new}}(p_{\mathrm{gen}}^* - \tilde{d}) \geq 0.
        \end{align*}
        \item $\tilde{b} - \tilde{d}H_{\mathrm{new}} - \tilde{c} H_{\mathrm{new}}^2 + \tilde{a}H_{\mathrm{new}} \geq 0 $
        \newline We have that
        \begin{align*}
            \tilde{b} - \tilde{d}H_{\mathrm{new}} - \tilde{c} H_{\mathrm{new}}^2 + \tilde{a}H_{\mathrm{new}} &=  \sum_{k=1}^n  \binom{n}{k}(1-p_{\mathrm{gen}})^{n-k}p_{\mathrm{gen}}^k (b_k - d_k H_{\mathrm{new}} - c_k H_{\mathrm{new}}^2 +a_k H_{\mathrm{new}} ) \\ &= \sum_{k=1}^n \binom{n}{k}(1-p_{\mathrm{gen}})^{n-k}p_{\mathrm{gen}}^k \left(\frac{a_k H_{\mathrm{new}}+b_k}{c_k H_{\mathrm{new}} +d_k}
            - H_{\mathrm{new}} \right) \left( c_k H_{\mathrm{new}} +d_k\right) \\ &= \sum_{k=1}^n  \binom{n}{k}(1-p_{\mathrm{gen}})^{n-k}p_{\mathrm{gen}}^k \cdot  \left( \tilde{J}_k(H_{\mathrm{new}}) -H_{\mathrm{new}}\right) \cdot \tilde{p}_k(H_{\mathrm{new}}),
        \end{align*}
        which is non-negative if all jump functions $\tilde{J}_k$ satisfy
        \begin{equation*}
            \tilde{J}_k(H_{\mathrm{new}})\geq H_{\mathrm{new}},
        \end{equation*}
    \end{enumerate}
    or equivalently $J_k(F_{\mathrm{new}})\geq F_{\mathrm{new}}$
    Since $\Gamma \geq 0$, we therefore see that (\ref{eqn:jump_function_non-decreasing_condition}) holds. 
\end{proof}

\subsection{Upper and lower bounding the average consumed fidelity}\label{app.sec.Fbounds}
Here, we only consider purification policies made of protocols that can increase the fidelity of newly generated links (i.e., $J_k(F_\mathrm{new}) \geq F_\mathrm{new}$, $\forall k \in \{1, ..., n\}$).
For these policies, $\partial \overline F/\partial q \geq 0$.
A tight lower bound can be found by setting $q=0$ in (\ref{eq.F}):
\begin{equation}
    \overline F \geq \overline F\bigr|_{q=0} = \frac{\gamma/4 + F_\mathrm{new} p_\mathrm{con}}{\gamma + p_\mathrm{con}},
\end{equation}
where $\gamma \coloneqq e^{\Gamma} -1$.

An upper bound for $\overline F$ can be found by upper bounding it maximum value, which occurs at $q=1$. Using (\ref{eq.app.G.oldform}), we can write the maximum value as
\begin{equation}\label{eq.app.F1}
    \overline{F}\bigr|_{q=1} = \frac{1}{4} + \frac{(1-p_{\mathrm{con}})(\tilde{b} - \tilde{d}H_\mathrm{new}) + H_\mathrm{new} \left(p_{\mathrm{con}}+ p_{\mathrm{gen}}^*(1-p_{\mathrm{con}})\right)}{(1-p_{\mathrm{con}})(\tilde{c}H_\mathrm{new} -\tilde{a}) + \gamma + p_\mathrm{con} + (1-p_{\mathrm{con}}) p_{\mathrm{gen}}^*},
\end{equation}
where $p_\mathrm{gen}^* = 1 - (1-p_\mathrm{gen})^n$.
Using (\ref{eq.bounds.b}) and (\ref{eq.dx}), it can be shown that $\tilde{b} - \tilde{d}H_\mathrm{new} \leq p_\mathrm{gen}^* (3/4-H_\mathrm{new})$. Moreover, from (\ref{eq.varepsilonn-nonnegative}), we know that $H_\mathrm{new} \tilde c - \tilde a \geq -p_\mathrm{gen}^*$.
Applying these two inequalities to (\ref{eq.app.F1}), we find the upper bound:
\begin{equation}
    \overline F \leq \overline F\bigr|_{q=1} \leq
    \frac{1}{4} + \frac{H_\mathrm{new} p_\mathrm{con} + (3/4) (1-p_{\mathrm{con}}) p_{\mathrm{gen}}^*
    }
    {\gamma + p_\mathrm{con}}
    = \frac{\gamma/4 + F_\mathrm{new} p_\mathrm{con}}{\gamma + p_\mathrm{con}}
    + \frac{3}{4}\frac{(1-p_\mathrm{con})p_\mathrm{gen}^*}{\gamma + p_\mathrm{con}}.
\end{equation}

\clearpage
\section{Entanglement buffering with concatenated purification}\label{app.concatenation_extra}
In this appendix, we discuss further features of 1G$n$B buffers that use concatenated purification policies.
In \ref{app.subsec.orderings}, we consider different orderings for the purification subroutines that are being concatenated.
In \ref{app.subsec.increasingconcats}, we show that increasing the number of concatenations is beneficial when noise in memory is very strong.

\subsection{Different concatenation orderings}\label{app.subsec.orderings}
As stated in the main text, we tested different orderings of the concatenated purification subroutines.
In Figure~\ref{fig.concatenation-sketch}, we showed two different orderings for a concatenated DEJMPS policy: sequentially concatenated DEJMPS and nested DEJMPS.
Here, we consider a policy that applies a nested DEJMPS protocol to all the newly generated links, and then uses the output state to purify the link in memory with a final round of DEJMPS.
This policy is only defined when the number of links generated is a power of 2. Hence, we assume $n=4$ bad memories and deterministic entanglement generation ($p_\mathrm{gen}=1$) in the following example.
Figure~\ref{fig.concat-vs-nested} shows the performance of this policy compared to concatenated versions of DEJMPS (in which DEJMPS is applied sequentially to all links, as shown in Figure~\ref{fig.concatenation-sketch}a).
The performance of all policies shown is qualitatively similar. We also observe that, in this case, nesting is better than concatenating as much as possible, but it is worse than concatenating twice.

\begin{figure}[th]
    \centering
    \includegraphics[width=0.5\linewidth]{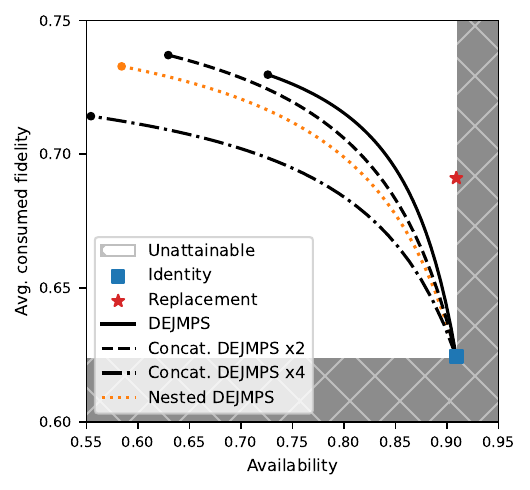}
    \caption{\textbf{Different concatenation orderings seem to yield qualitatively similar performance.}
    Performance of 1G$n$B systems with different purification policies, in terms of availability $A$ and average consumed fidelity $\overline F$. The shaded area corresponds to unattainable values of $A$ and $\overline F$ (see (\ref{eq.A_bounds}) and (\ref{eq.Fbounds})).
    Lines and markers show the combinations of $A$ and $\overline F$ achievable by different purification policies: identity (square marker), replacement (star marker),
    DEJMPS (solid line), twice-concatenated DEJMPS (dashed line), thrice-concatenated DEJMPS (dotted-dashed line), and nested DEJMPS (orange dotted line).
    Parameter values used in this example: $n=4$, $p_\mathrm{gen}=1$, $F_\mathrm{new}=0.7$ ($\rho_\mathrm{new}$ is a Werner state), $p_\mathrm{con}=0.1$, and $\Gamma=0.02$.
    }
    \label{fig.concat-vs-nested}
\end{figure}

\subsection{Increasing number of concatenations}\label{app.subsec.increasingconcats}
In the main text, we showed that using some newly generated entangled links in the purification protocol and discarding the rest may provide a better buffering performance than implementing a more complex protocol that uses all the newly generated links.
In particular, we showed that increasing the maximum number of concatenations in a concatenated DEJMPS policy does not necessarily lead to better performance.
The reason was that, as we increase the number of concatenations, the overall probability of success of the protocol decreases.
Nevertheless, this effect is irrelevant when noise is strong: the quality of the buffered entanglement decays so rapidly that we need a protocol that can compensate noise with large boosts in fidelity, even if the probability of failure is large.
This is shown in Figure~\ref{fig.goodconcatenations}, where we display the maximum average consumed fidelity (i.e. assuming purification probability $q=1$, see Proposition~\ref{prop.dFdq}) versus the number of concatenations. When no purification is applied (zero concatenations), $\overline{F}$ is below 0.5, meaning that the good memory stores no entanglement, on average (see ref.~\cite{Davies2023a}).
As we increase the number of concatenations in the purification protocol, $\overline{F}$ increases, although the increase is marginal.
Note that this is a consequence of the strong noise experienced by the buffered entanglement -- in Figure~\ref{fig.excessive_concatenation} we showed the same plot but considering a lower noise level and the conclusions were different: increasing the number of concatenations eventually led to a decrease in average consumed fidelity.

\begin{figure}[th]
    \centering
    \includegraphics[width=0.5\linewidth]{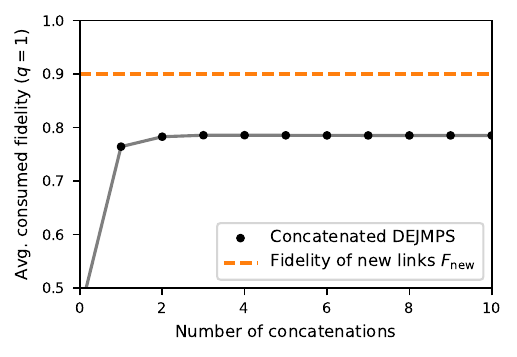}
    \caption{\textbf{Additional concatenation may improve the performance when noise is strong.}
    Maximum average consumed fidelity $\overline F$ achieved by a purification policy that concatenates DEJMPS a limited number of times.
    Zero concatenations corresponds to an identity policy (no purification is performed). One concatenation corresponds to the DEJMPS policy.
    Parameter values used in this example: $n=10$, $p_\mathrm{gen}=0.5$, $F_\mathrm{new}=0.9$ ($\rho_\mathrm{new}$ is a Werner state), $p_\mathrm{con}=0.1$, and $\Gamma=0.2$.
    }
    \label{fig.goodconcatenations}
\end{figure}

%%%%%%%%%%%%%%%%%%%%%%%%%%%%%%%%%%%%%
%%%%%%%%%%%%%%%%%%%%%%%%%%%%%%%%%%%%%
%%%%%%%%%%%%%%%%%%%%%%%%%%%%%%%%%%%%%

\end{document}